\documentclass[11pt]{article} 
\usepackage{sty/symbols,sty/alg2emod}
\usepackage{fullpage}
\usepackage[colorlinks=false, linktocpage=true, pdfborder={0 0 0}]{hyperref}
\usepackage[english]{babel}
\usepackage[square,numbers,comma,sort&compress,numbers]{natbib} 
\usepackage{amssymb,sty/amsmod}

\newcommand{\algoSize}{\smaller}
\newcommand{\remove}[1]{}
\newcommand{\bigO}{\mathcal{O}}

\usepackage{color}

\usepackage[
   pages=all,
   placement=top,
   angle=0,
   color=black!00,
   scale=2
]{background}

\begin{document}

\setcounter{footnote}{1}

\title{Self-Stabilizing and Private Distributed Shared Atomic Memory\\ in Seldomly Fair Message Passing Networks\\\large{(Technical Report)~\footnote{An earlier version of this technical report appeared as a brief announcement in~\cite{DBLP:conf/podc/DolevPS15}}}}


\author{Shlomi Dolev~\footnote{Department of Computer Science, Ben-Gurion University of the
Negev, 84105 Beer-Sheva, Israel, \texttt{dolev@cs.bgu.ac.il}}\and Thomas Petig~\footnote{Department of Computer Science and Engineering, Chalmers University of Technology, 41296 Gothenburg, Sweden, \texttt{$\{$petig,elad$\}$@chalmers.se}}\and Elad M.\ Schiller~$^{\footnotemark[3]}$}

\date{}

\maketitle

\begin{abstract}
We study the problem of privately emulating shared memory in message-passing networks. The system includes clients that store and retrieve replicated information on $N$ servers, out of which $e$ are malicious. When a client access a malicious server, the data field of that server response might be different than the value it originally stored. However, all other control variables in the server reply and protocol actions are according to the server algorithm. For the coded atomic storage (CAS) algorithms by Cadambe et al., we present an enhancement that ensures no information leakage and
malicious fault-tolerance. 

We also consider recovery after the occurrence of transient faults that violate the assumptions according to which the system is to behave. 
%
%
After their last occurrence, transient faults leave the system in an arbitrary state (while the program code stays intact). We present a self-stabilizing algorithm, which recovers after the occurrence of transient faults. This addition to Cadambe et al. considers asynchronous settings as long as no transient faults occur. The recovery from transient faults that bring the system counters (close) to their maximal values may include the use of a global reset procedure, which requires the system run to be controlled by a fair scheduler. After the recovery period, the safety properties are provided for asynchronous system runs that are not necessarily controlled by fair schedulers. 

Since the recovery period is bounded and the occurrence of transient faults is extremely rare, we call this design criteria self-stabilization in the presence of seldom fairness. Our self-stabilizing algorithm uses a bounded amount of storage during asynchronous executions (that are not necessarily controlled by fair schedulers). To the best of our knowledge, we are the first to address privacy, malicious behavior and self-stabilization in the context of emulating atomic shared memory in message-passing systems.
\end{abstract}

\section{Introduction}
\label{sec:intro}
%
%
The increasing availability of fast ubiquitous networking, the appearance of Cloud and Fog computing, have offered computer users attractive opportunities for remotely storing massive amounts of data in decentralized storage systems. In such systems, privacy and dependability are imperative. We consider distributed fault-tolerant systems that prevent information leakage, deal with malicious behavior and can recover after the occurrence of transient faults, which cause an arbitrary corruption of the system state, including the state of the mechanisms for storing information, so long as the program's code is still intact. To the best of our knowledge, we are the first to show that the emulation of atomic shared memory in message-passing systems can be done in a way that considers information privacy, resilience to malicious behavior and recovery from transient-faults.  

\subsection{The problem}
A distributed storage system uses a decentralized set of servers for allowing clients to access a shared object concurrently. Register emulation is a well-known method for sharing objects. Among the three kinds of consistency requirements for registers, atomicity is the strongest one, since it requires every sequence of concurrent access to the register to appear sequential~\cite{DBLP:journals/dc/Lamport86a}. Another classification of register emulation considers the number of clients that can read or write the shared register concurrently. We consider the more general form of shared memory emulation of an atomic register in which many clients can read and write concurrently. 

\subsubsection{Storage and communication costs}
Early approaches~\cite{Attiya:1995,DBLP:conf/ftcs/LynchS97} provided fault-tolerance for distributed emulation of shared registers via replication. That is, each server is to store an identical copy of the most recent version of the shared object. These solutions require the read procedure to include a propagation phase in which the reader updates the servers with the most recent value they read; details appear in~\cite{DBLP:conf/ftcs/LynchS97}. Since in these early approaches the interaction between the clients and the servers includes sending of the entire replica, high communication costs are implied. Recent advances in the area~\cite{DBLP:conf/wdag/FanL03,DBLP:journals/dc/CadambeLMM17} are less costly than these early approaches~\cite{Attiya:1995,DBLP:conf/ftcs/LynchS97}, because their propagation phase messages include only the control variables, rather than the entire replica. Moreover, using erasure coding, the servers avoid storing the entire replica by storing only the \textit{coded elements}, which are tailed individually to every server. This leads to further reduction in the size of messages in all phases (see further details in~\cite{DBLP:conf/ftcs/LynchS97,DBLP:conf/podc/SpiegelmanCCK16}). 

\subsubsection{Malicious behavior and privacy}
The use of erasure coding facilitates, as we show in this paper, allows the satisfaction of requirements related to malicious behavior and privacy. That is, when a client access a malicious server, the data field of that server response might be different than the value it originally stored (however, all other control variables in the response and protocol actions follows the algorithm). Our privacy requirement is that the collective storage of any set of less than $k$ servers cannot leak information, where $k$ is a number that we specify next.

\subsubsection{Problem specifications}
The system has $N \in \mathbb{Z}^+$ servers that emulate an atomic shared memory, which any bounded set of clients may access. We consider the integers $f,e \in \mathbb{Z}^+$ and $k \in \{1, \ldots, N-2(f+e)\}$. The coded atomic storage ${\text{CAS}(k)}$ task addresses the problem of multi-writer, multi-reader (MWMR) emulation of atomic shared memory of a single object $\text{CAS}(k)$'s safety requirement says that the algorithm's external behavior follows the ones of atomic memory, and $\text{CAS}(k)$'s liveness require the completion of all (non-failing) operations independently of the node availability. From the communication and storage costs, we interpret the task of ${\text{CAS}(k)}$ to restrict the messages between clients and server, as well as the storage records, to include only individualized coded elements and control variables, as in~\cite{DBLP:journals/dc/CadambeLMM17}. 


\subsection{Fault model}
Our message-passing system is asynchronous and it is prone to (a) fail-stop failures of nodes that may resume at any time, (b) packet failures, such as omission, duplication, and reordering, and (c) malicious servers can reply with a message that its data field is different than the originally stored value (however, all other control variables stay intact and in all other matters malicious servers do not deviate from the algorithm).

We bound the number of malicious servers is bounded by $e$. We assume that the number of failing servers is bounded by $f$ (for correctness sake~\cite{DBLP:books/mk/Lynch96}) and allow them to resume operation at any time. We assume that failing clients stop taking steps. Note that the client identifiers are recyclable using incarnation numbers (as we explain in Section~\ref{sec:incarnationNumber}). Thus, although we bound by {$N$} the number of clients that are concurrently active, the number of client life cycles is unbounded for any practical purpose.

In addition to these benign failures, we consider \emph{transient faults}, i.e., any temporary violation of assumptions according to which the system and network were designed to behave, e.g., the corruption of the system state due to soft errors or the presence of Byzantine nodes. 
We assume that these transient faults arbitrarily change the system state in unpredictable manners (while keeping the program code intact).

\subsection{Design criteria}
Dijkstra's seminal work~\cite{DBLP:journals/cacm/Dijkstra74} proposed the self-stabilization design criteria, which models transient faults to occur before the start of the system run, because Dijkstra considered the occurrence of transient faults to be an extremely rare event that causes the system to start in an arbitrary state. Dijkstra required self-stabilizing systems to return to correct behavior within a bounded period.
 
An unfair scheduler of an asynchronous system (with bounded memory and channel capacity) can cause the system to indefinitely hide stale information, because this adversarial scheduler does not guarantee that all nodes take steps infinitely often.
This stale information is the result of transient faults, which occur before the system start, however its presence may have a long-term effect, because this corrupted data can cause the system, at any time, to violate safety. 
This is true for any system, in particular, for Dijkstra's self-stabilizing systems~\cite{DBLP:journals/cacm/Dijkstra74}, which are required to specify a bound on the time in which they remove all stale information (whenever they appear).
Without ever restricting the degree in which the scheduler can be unfair, we cannot specify, \textit{for any given system}, when it can remove all stale information. Thus, we cannot demonstrate, for any given system, that the proposed algorithm fulfills Dijkstra's requirements.

This paper proposes to restrict the scheduler unfairness degree in the following manner: (i) In the absence of transient faults, the scheduler is unfair and the system guarantees safety and liveness. (ii) After the occurrence of the last transient fault, the scheduler becomes fair for a period that is at least as long as the specified system recovery period. That is, after the recovery period, the scheduler returns to be unfair as in (i) and the system returns to guarantee safety and liveness. The proposed design criteria, which we call \emph{self-stabilization in the presence of seldom fairness}, address challenges that Dijkstra's self-stabilization cannot address when the scheduler is always either fair or unfair. 

For example, any transient fault can cause a counter to reach its maximum value and yet the system might need to increment the counter for an unbounded number of times after that overflow event. This challenge is greater when there is no elegant way to maintain an order among the different counter values, say, by wrapping to zero upon counter overflow. Since the scheduler fair after the occurrence of transient faults and until the end of the recovery period, the system recovery can handle the overflow event by restarting the system. This work uses a restart procedure that does not violate safety (but may violate liveness after the occurrence of transient faults). This is a challenge that Dijkstra's self-stabilization cannot address when the scheduler is always unfair. 

Systems that are self-stabilization in Dijkstra's sense often assume that the scheduler is always fair. Here, fail-stop failures are modeled as transient faults, which can occur only before the system starts to run. Challenges related to fail-stop failures are addressed more adequately by the proposed design criteria than Dijkstra's self-stabilization when the scheduler is always fair.

\subsection{Related work}
\subsubsection{Non-self-stabilizing register emulation in message-passing systems}
\label{sec:noSelfABD}
The literature on (non-self-stabilizing) register emulation in message-passing
systems includes~\cite{Attiya:1995} single-writer multi-reader (SWMR), and
their multi-writer (MWMR) counterparts~\cite{DBLP:conf/ftcs/LynchS97,DBLP:conf/wdag/FanL03} as well as solutions that provides (non-self-stabilizing) quorum reconfiguration~\cite{DBLP:journals/dc/GilbertLS10,DBLP:journals/jpdc/GeorgiouNS09}. 

\subparagraph{Recent advances to the state-of-the-art.~~}
As of the time of our publication~\cite{DBLP:conf/podc/DolevPS15}, the literature considered either (i) unbounded storage during asynchronous system runs that are not controlled by a fair scheduler, such as CASGC~\cite{DBLP:journals/dc/CadambeLMM17}, AWE~\cite{DBLP:conf/opodis/AndroulakiCDV14}, HGR~\cite{DBLP:conf/sosp/HendricksGR07} and ORCAS-B~\cite{DBLP:journals/siamcomp/DuttaGLV10}, (ii) store, during a write operation, the entire value being written in each server, such as ORCAS-A~\cite{DBLP:journals/siamcomp/DuttaGLV10}, and by that incurs a worst-case storage cost, as in~\cite{Attiya:1995,DBLP:conf/ftcs/LynchS97}, or (iii) uses a message dispersal primitive and a reliable broadcast primitive, such as~\cite{DBLP:conf/dsn/CachinT06}, which during write operations, can let the storage cost to become as large as the storage cost of replication, see~\cite{DBLP:journals/dc/CadambeLMM17} for details. 

\subparagraph{Advancing the state-of-the-art.~~}
In the context of self-stabilization, we cannot consider unbounded storage cost and this paper, unlike~\cite{DBLP:journals/dc/CadambeLMM17,DBLP:conf/opodis/AndroulakiCDV14,DBLP:conf/sosp/HendricksGR07,DBLP:journals/siamcomp/DuttaGLV10}, presents a bound on the storage costs also in the absence of a fair scheduler.
Thus, our proposal goes beyond the state-of-the-art in the case of (i) not only in the context of self-stabilization and privacy. Moreover, unlike~\cite{DBLP:conf/dsn/CachinT06,Attiya:1995,DBLP:conf/ftcs/LynchS97,DBLP:journals/siamcomp/DuttaGLV10}, during a single write operation, the added storage cost of the proposed algorithm are similar to the ones of CASGC~\cite{DBLP:journals/dc/CadambeLMM17}. Thus, our proposal goes beyond the state-of-the-art in the cases of (ii) and (iii) not only in the context of self-stabilization and privacy.

\subsubsection{Self-stabilizing register emulation in message-passing systems}
As of the time of our publication~\cite{DBLP:conf/podc/DolevPS15}, to the best of our knowledge, there was no self-stabilizing solution with write operations that do not replicate the new object version among all the system servers. Also, privacy is not considered. 

Self-stabilizing emulation of shared registers that have weaker properties than atomicity (and do not consider fail-stop failures) exists~\cite{DBLP:conf/wss/DolevH01,DBLP:conf/icdcn/JohnenH09} as well as~\cite{DBLP:journals/corr/BonomiPPT16a}, which consider Byzantine nodes but not atomicity. Dolev et al.~\cite{DBLP:journals/tpds/DolevIM97} presented a self-stabilizing algorithm for emulating atomic single-writer single-reader (SRSW) shared register in message-passing systems. This work considers many-reader and many writer (MRMW) atomic registers.

\subparagraph{Recent advances to the state-of-the-art.~~}
Recent solutions for shared memory emulation include practically-stabilizing emulation of SWMR registers~\cite{DBLP:journals/jcss/AlonADDPT15}, and MRMW registers~\cite{DBLP:conf/sss/DolevGMS15,DBLP:conf/podc/BonomiDPR15}. Pseudo-self-stabilizing emulation of atomic registers is considered in~\cite{DBLP:conf/opodis/DolevDPT12} for the case of SWMR. 

During asynchronous system runs that are not controlled by fair schedulers, pseudo-self-stabilizing and practically-self-stabilizing systems satisfy safety requirements after an unbounded recovery period (yet finite in the former case). The case of asynchronous system runs that are controlled by fair schedulers is not considered in~\cite{DBLP:conf/podc/DolevPS15,DBLP:conf/opodis/DolevDPT12} for the case of SWMR and in~\cite{DBLP:conf/sss/DolevGMS15,DBLP:conf/podc/BonomiDPR15} for the case of MWMR.

\subparagraph{Advancing the state-of-the-art.~~}
We do not claim that, in the presence of a fair scheduler, the solutions in~\cite{DBLP:conf/podc/DolevPS15,DBLP:conf/opodis/DolevDPT12,DBLP:conf/sss/DolevGMS15,DBLP:conf/podc/BonomiDPR15} have (or have not) a bounded recovery period, but we do point out that their message size is greater than our proposal by a multiplicative factor of polynomial order in the number of system nodes (in addition to the fact that their write operations replicate the new object among all servers).

Our self-stabilizing proposal has a bounded recovery period in the presence of seldomly fair schedulers. Moreover, in the absence of transient faults (that corrupt the control variables), our self-stabilizing solution works well in the absence of fair schedulers. Furthermore, one can replace the type of control variables (tags) that we use with the one of the control variables in~\cite{DBLP:conf/podc/DolevPS15,DBLP:journals/jcss/AlonADDPT15} and abandon merely the part of our proposal that appears in Section~\ref{fig:boundCAS}. This replacement is straightforward. The result will be a practically-self-stabilizing variance of Cadambe et al.~\cite{DBLP:journals/dc/CadambeLMM17} that has a much better use of storage comparing to~\cite{DBLP:conf/podc/DolevPS15,DBLP:conf/opodis/DolevDPT12,DBLP:conf/sss/DolevGMS15,DBLP:conf/podc/BonomiDPR15} (at a costs of polynomial factor of the message size and no bounded recovery period).

We note that after our publication~\cite{DBLP:conf/podc/DolevPS15}, several important results were added to the literature. For example, Spiegelman et al.~\cite{DBLP:conf/podc/SpiegelmanCCK16} considered data items of $D$ bits, concurrency degree if $\delta$, and an upper bound on the number of storage node failures $t$, they show a lower bound of $\Omega(\min(t,\delta)D)$ bits on the space complexity of asynchronous distributed storage algorithms. This implies, for example, that the asymptotic storage cost can be  as high as $\bigO(\delta D)$. Our upper bound on the storage size (Section~\ref{sec:cost}) does not contradicts the lower bound of Spiegelman et al.~\cite{DBLP:conf/podc/SpiegelmanCCK16} and their $\Theta(\min(t,\delta)D)$ upper-bound does not consider self-stabilization.   
To the best of our knowledge, additional advances in the area of coded atomic storage~\cite{DBLP:conf/podc/KonwarPLM17,DBLP:conf/ipps/KonwarPKLMS16,DBLP:conf/opodis/KonwarPLM16}, which appeared after our publication~\cite{DBLP:conf/podc/DolevPS15}, do not consider self-stabilization. 



\subsubsection{Privacy preservation and malicious tolerance}
The CAS algorithm~\cite{DBLP:journals/dc/CadambeLMM17} uses erasure codes for splitting the data into different coded elements that each server stores. As long as at least $k$ coded-elements are available, the algorithm can retrieve the original information. Cadambe et al.~\cite{DBLP:journals/dc/CadambeLMM17} show how to use $(N,k)$-maximum distance separable (MDS) codes~\cite{Roth:2006} for
improving communication and storage performances. $(N,k)$-MDS codes map
$k$-length vectors to an $N$-length ones. The CAS algorithm lets the writers
to store on $N$ servers $k$-length vectors. Each of the $N$ servers stores
(uniquely) one of the $N$ coordinates of the $(N,k)$-MDS-coded information.
When retrieving the information, the algorithm can tolerate up to $(N-k)$
erasures. We address privacy by storing on each server merely parts of the data, as in Shamir's secret sharing scheme~\cite{Shamir:1979}, which we can implement by Reed-Solomon codes~\cite{McEliece:1981} and a matching error correction
algorithm (Berlekamp-Welch~\cite{BerlekamWelch:1986}).


\subsubsection{Proposed techniques of independent interest.~~}
We note that our proposal~\cite{DBLP:conf/podc/DolevPS15} enhances CASGC~\cite{DBLP:journals/dc/CadambeLMM17} from the privacy perspective by also from the system robustness point-of-view. We use here several techniques of independent interest that facilitate this improvement. For example, before adding a new object version, the writers query for the maximum value of the control variables, which are called tags. The writer then couples the new version with a tag number that is greater than the one returned by the query. We show this technique preserves atomicity proof of~\cite{DBLP:journals/dc/CadambeLMM17} and believe that it is suitable for many other self-stabilizing algorithms. 

Our solution also deals with the following interesting challenge. The rate in which clients complete write operations can be much faster than the rate in which these clients can inform all the servers about these operations. This rate can also exceed the rate in which the servers can inform each other about such updates. The challenge here is imposed by the fact that self-stabilizing end-to-end protocols must assume that the communication channels have bounded capacities due to well-known impossibility results~\cite[Chapter 3.2]{Dolev:2000}. Our solution overcomes this challenge using techniques that resemble the ones for converting shared memory models to message-passing ones~\cite[Chapter 4.2]{Dolev:2000} and an extra phase in the writer procedure. This part of the solution is another key difference between the proposed algorithm and the one by Cadambe et al.~\cite{DBLP:journals/dc/CadambeLMM17}.

To the end of bounding the number records that each server needs to store, at any point of time, a given server record is considered relevant only as long as the servers use it. We show that no server store more than $N + \delta + 3$ relevant records during asynchronous system runs that are not necessarily controlled by a fair scheduler, where $\delta$ is a bound on the number of write operations that occur concurrently with any read operation; this is similar to the $\delta$ parameter defined by Cadambe et al.~\cite{DBLP:journals/dc/CadambeLMM17}. 
The proof technique serves as a self-stabilizing alternative to existing non-self-stabilizing algorithms that provide bounds on the number of records at the server storage, such as~\cite{DBLP:journals/siamcomp/DuttaGLV10,DBLP:conf/dsn/CachinT06}, in a way that does not require storage costs during write operations to be the ones of a fully replicated solution.

\subsection{Our Contributions}
We present the algorithmic design for an important component for dependable distributed systems: a robust shared storage that preserves privacy. In particular, we provide a privacy-preserving and self-stabilizing algorithm for decentralized shared memory emulation (over asynchronous message-passing systems)   that is resilient to a wide spectrum of node and communication failures as well as malicious behavior. Moreover, our self-stabilizing algorithm can automatically recovery after the occurrence of transient faults that violate the assumptions according to which the system is to behave. Concretely, we present, to the best of our knowledge, the first solution that provides: 

\begin{enumerate} 
\item  \emph{Dependable and efficient emulation of atomic registers over asynchronous message-passing systems.~~} 
When starting from a legitimate state, our self-stabilizing solution can:

\begin{itemize} 

\item \textit{Deal with communication failures:~~}
The communication channels that are prone to packet failures, such as omission, duplication, reordering, but the resulted communication delays are unbounded yet finite since we assume communication fairness. (That is, it might take a finite number of retransmissions, but packets are received eventually.)

\item \textit{Deal with node failures:~~} 
We show that non-failing clients can retrieve information stored privately by the $N-f$ non-failing servers. We do not bound the number of failing clients but we do assume a bound of {$N$} on the number of concurrently active clients.

\item \textit{Deal with malicious behavior:~~} 
We show that the client can retrieve the originally stored object in the presence of at most $e$ malicious servers.

\item \textit{Prevent information leakage:~~} 
We show that the collective storage of any set of fewer than $k-1$ servers cannot reveal (any version) of the object.
\end{itemize}

\item \emph{Recovering after the occurrence of transient failures.~~} We show that our algorithm can even recover after the occurrence of transient failures in the following cases. The solution presentation considers two `attempts' to solve the problems until the third attempt provides a self-stabilizing solution. 


\begin{itemize} 

\item \textit{Unbounded control variables and number records at the server storage:~~} 
We show that starting from an arbitrary system state and within $\bigO(1)$ time of fair execution, the system reaches a legitimate state after which the algorithm satisfies the ${\text{CAS}(k)}$'s task requirements even when the scheduler stops been fair and the execution becomes asynchronous. This `first attempt' solution assumes that the servers can store all the object versions (in addition stale information originated from the system starting state). 

\item \textit{Unbounded control variables but a bounded number of records at the server storage:~~} 
We bound the number of relevant records that any server stores, at any point of time, by $N + \delta + 3$ during asynchronous system runs that are not necessarily controlled by a fair scheduler, where $\delta$, similar to Cadambe et al.~\cite{DBLP:journals/dc/CadambeLMM17}, is a bound on the number of write operations that occur concurrently with any read operation. 

\item \textit{Bounded control variables and number of records at the server storage:~~} 
The challenge here comes from the fact that any transient fault can bring the control variables to their maximal values. The difficulty here is that there is a need to allow the system to perform an unbounded number of write operations after this overflow event. We address this challenge by using a safety-preserving global restart of the control variables (in a way that may temporarily violate liveness but will leave the most recent version of the object intact). 

\end{itemize}

\end{enumerate}

Another important contribution of this work is the proposal of new design criteria for self-stabilizing systems of self-stabilization in the presence of seldom fairness.  
On the one hand, the proposed design criteria consider a greater set of algorithms that can be considered self-stabilizing when comparing to other design criteria~\cite{Alon:2011,DBLP:conf/netys/SS18,DBLP:conf/opodis/DolevDPT12,DBLP:conf/sss/DolevGMS15,DBLP:journals/dc/BurnsGM93} that do not consider execution fairness at all, not even seldomly. On the other hand, it is much easier to design algorithms for the proposed design criteria than the ones in~\cite{Alon:2011,DBLP:conf/netys/SS18,DBLP:conf/opodis/DolevDPT12,DBLP:conf/sss/DolevGMS15,DBLP:journals/dc/BurnsGM93}.

\subsection{Solution outline and document organization}
We bring our interpretation of the system in the self-stabilization context and the CAS task (Section~\ref{sec:arch}) before the bringing Cadambe et al.~\cite{DBLP:journals/dc/CadambeLMM17} version of CAS (Section~\ref{sec:back}). We present our privacy-preserving variation of Cadambe et al.'s algorithm as a basic result (Section~\ref{sec:basicRes}). 

Our self-stabilizing algorithm requires the specification of a formal model (Section~\ref{sec:sys}) and external building blocks (Section~\ref{s:exbld}). The presentation of this algorithm starts by considering it unbounded version  (Section~\ref{sec:unbAlg}) together with its correctness proof (Section~\ref{sec:proof}). Our proof also shows that there is a bounded set of relevant records that the servers store (Section~\ref{sec:bound}). This bound is the basis for the bounded variation of the proposed self-stabilizing algorithm (Section~\ref{sec:bounded}) and its cost analysis (Section~\ref{sec:cost}).

The discussion (Section~\ref{sec:disc}) includes also an elegant extension that extends our settings to consider the possible recovery of failing nodes. We present self-stabilizing implementations (Section~\ref{sec:basic:App}) of the gossip and quorum services (specified in Section~\ref{s:exbld}). This part appears in the Appendix because we do not consider it to be a major part of our contribution.

\section{System Overview}
\label{sec:arch}
The design criteria of self-stabilization have considerations that must be taken into account (in addition to the ones that exists for non-stabilizing systems). Therefore, before describing the algorithm by Cadambe et al.~\cite{DBLP:journals/dc/CadambeLMM17} and proposing our variation (sections~\ref{sec:sys} to~\ref{sec:cost}), this section brings the studied task (Section~\ref{ss:emulatedObject}) and our interpretation of the system contexts that do (Figure~\ref{fig:system}, right), and do not (Figure~\ref{fig:system}, left), consider privacy and self-stabilization. We note that in the context of self-stabilization, all system components have to follow the self-stabilization criteria.  

\begin{figure*}[t]
   \centering
   \includegraphics[clip=true,scale=1.3]{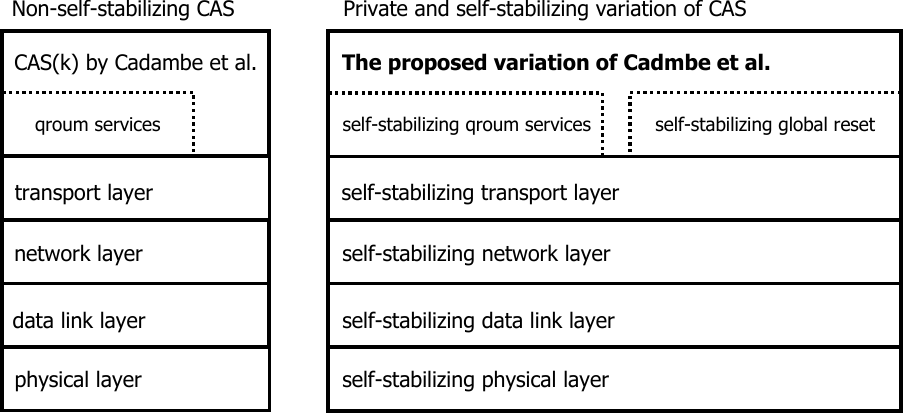}
\caption{A possible deployment of the CAS(k) algorithm by Cadambe et al.~\cite{DBLP:journals/dc/CadambeLMM17} (on the left) and the proposed self-stabilizing variation (on the right); this paper contribution appears in bold font.}
   \label{fig:system}
\end{figure*}

\subsection{Emulating shared objects}
\label{ss:emulatedObject}
The network include \textit{nodes} $\sP  = \{p_1, \ldots, p_N\}$ (processors). Each node $p_i \in P$ has access to a unique identifier $i$ and hosts either (i) a server, (ii) a client or (iii) both a server and a client. The server has access to a storage $S$, which is a set of records, and the client request the servers to use these records for updating and retrieving the latest version of the emulated shared object. The coded atomic storage ${\text{CAS}(k)}$ task addresses the problem of multi-writer, multi-reader (MWMR) emulation of atomic shared memory of a single object in the above settings. The system uses erasure coding for the sake of tolerating fail-stop failures of at most $f$ servers (Section~\ref{sec:MDSback}). 

The object value is a member of a finite set $\sV$, which $\lceil \log_2 |\sV| \rceil$ bits can represent. We refer to $v_0 \in \sV$ as the default (initial) state of the emulated object. A local source commands its client to run the reader or writer procedures, sequentially. A call to a reader returns the current version of the object value. A call to a writer includes the new version of object value and returns upon completion. A writer associates each write request with a unique tag, $t \in \sT$, where $\sT = \mathbb{Z}^+ \times P$ and $\mathbb{Z}^+$ is the set of all positive integers. Note that $\sT$ is a set for which the relation $<\equiv (z_1<z_2)\lor ((z_1=z_2) \land (i<j))$ can order totally any pair of tags, $(z_1,p_i)$ and  $(z_2,p_j)$. We denote the default tag value, $t_0 < \min \sT$, as a tag that is not in $\sT$ and yet it is smaller than any other tag in $\sT$.

The detailed specification of task $\text{CAS}(k)$~\cite{DBLP:journals/dc/CadambeLMM17} and~\citep[Chapter 13]{DBLP:books/mk/Lynch96} considers each version of the object and requires the algorithm's external behavior follows the ones of atomic (linearizable) memory. An atomic shared memory object is one where the commands to the clients and the returned values from these calls appear as if the object is being accessed sequentially rather than via concurrent calls to the client procedure. The detailed task specification requires that it would be possible to include in the system execution serialization points so that the trace of the complete operations corresponds to the one of a read-write variable type. $\text{CAS}(k)$ also requires liveness with respect to the completion of all (non-failing) operations in any (not necessarily always fair) execution in which the number of server failures is at most $f$, where $k \in \{1, \ldots, N-2f\}$.  

\subsection{External building blocks}
We handle node and communication failures as well as transient faults using common external building blocks.

\subsubsection{End-to-end protocols}
%
%
The implementation of the system services requires the availability of an end-to-end protocol. Our self-stabilizing implementation of the services below assumes the availability of self-stabilizing end-to-end protocols, such as ones in~\cite{DBLP:journals/ipl/DolevDPT11,DBLP:conf/sss/DolevHSS12}. Note that self-stabilizing end-to-end protocols assume that the channel has a bounded capacity due to well-known impossibility results~\cite[Chapter 3.2]{Dolev:2000}.

\subsubsection{Gossip services}
Cadambe et al. assume the availability of a reliable gossip service. They use this non-self-stabilizing service to propagate reliably among the servers the tag of every object version. We consider a self-stabilizing gossip service (which we specify in Section~\ref{sec:basic} and suggest an implementation in Appendix~\ref{sec:basic:App}). This service lets each gossip message to overwrite the previous gossip message that is stored in the buffers (without considering whether the previous message was delivered to all receivers). 

Our specifications are motivated by the fact that self-stabilizing end-to-end protocols must consider communication channels with bounded capacities~\cite[Chapter 3.2]{Dolev:2000}. Therefore, a specific quorum of servers might process write operations much faster than the rate in which gossip messages arrive \emph{reliably} to servers that are not part of that quorum. Since the communication channels are assumed to be bounded, it is not clear how can the writer avoid blocking (and still deliver all gossip messages).  

\subsubsection{Quorum services}
Quorum systems can be used for ensuring transaction atomicity in replica system despite the presence of network failures~\cite{DBLP:conf/berkeley/Skeen82}. The term quorum system, $\sQ$, refers to all subsets of $P$, such that each quorum set $Q \in \sQ$ satisfies the quorum system specifications. For example, Attiya et al.~\cite{Attiya:1995} specify the criterion of $\lceil \frac{N}{2}\rceil < |Q|$, Cadambe et al.~\cite{DBLP:journals/dc/CadambeLMM17} consider $k \in \{1, \ldots, N-2f\}$ and specify $\lceil \frac{N+k}{2}\rceil \leq |Q|$. Our specifications (Section~\ref{sec:specQ}) consider $\lceil \frac{N+k+2e}{2}\rceil$, where $e$ is the maximal number of malicious servers and $f$ is a bound on the ones that can fail-stop. 

Cadambe et al.~\cite{DBLP:journals/dc/CadambeLMM17} assume that the operations at a given client follow a ``handshake'' discipline, where a new invocation awaits the response of a preceding invocation. In the context of self-stabilization, this synchronization between clients and servers is subject to transient faults. Thus, we specify a service that provides this ``handshake'' discipline in a self-stabilizing manner (Section~\ref{sec:basic}). (We offer a self-stabilizing implementation of the service in Appendix~\ref{sec:basic:App}.)   

\subsubsection{Reset services}
\label{sec:resetServices}
Non-self-stabilizing algorithms for shared memory emulation might never reach a tag value that is (close to) the maximum of $\sT$, because $\sT = \mathbb{Z}^+ \times P$ is unbounded. However, self-stabilizing algorithm consider only bounded tag values, i.e., for them $\sT = \{1, \ldots, z_{\max}\} \times P$, where $z_{\max} \in \mathbb{Z}^+$ is a predefined positive integer. (We use the same notation of $\sT$ for both variations whenever it is clear from the context whether the system considers self-stabilization.) Since in the context of self-stabilizing systems a single transient fault can introduce a tag value that is (close to) the maximum of $\sT$, the proposed algorithm has to overcome a challenge that non-self-stabilizing algorithms for shared memory emulation often do not consider (with some notable exceptions, e.g.,~\cite[Section 5]{Attiya:1995}).   

The proposed self-stabilizing algorithm uses a self-stabilizing global reset mechanism that resembles the one considered in~\cite{DBLP:conf/wdag/AwerbuchPVD94}. It helps the algorithm to overcome the case in which the system state includes a tag that is (close to) the maximum value in $\sT$. This reset mechanism leaves the storage of every server only with the most recent version of the object and replaces its tag value with a tag that is slightly above $t_0$. We specify the interface between the proposed algorithm and the self-stabilizing global reset mechanism (Section~\ref{s:bld}) and note that its liveness property requires schedule fairness. 

\section{Background}
\label{sec:back}
Cadambe et al.~\cite{DBLP:journals/dc/CadambeLMM17} use erasure codes for emulating shared memory and use quorums to distinguish among writer, server and reader nodes. Their algorithm allows multiple writers using a $(N,k)$ maximum distance separable~\cite{Roth:2006} (MDS) code to write data concurrently to the group of servers while ensuring atomicity and liveness. This section reviews the definition of $(N,k)$ MDS code before explaining how to use them for secret sharing in a slightly adapted variation of Cadambe et al.~\cite{DBLP:journals/dc/CadambeLMM17}.

Cadambe et al.~\cite{DBLP:journals/dc/CadambeLMM17} divide the data into a number of coded elements. Each server stores at most one coded element. Cadambe et al. guarantee that the reader client can fetch the necessary number of coded elements, such that the reader can retrieve the original data. Given two positive integers $m, k \in \mathbb{Z}^+: k < m$,  Cadambe et al. consider an $(m, k)$ Maximum Distance Separable (MDS) code that maps a $k$-length vector (the input) to an $m$-length vector (the output). The aim is that after altering arbitrarily $k$ coordinates of the output vector, a decoding algorithm can still retrieve the input vector. This way, Cadambe et al. use an $(m, k)$ code for storing the input vector on $m$ servers, i.e., the server at $p_i$ stores the output's $i$-th coordinate, because the decoding algorithm is resilient to $(m-k)$ node failures. We bring the definition of $(m, k)$ MDS code (Section~\ref{def:MDS}) before proving the Cadambe et al.'s ${\text{CAS}(k)}$ algorithm (Section~\ref{sec:CadambeAlg}). 

\subsection{Maximum Distance Separable (MDS) codes}
\label{sec:MDSback}
Let $A$ be an arbitrary finite set and $S \subseteq \{1, 2, \ldots ,m\}$. Denote by $\pi_S$ the natural projection mapping from $A^m$ onto $S$'s corresponding coordinates, i.e., $S = \{s_1, s_2, \ldots , s_{|S|}\}$, where $s_1 < s_2 \ldots < s_{|S|}$, and define $\pi_S : A^m \rightarrow A^{|S|}$ as $\pi_S (x_1, x_2, \ldots , x_m) = (x_{s_1}, x_{s_2}, \ldots , x_{s_{|S|}})$.

\begin{definition}[Maximum Distance Separable (MDS) code]
\label{def:MDS}
Let $A$ be a finite set and $m, k \in \mathbb{Z}^+: k < m$ two positive integers. An $(m, k)$ code over $A$ is a map $\Phi : A^k \rightarrow A^m$. An $(m, k)$ code $\Phi$ over $A$ is said to be Maximum Distance Separable (MDS) if, for every $S \subseteq \{1, 2, \ldots ,m\}$, such that $|S| = k$, there is a function $\Phi^{-1}_S : A^k \rightarrow A^k$, such that $\Phi^{-1}_S (\pi_S(\Phi(x)) = x $ for every $x \in A^k$, where $\pi_S$ is the natural projection mapping.
\end{definition}

Cadambe et al.~\cite{DBLP:journals/dc/CadambeLMM17} refer to each of the output's coordinates of an $(m, k)$ code $\Phi$ as a coded element. Further details about $\Phi$ and erasure code appear in~\cite{DBLP:journals/dc/CadambeLMM17}. We extend the use of $(m, k)$ MDS code to secret sharing (Section~\ref{sec:MDSuse}).

\begin{algorithm*}[t!]
\algoSize
%
\BlankLine

\textbf{The client:~}\tcc*{At any time, $p_i$'s client is a writer, a reader or none but not both}
\textbf{$writer(s)$:~}\tcc*{The writer stores secret $s$ as the new version of the shared object}
\tcc{Query for finalized tags and after hearing from a quorum get the maximal tag}
\textbf{let} $(z,j) := \max (\{t' : ( t', \bullet ) \in \text{qrmAccess}((\bot,\bot,\text{`qry'}))\})$\nllabel{BACKalg:w:SelMax}\tcc*{Obtain coded elements $w_1,w_2, \ldots, w_N$, such that $p_i \in P$ has a server, by applying the $\Phi$ to the secret $s$. Then, send $(t,w_i, `pre')$ to every server and wait for a quorum of replies.}
$\text{qrmAccess}(((z+ 1,i),\{\Phi_{p_j}(s)\}_{p_j\in P},\text{`pre'}))$\tcc*{ The prewrite phase }\nllabel{BACKln:preWrite}
\tcc{For each server, send $(t, `null', \text{`fin'})$ and wait for a quorum of replies.}
$\text{qrmAccess}(((z+ 1,i),\bot,\text{`fin'}))$\nllabel{BACKln:finalize}\tcc*{ The finalize phase }
$\Return$\nllabel{BACKalg:w:return}\;

\BlankLine\BlankLine

\textbf{$reader()$:~}\tcc*{The reader retrieves the current object version, or $\bot$ upon failure}
\textbf{let} $t := \max (\{t' : (t', \bullet ) \in \text{qrmAccess}((\bot,\bot,\text{`qry'}))\})$\nllabel{BACKalg:r:SelMax}~\tcc*{Query as in line~\ref{alg:r:SelMax}}
\tcc{For each server, send $(t,\bot,\text{`fin'})$ and wait for a quorum of replies with the requested coded elements, which are associated with tag $t$.}
\textbf{let} $Q:=\text{qrmAccess}((t,\bot,\text{`fin'}))$\nllabel{BACKln:finalizeR}\tcc*{Ask and wait for finalized records from a quorum}
\tcc{Test whether at least $k_{threshold}$ replies include coded elements so that $\Phi^{-1}$ can decode the secret before returning it. If the test fails, return $\bot$.
}
\lIf{$|\{(t,w,\text{`fin'})\in Q:w\neq\bot\}|\geq k_{threshold}$}{$\Return(\Phi^{-1}(w:\{(t,w,s)\in Q:w\neq\bot\}))$\nllabel{BACKalg:read:BW}}
\lElse{\Return $\bot$}\nllabel{BACKalg:read:BWfail}

\BlankLine\BlankLine

\textbf{The server:}\\ 
$S\subset\sT\times(\sW\cup\{\bot\})\times \sD$ is a record set, where $\sT$ $=$ $\mathbb{Z}^+ \times \sP$ is the set of tags, $\sW$ the set of coded words and $\sD=\{\text{`pre'}, \text{`fin'} \}$ the set of phases. When $S=\emptyset$, we use the default triple $(t_0,w_{t_0,i}, \text{`fin'})$ when reporting on the triple with the highest locally known tag\nllabel{BACKalg:srv:default}\; 

\BlankLine\BlankLine

\Upon query arrival from $p_j$'s client to $p_i$'s server \Do{\label{BACKgc:alg:srv:queryInit}
Reply with $(maxPhase(\text{`fin'}), \bot, \text{`qry'})$, where $maxPhase(\text{`fin'})$ refers to the highest tag in any record record in $S$ that has a label $\text{`fin'}$ (whether that record includes a coded element or not)}


\Upon pre-write $m:=(t, w, \text{`pre'})$ arrival from the $p_j$'s writer to $p_i$'s server 
   \Do{\label{BACKgc:alg:srv:pre-write}
\lIf(~~~~~~~~~~~~~~~~~~~~~~~~~~~~~~~~~~~\texttt{/* add the arriving record to $S$ */}){$\nexists (t,\bullet)\in S$}{$S\gets S\cup\{(t,w,\text{`pre'})\}$} Moreover, acknowledge the arriving record by calling $\text{reply}(j,m)$.
}


\Upon finalize $m=\langle t, \bot, \text{`fin'} \rangle$ arrival {from} $p_j$'s writer to $p_i$'s server \label{BACKgc:alg:srv:fin-writer} \Do{
\uIf(~~~~~~~~~~~~~~~~~~~~~~~~\texttt{/* update the record $(t,w, \text{`pre'})$ to $(t,w, \text{`fin'})$ in $S$ */}){$\exists (t,w, \text{`pre'}) \in S$}{$S\gets (S\setminus\{(t,w,\text{`pre'})\})\cup(t,w,\text{`fin'})$}
\lElse{add $(t, \bot, \text{`fin'})$ to $S$} Moreover, acknowledge to the writer by calling $\text{reply}(j,m)$ and gossip the message $(t)$ to all other servers by calling $\text{gossip}(t)$.
}


\Upon finalize $m:=(t,\bot,\text{`fin'})$ arrival {from} $p_j$'s reader to $p_i$'s server \label{BACKgc:alg:srv:fin-reader} \Do{
\uIf{$\exists (t,w_i, \bullet) \in S$}{$S\gets(S\setminus\{(t,w,\bullet)\})\cup\{(t,w,\text{`fin'})\}$ \tcc*{update the record $(t,w_i, \bullet)$ to $(t,w_i, \text{`fin'})$ in $S$} acknowledge the reader with $(t,w_i,\text{`fin'})$\;} 
\Else{$S\gets S\cup\{(t,\bot,\text{`fin'})\}$ \tcc*{add $(t, \bot, \text{`fin'})$ to $S$} acknowledge to the reader by calling $\text{reply}(j,m)$\;} 
Moreover, gossip the message $(t)$ to all other servers by calling $\text{gossip}(t)$.
}


\Upon gossip $(t)$ arrival {from} $p_j$'s server to $p_i$'s server\label{BACKgc:alg:srv:uponGossip} \Do{
\lIf{ $\exists (t, \bullet) \in S$}{update the record $(t, \bullet)$ to $(t, \bullet, \text{`fin'})$ in $S$ \textbf{else} add $(t, \bot, \text{`fin'})$ to $S$}
}

\BlankLine

%
%
\caption{\algoSize A non-self-stabilizing $CAS(k)$ algorithm (that is based on Cadambe et al. with adaptations for the proposed secret sharing scheme), code for $p_i$'s client and server.}
\label{BACKalg:cas}
\end{algorithm*}

\subsection{Cadambe et al.'s {\textbf{CAS}(\textit{k})} algorithm}
\label{sec:CadambeAlg}
Cadambe et al.~\cite{DBLP:journals/dc/CadambeLMM17} present a quorum-based algorithm for implementing the $CAS(k)$ task. Algorithm~\ref{BACKalg:cas} is our interpretation of the non-self-stabilizing $CAS(k)$ algorithm by Cadambe et al.~\cite{DBLP:journals/dc/CadambeLMM17} with slight adaptations for the proposed secret sharing scheme. 

\subsubsection{External building blocks: quorum and gossip communications}
\label{sec:extBack}
Cadambe et al.~\cite{DBLP:journals/dc/CadambeLMM17} specify $\lceil \frac{N+k}{2}\rceil \leq |Q|$ for any $k \in \{1, \ldots, N-2f\}$ and show Lemma~\ref{thm:Q}.

\begin{lemma}[Lemma 5.1 in~\cite{DBLP:journals/dc/CadambeLMM17}]
\label{thm:Q}
Suppose that $k\in\{1\ldots, N-2f\}$. (i) If $Q_1, Q_2 \in \sQ$, then $|Q_1 \cap Q_2| \geq k$. (ii) If the number of failed servers is at most $f$, then $Q$ contains at least one quorum set $Q$ of non-failed servers. 
\end{lemma}

Algorithm~\ref{BACKalg:cas} access the servers via a call to the function $\text{qrmAccess}()$, which returns a set of replies (records) from at least a quorum of servers. Algorithm~\ref{BACKalg:cas} also assumes the availability of a reliable gossip service, which allows the servers to send their most recent finalized tags $t \in T$.

\subsubsection{Local variables}
The state of the server includes a set of records $(t, w, label)\in S\subset\sT\times(\sW\cup\{\bot\})\times \sD$ (line~\ref{BACKalg:srv:default}), where the label $d \in \{\text{`pre'}, \text{`fin'}\}$ refers to metadata that records the phases of the shared-object updates. The clients carry these updates sequentially and in each phase they access the quorum system do not end the phase before getting replies from at least a quorum. Algorithm~\ref{BACKalg:cas} assumes that when $S=\emptyset$, the default triple $(t_0,w_{t_0,i}, \text{`fin'})$ is included in $S$ when reporting on the triple with the highest locally known tag.  

\subsubsection{Protocol phases}
Both the writer and reader protocols use the query phase for discovering a recent record with the label $\text{`fin'}$ as its metadata (line~\ref{BACKgc:alg:srv:queryInit}). During the pre-write phase of write operations (line~\ref{BACKgc:alg:srv:pre-write}), the writer makes sure that at least a quorum of servers, say $Q_{pw}$, store each a coded element with the tag $t'$ and label $\text{`pre'}$. Note that immediately at the end of the prewrite phase, the stored record cannot be accessed by the readers, because when a server replies to queries considers only records with finalized tags (line~\ref{BACKgc:alg:srv:queryInit}). However, after the prewrite phase, the writer starts the finalize phase (line~\ref{BACKln:finalize} and~\ref{BACKgc:alg:srv:fin-writer}), which diffuses the records with the label $\text{`fin'}$ and the tag $t'$ and then waits for a quorum of servers, say $Q_{fw}$, to reply. Immediately after this finalized phase, any query phase (of any read or write operation) will retrieve a tag that is at least as high as $t'$ (because by Lemma~\ref{thm:Q} it holds that $Q_{pw}$ and $Q_{fw}$ must interest) and in that sense tag $t'$ is viable to all clients. Moreover, the existence a stored record that its label is $\text{`fin'}$ implies that the coded elements associated tag $t'$ are stored by at least a quorum of servers, which is $Q_{pw}$. This property allows the reader to retrieve at least $k_{threshold}$ unique coded elements (line~\ref{alg:r:SelMax} to~\ref{BACKalg:read:BWfail}), which are stored at the servers of $Q_{pw}$. Cadambe et al.~\cite{DBLP:journals/dc/CadambeLMM17} set the value of $k_{threshold}$ to $k$ (whereas we consider another value in Section~\ref{sec:basicRes}). We also note that the reader further facilitates the diffusion of finalized tags to a quorum (line~\ref{BACKln:finalizeR}). This and the gossip messages (line~\ref{BACKgc:alg:srv:uponGossip}) allows the system to complete the diffusion of finalized records in the presence of fail-stop failures of writers.

\begin{corollary}[\cite{DBLP:journals/dc/CadambeLMM17}, Theorem~1]
\label{thm:memory}
Algorithm~\ref{BACKalg:cas} emulates a shared atomic read/write memory.
\end{corollary}


\section{Basic Results}
\label{sec:basicRes}
We present a variance that adds privacy provision to the implementation proposed by Cadambe et al.~\cite{DBLP:journals/dc/CadambeLMM17}. 
Our variation allows at most $e$ malicious servers and at most $f$ failures in an asynchronous message-passing system. 
In this section, we consider malicious servers can send corrupted secret shares to readers, but not corrupted tags or labels, i.e., when a malicious server replies with a tuple $(t,w,d)$, only $w$ might be corrupted. Writers divide secrets and submit the resulting secret shares to the servers. Servers store their secret shares and deliver them to the readers upon request.
In sections~\ref{sec:sys} to~\ref{sec:cost}, we extend our proposal to withstand fail-stop failures and server malicious behavior to also consider recovery after the occurrence of transient faults.

\subsection{Using \textbf{(m, k)} MDS codes for secret sharing}
\label{sec:MDSuse}
The $(N,k)$-MDS code enables the reader to restore the data under the presence of $\frac{N-k}2$ stop-failed servers.
The $(N,k)$-\textit{threshold scheme} for integers $k$ and $N$, such that
$0<k\leq N$,
is defined by Shamir~\cite{Shamir:1979} and splits a secret $s$
into $N$ secret shares $\{s_i\}_{i\in\{1,\ldots,N\}}$. This scheme requires that there exists a
mapping from any $S\subseteq\{s_i\}_{i\in\{1,\ldots,N\}}$ with $|S|\geq k$ to the secret $s$, but
it is impossible to determine $s$ from a set of less than $k$ secret shares.
%

Let $K$ be a finite field, such that its size $|K|$ is a prime number.
The $(N,k)$-Reed-Solomon code, $\Phi:\sS\to\sW$, transforms the input data, i.e., one
element of a $k$ dimensional vector space, $\sS$, over $K$,
into $N$ dimensional vector space, $\sW$, over the same field, $K$, where $k$ and $N$ are as above.  
We call $N$ the block length and $k$ the message length.
The Berlekamp-Welch algorithm, $\Phi^{-1}$, can correct
$(N,k)$-Reed-Solomon codes within $\sO(N^3)$ time in the presence of $e$ 
errors and $f$ erasures, as long as $2e+f<N-k+1$~\cite{BerlekamWelch:1986}, as described by Gemmell and Sudan~\cite{GemmellSudan:1992}.
Note that $(N,k)$-Reed-Solomon codes are a $(N,k)$-threshold scheme~\cite{McEliece:1981}. 
To that end, the input vector $(\sigma_1,\ldots,\sigma_k)\in \sS$ consists of the secret
$\sigma_1$ and randomly chosen values $\sigma_2,\ldots,\sigma_k$ from a
uniform distribution over $\sS$. We use $\Phi$ to map
$(\sigma_1,\ldots,\sigma_k)$ to the secret shares $(w_1,\ldots,w_N)\in\sW$.

\subsection{Quorums of \textbf{(k+2e)}-overlap}
\label{sec:specQ}
We require that any quorum $Q \in \sQ$ has at least $\lceil
\frac{N+k+2e}{2}\rceil$ servers.
Lemma~\ref{thm:quor} uses the quorum definition to shows that any two different
quorums share at least $k+2e$ servers, rather than just $k$ of them as in Cadambe et al.~\cite{DBLP:journals/dc/CadambeLMM17}. These quorums guarantee that once a writer finishes its write operation, any reader can retrieve at least $k+2e$ secret shares and reconstruct the secret. The lemma also shows, similar to Cadambe et al.~\cite{DBLP:journals/dc/CadambeLMM17}, that any two different quorums share at least $k+2e$ servers. This guarantees that after a writer wrote to a quorum, the readers can read a set of coded elements that allows the secret reconstruction.
\begin{lemma}[Variation of~\cite{DBLP:journals/dc/CadambeLMM17}, Lemma 5.1]
\label{thm:quor}
   Suppose that $k\in\{1\ldots, N-2(f+e)\}$. \textbf{(1)} If $Q_1,Q_2\in\sQ$, then
   $|Q_1\cap Q_2|\geq k+2e$. \textbf{(2)} The existence of such a $k$
   implies the existence of $Q\in\sQ$ such that $Q$ has no crashed servers.
\end{lemma}
\begin{proof}
   \textbf{(1)} Let
   $Q_1,Q_2\in\sQ$, then 
   $
      |Q_1\cap Q_2| = |Q_1|+|Q_2|-|Q_1\cup Q_2|\geq 2\left\lceil
      \frac{N+k+2e}2\right\rceil-N\geq k+2e.
   $
   \textbf{(2)} Since there are at most $f$ crashed servers, we can
   show that without such $f$ servers, there are still enough alive servers for a quorum.
   It follows that 
   $
      N-f\geq N-\left\lfloor\frac{N-k-2e}2\right\rfloor=\left\lceil\frac{N+k+2e}2\right\rceil\text{.}
   $
\end{proof}
By Lemma~\ref{thm:quor}, the atomicity and liveness analysis in~\cite[][Theorem 5.2 to Lemma 5.9]{DBLP:journals/dc/CadambeLMM17} also holds when Algorithm~\ref{BACKalg:cas} uses $(k+2e)$-overlap quorums rather than $k$, as Cadambe et al.~\cite{DBLP:journals/dc/CadambeLMM17} indented. 

\remove{

\subsection{Quorum Systems of $(k+2e)$-overlap}
\label{s:q}
A \emph{quorum} is a server subset $Q \subseteq \sP$ with at least $\lceil \frac{N+k+2e}{2}\rceil$ elements and $\sQ$ is a \emph{quorum system} that
includes all such quorums. Note that
any two different quorums share at least $k+2e$ servers (as we show in~\cite{DBLP:conf/podc/DolevPS15}, Lemma 1). Our argumentation
also holds for the CASGC algorithm with quorums with an overlap of at least
$k$ as in~\cite{DBLP:journals/dc/CadambeLMM17}, since for our contribution we only
require that $|Q_1\cap Q_2|>0$ for every $Q_1,Q_2\in\sQ$.

The quorum guarantees (both in~\cite{DBLP:conf/podc/DolevPS15} and~\cite{DBLP:journals/dc/CadambeLMM17}) say that once a writer client finishes its write operation, any reader client can access the servers and retrieve at least $k+2e$ secret shares (or,
respectively $k$ coded values) that the client can use for reconstructing the
originally stored data.

\begin{lemma}[A variation of Lemma 5.1 in~\cite{DBLP:journals/dc/CadambeLMM17}]
\label{thm:quor}
%
   $|Q_1\cap Q_2|\geq k+2e$: $1\le k\le N-2e$ $\land$
   $Q_1,Q_2\in\sQ$.

[[@@ change (i) and (ii). @@]]
\end{lemma}

} 

\subsection{Privacy preserving variation of Cadambe et al.}
We say that a secret sharing protocol is $t$\emph{-private} when a set of at
most $t$ servers cannot compute the secret, as in~\cite{Ben-Or:1988}. Note
that a $0$-private protocol preserves no privacy. When the presence of at most $s$ failing servers (which do not deviate from the algorithm behavior) and at most $t$ malicious servers (which deviate from the algorithm behavior only by modifying the data filed of their replies to the clients), we say that the protocol is $(s,t)$\emph{-robust}. This notion is similar to $t$\emph{-resilience}~\cite{Ben-Or:1988}.

In order to tolerate at most $e$ (secret share corruptions made by) 
malicious servers, we propose Algorithm~\ref{BACKalg:cas} as a variation of Cadambe et al.~\cite{DBLP:journals/dc/CadambeLMM17} CAS algorithm that uses $(k+2e)$-overlap quorums and $(N,k)$-Reed-Solomon codes~\cite{ReedSolomon:1960}, which is an $(N,k)$-MDS~\cite{Roth:2006} code that Cadambe et al.~\cite{DBLP:journals/dc/CadambeLMM17} uses. By the atomicity and liveness analysis for the case of $(k+2e)$-overlap quorums (the remark after Lemma~\ref{thm:quor}), the reader retrieves $k+2e$ unique secret shares with at most $e$ manipulated shares. 

\subsubsection{Robustness}
Robustness is added by the ability of the Berlekamp-Welch algorithm to correct errors in the Reed-Solomon codes. Note that malicious servers only introduce corrupted secret shares. Lemma~\ref{thm:robust} shows Algorithm~\ref{BACKalg:cas}'s resilience against up to $e$ malicious servers and up to $f$ stop-failed servers.
\begin{lemma}\label{thm:robust}
   For $k\in\{1\ldots, N-2(f+e)\}$, Algorithm~\ref{BACKalg:cas} is $(f,e)$-robust.
\end{lemma}
\begin{proof}
   If a writer issues a query, pre-write and finalize operations it does not
   retrieve the secret from the server. Thus, writers are immune to
   malicious servers. Servers do not exchange secrets with other servers and
   thus are not directly affected by malicious servers. 
   
   The rest of the proof focuses on showing that when reconstructing the secret, the read operation $\pi_r$ is able to be resilient against corrupted secret shares that malicious nodes may send.
   To do this, the reader queries all the servers about the maximal finalized tags and wait for a response from at least a quorum of servers. Algorithm~\ref{BACKalg:cas} selects the maximum tag, $t$, for the returned set of tags. This tag $t$ is uniquely associated to a write that reached the finalize phase before $\pi_r$'s query. 
   The read operation $\pi_r$ then sends a finalize command on its own and waits for a quorum of servers to respond.
   Note that the reader merely collects
   secret shares from a quorum of servers, but never update coded elements that the servers stores, since $\pi_r$'s  query
   and finalize records only contains a $\bot$ in place of the coded elements, which are the secret shares.
   By Lemma~\ref{thm:quor} and Corollary~\ref{thm:memory}, it follows that any
   reader $p_i$ receives at least $k+2e$ secret shares from the finalize
   phase. Out of these $k+2e$ secret shares, at most $e$ might be corrupted. This is the case even if up to $f$ server are
   failing. Therefore, the reader can decode the secret from this collection of $k+2e$ responses by applying the Berlekamp-Welch error-correction algorithm~\cite{BerlekamWelch:1986}.
\end{proof}

\subsubsection{Privacy} Our approach ensures the privacy of the secret among servers.
Lemma~\ref{thm:private} shows that a group of less than $k$ servers are
not able to reconstruct the secret by combining the secret shares they have
stored locally.
\begin{lemma}\label{thm:private}
   For $k\in\{1\ldots, N-2(f+e)\}$, Algorithm~\ref{BACKalg:cas} is $(k-1)$-private.
\end{lemma}
\begin{proof}
   Let $t$ be a tag and $k>1$. A set of $k-1$ servers store together $k-1$
   secret shares associated with the tag $t$. Since the secret shares encode a
   secret using Reed-Solomon codes, it is impossible to compute the original
   secret with less than $k$ secret shares~\cite{McEliece:1981}. The case of
   $k=1$ implies that the secret shares are the secret itself and, thus,
   privacy is compromised, i.e., it is $0$-private. 
   It follows that Algorithm~\ref{BACKalg:cas} is $(k-1)$-private.
\end{proof}
Note that in the case of $k=1$, even if privacy is not protected, it is still
possible to decipher correctly corrupted secret shares. This holds because the reader
blocks until it reads at least $1+2e$ secret shares and, thus, the additional $2e$ secret shares
contain redundant information that allows the success of the Berlekamp-Welch code for error correction.

\section{Models}
\label{sec:sys}
We consider an asynchronous message-passing networks in which the nodes can be modeled as finite state-machines that exchange messages via communication links (with bounded capacity).
 
\subsection{Communication model}
\label{sec:commModel}
The network topology is of a fully-connected graph, $K_{N}$, and any pair of nodes has access to a bidirectional communication channel that, at any time, has at most $\capacity \in \N$ packets. Every two nodes exchange (low-level messages called) \emph{packets} to permit delivery of (high level) messages. When node $p_i \in \sP$ sends a packet, $m$, to node $p_j \in \sP\setminus \{p_i\}$, the operation ${send}$ inserts a copy of $m$ to $\mathit{channel}_{i,j}$, while respecting the upper bound $\capacity$ on the number of packets in the channel. In case $\mathit{channel}_{i,j}$ is full, i.e., $|\mathit{channel}_{i,j}|=\capacity$, the sending-side simply overwrites any message in $\mathit{channel}_{i,j}$. When $p_j$ receives $m$ from $p_i$, the system removes $m$ from $\mathit{channel}_{i,j}$. As long as $m \in \mathit{channel}_{i,j}$, we say that $m$'s message is in transit from $p_i$ to $p_j$. Recall that we assume access to a self-stabilizing end-to-end protocol~\cite{DBLP:journals/ipl/DolevDPT11,DBLP:conf/sss/DolevHSS12} that provides reliable (FIFO) message delivery (over unreliable non-FIFO channels that are subject to packet omissions, reordering and duplication).

\subsection{Execution model}
\label{sec:interModel}
Our analysis considers the \emph{interleaving model}~\cite{Dolev:2000}, in which the node's program is a sequence of \emph{(atomic) steps}. 
Each step starts with an internal computation and finishes with a single communication operation, i.e., message $send$ or $receive$. 

%
%
The {\em state}, $s_i$, of $p_i \in \sP$ includes all of $p_i$'s variables as well as the set of all incoming communication channels. Note that $p_i$'s step can change $s_i$ as well as remove a message from $channel_{j,i}$ (upon message arrival) or add a message in $channel_{i,j}$ (when a message is sent). The term {\em system state} refers to a tuple of the form $c = (s_1, s_2, \cdots, s_N)$ (system configuration), where each $s_i$ is $p_i$'s state (including messages in transit to $p_i$). We define an {\em execution (or run)} $R={c_0,a_0,c_1,a_1,\ldots}$ as an alternating sequence of system states $c_x$ and steps $a_x$, such that each system state $c_{x+1}$, except for the starting one, $c_0$, is obtained from the preceding system state $c_x$ by the execution of step $a_x$. 
%

Let $R'$ and $R''$ be a prefix, and respectively, a suffix of $R$, such that $R'$ is finite sequence, which starts with a system state and ends with a step $a_x \in R'$, and $R''$ is an unbounded sequence, which starts in the system state that immediately follows step $a_x$ in $R$. In this case, we can use $\circ$ as the operator to denote that $R=R' \circ R''$ concatenates $R'$ with $R''$. 

\subsection{Fault model}
%
%
We model a failure as a step that the environment takes rather than the algorithm. 
We consider failures that can and cannot cause the system to deviate from fulfilling its task (Figure~\ref{fig:fModel}).
The set of \emph{legal executions} ($LE$) refers to all the executions in which the requirements of the task $T$ hold. For example, $T_{\text{CAS}(k)}$ denotes our studied task of shared memory emulation and $LE_{\text{CAS}(k)}$ denotes the set of executions in which the system fulfills $T_{\text{CAS}(k)}$'s requirements. We say that a system state $c$ is {\em legitimate} when every execution $R$ that starts from $c$ is in $LE$. When a failure cannot cause the system execution (that starts in a legitimate state) to leave the set $LE$, we refer to that failure as a benign one. We consider failures that can cause the system execution to leave the set $LE$ as transient faults, which refer to any temporary violation of the assumptions according to which the system was designed to operate (as long as program code remains intact). Self-stabilizing algorithms deals with benign failures (while fulfilling the task requirements) and they can also recover after the occurrence of transient faults within a bounded period.  

\begin{figure*}[t!]
   \centering
   \includegraphics[clip=true,scale=0.75]{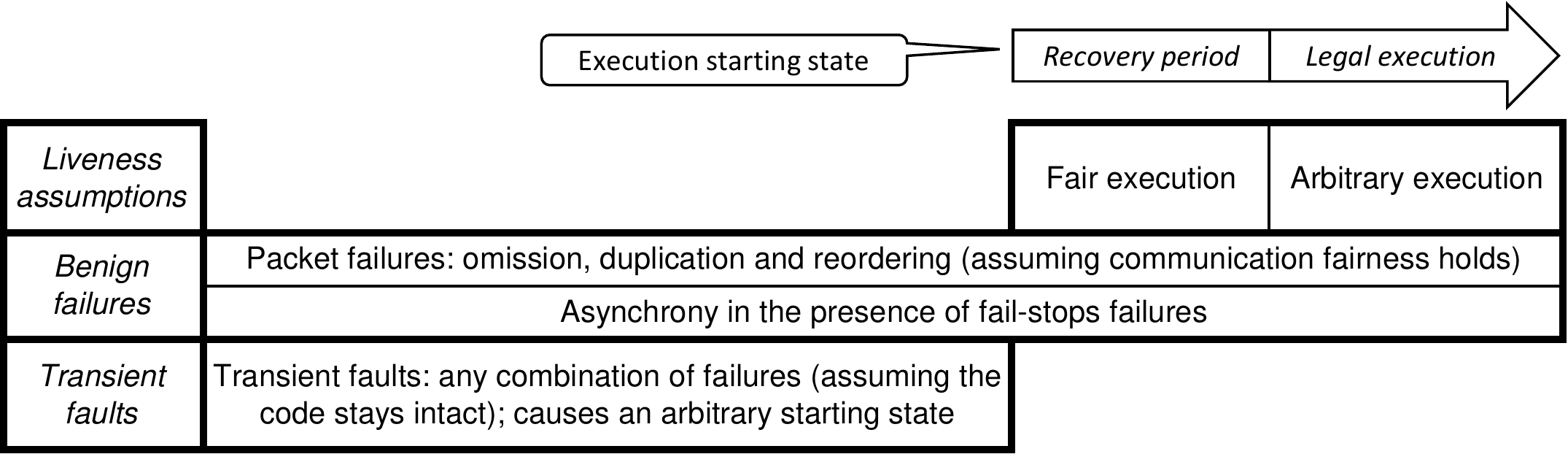}
\caption{The fault model and liveness assumptions during the system execution.}
   \label{fig:fModel}
\end{figure*}

\subsubsection{Benign failures}
\label{sec:benignFailures}
The algorithmic solutions that we consider are for asynchronous (message-passing) systems and thus they are oblivious to the time in which the packets arrive and departure (and require no explicit access to clock-based mechanisms, which may or may not be used by the system underlying mechanisms, say, for congestion control at the end-to-end protocol). 

\paragraph{Communication fairness.~~}
Recall that we assume that the communication channel handles packet failures, such as omission, duplication, reordering (Section~\ref{sec:commModel}). We assume that if $p_i$ sends a message infinitely often to $p_j$, node $p_j$ receives that message infinitely often. We call the latter the \textit{fair communication} assumption. Note that fair communication provides no bound on the channel communication delays. It merely says that a message is received within some finite time if its sender does not stop sending it (until it receives the acknowledgment message). 

\paragraph{Node failure.~~}
We assume that the failure of node $p_i \in \sP$ implies that its hosted client and server stop sending and receiving messages (and it also stops executing all other steps). We assume that the number of nodes that host servers and fail is bounded by $f$ and that $2f<N$ for the sake of guaranteeing correctness~\cite{DBLP:books/mk/Lynch96}. We bound only by $N$ the number of nodes that host clients and fail. Moreover, nodes that host servers resume within some unknown finite time and reset the server state machines by removing all stored records. However, nodes that host clients do not resume (or allow the invocation) any client process until they call a procedure that we name $\mathit{localReset}()$. (We specify how a global reset mechanism uses this procedure in Section~\ref{s:bld}. Moreover, Section~\ref{sec:incarnationNumber} provides an elegant extension that lets nodes to recycle their client identifiers and thus the above assumption is not restrictive.) 

\subsubsection{Transient faults}
We consider arbitrary violations of the assumptions according to which the system and the communication network design to operate. 
We refer to these violations and deviations as \textit{transient faults} and assume that they can corrupt the system state arbitrarily (while keeping the program code intact). We preserve the occurrence of transient faults as an extremely rare event. Our model assumes that the last transient fault occurred before the system execution started. Moreover, it left the system to start in an arbitrary state (while keeping the program code intact).

\subsection{Dijkstra's self-stabilization criterion}
\label{sec:Dijkstra}
An algorithm is \textit{self-stabilizing} with relation to the task $LE$, when every (unbounded) execution $R$ of the algorithm reaches within a bounded period a suffix $R_{legal} \in LE$ that is legal. That is, Dijkstra~\cite{DBLP:journals/cacm/Dijkstra74} requires that $\forall R:\exists R': R=R' \circ R_{legal} \land R_{legal} \in LE$, where the length of  $R'$ is polynomial in $n$. 
We say that a system \textit{execution is fair} when every step that is applicable infinitely often is executed infinitely often and fair communication is kept. Self-stabilizing algorithms often assume that $R$ is a fair execution. Wait-free algorithms guarantee that non-failing operations always become (within a finite number of steps) complete even in the presence of benign failures. Note that fair executions do not consider fail-stop failures (that were not detected by the system whom then excluded these failing nodes from reentering the system). Therefore, we cannot demonstrate that an algorithm is wait-free by assuming that the system execution is always fair. 

\subsection{Self-stabilization in the presence of seldom fairness}
As a variation of Dijkstra's self-stabilization criterion, we propose design criteria in which (i) any execution $R=R_{recoveryPeriod}\circ R': R' \in LE$, which starts in an arbitrary execution and has a prefix ($R_{recoveryPeriod}$) that is fair, reaches a legitimate system state within a bounded prefix $R_{recoveryPeriod}$. (Note that the legal suffix $R'$ is not required to be fair.) Moreover, (ii) any execution $R=R'' \circ R_{globalReset}\circ R''' \circ R_{globalReset}\circ \ldots: R'',R''' \in LE$ in which the prefix of $R$ is legal, and not necessarily fair but includes at most $\bigO(N \cdot z_{\max} )$ (Section~\ref{sec:resetServices}) write operations, has a suffix, $R_{globalReset}\circ R''' \circ R_{globalReset}\circ \ldots$, such that $R_{globalReset}$ is required to be fair and bounded in length but might permit the violation of liveness requirements, i.e., a bounded number of operations might be aborted (as long as the safety requirement holds). Furthermore, $R'''$ is legal and not necessarily fair but includes at least $z_{\max}$ write operations before the system reaches another $R_{globalReset}$. Since we can choose $z_{\max}\in \mathbb{Z}^+$ to be a very large value, say $2^{64}$, and the occurrence of transient faults is very rare, we refer to the proposed criteria as one for self-stabilizing systems that their executions fairness is unrequited except for seldom periods. Next, we define how we bound the length of $R_{recoveryPeriod}$ and $R_{globalReset}$, which are the complexity measures.

\subsection{Complexity Measures}
\label{sec:timeComplexity}
The main complexity measure of self-stabilizing systems is the time it takes the system to recover after the occurrence of a last transient fault. In detail, in the presence of seldom fairness this complexity measure considers the maximum of two values: (i) the maximum length of $R_{recoveryPeriod}$, which is the period during which the system recovers after the occurrence of transient failures, and (ii) the maximum length of $R_{globalReset}$. We consider systems that use of bounded memory and thus this is a secondary complexity measure we bound the memory that each node needs to have (after considering a version that does not consider such bounds, as in Cadambe et al.~\cite{DBLP:journals/dc/CadambeLMM17}). However, the number of messages sent during an execution does not have an immediate relevance in the context of self-stabilization, because self-stabilizing systems never stop sending messages~\cite[Chapter 3.3]{Dolev:2000}. Next, we present the definitions, notations and assumptions related to the main complexity measure. 

\subsubsection{Message round-trips}
\label{sec:communicationFairness}
Let $c \in R$ be a state, such that immediately after $c$, node $p_i$ sends a message $m$ to $p_j$. Moreover, immediately after $c'$ (that follows $c$), $p_j$ receives message $m$ (or a message that was sent from $p_i$ to $p_j$ after $m$) and sends a response message $r_m$ back to $p_i$. Then, immediately after state $c'' \in R$ (that appears after $c'$ in $R$), $p_i$ receives $p_j$'s response, $r_m$ (or a response that was sent from $p_j$ to $p_i$ after $r_m$). If $c$, $c'$ and $c''$ do appear in $R$, we say that $p_i$ has completed with $p_j$ a round-trip of message $m$.

\subsubsection{Completing client rounds}
A call to a client procedure results in a number of requests that the client sends to all servers and then waits for the server responses. The client may decide not to wait for responses from all servers and continue to the request for the next phase or reach the end of the procedure execution. We say that a client starts a new round when, after a finite period of internal computation, it sends the first request (of any phase) to the servers. Moreover,  this client ends this round when it finished waiting for the server responses (and perhaps also reaches the procedure end; regardless of whether it enters branches). The client at node $p_i$ performs a complete round when it starts a new round in $c_{start} \in R$ and ends it in $c_{end} \in R$. 

We are also interested in the cases of incomplete operations, which do not have necessarily a proper start to their first round. In this case, we say that the client at node $p_i$ completes a round when it reaches $c_{end} \in R$. Note that whenever $p_i$ does not fail, $c_{end}$ is well-defined, because it refers to the case in which $p_i$ stops waiting for the server responses and move on (to the next phase or reaching the end of the client procedure). For the case in which $p_i$ fails, we define $c_{end}$ to be the system state that immediately follows the step in which $p_i$ fails.  

\subsubsection{Complete node and server iterations}
\label{ss:completeIteration}
Recall the fact that self-stabilizing algorithms can never stop communicating~\cite[Chapter 3.3]{Dolev:2000}. The program of a self-stabilizing algorithm often includes a do-forever loop or, as in case of the proposed algorithm, a repeated gossip exchange among the servers. Next, we define the terms complete iterations, which refers to such gossip exchanges.

\paragraph{Node complete iterations.~~}
Let $P(i) \subseteq P$ be the set of nodes with whom $p_i$ completes a message round trip infinitely often in $R$.
Suppose that immediately after the system state $c_{start} \in R$, node $p_i$ takes a step that includes the execution of the first line of the do forever loop (or of the gossip procedure), and immediately after system state $c_{end} \in R$, it holds that: (i) $p_i$ had finished the iteration that it had started in $c_{begin}$ (regardless of whether it enters branches), and (ii) the message $m_j$ completes its round trip, where $m_j$ refer to any message that $p_i$ sends during that iteration to node $p_j \in P(i)$. In this case, we say that $p_i$'s iteration starts at $c_{begin}$ and ends at $c_{end}$. 
 
\paragraph{Server complete iterations.~~}
The servers repeatedly receive messages from all other non-failing servers and then, after some internal processing, send messages to all other servers. The successful arrival of such a message to any server results again in some internal processing and then sending messages to all other servers. We say that: (i) the iteration of the server at $p_i \in \sP$ starts when $p_i$ first gets a message from another server, (ii) after some internal processing, $p_i$ sends a message to every other server at $p_j$, (iii) this iteration continues toward letting $p_j$ to receive that message (or a later message) from $p_i$; at least once, and then (iv) letting $p_j$'s responds (or a later message from $p_j$) to arrive to $p_i$ and by that ending this server iteration. Given an execution $R$, we say that its prefix $R'$ includes a complete iteration of the server at $p_i \in \sP$ if $R'$ includes $p_i$'s iteration start and then (after that start) and $p_i$'s iteration end appears in $R'$.

\subsubsection{Asynchronous cycles}
\label{ss:asynchronousCycles}
We measure the time between two system states in a fair execution by the number of (asynchronous) cycles between them. The definition of (asynchronous) cycles considers the term of complete iterations.  
The first (asynchronous) cycle (with round-trips) of fair execution $R=R'' \circ R'''$ is the shortest prefix $R''$ of $R$, such that each non-failing node and server in the network executes at least one complete iteration in $R''$, where $\circ$ is the concatenation operator (Section~\ref{sec:interModel}). Moreover, each node that runs a client procedure during $R''$ must complete within $R''$ at least one client round. The second cycle in execution $R$ is the first cycle in execution $R''$, and so on.

\section{External Building Blocks}
\label{s:exbld}

The proposed algorithm uses a number of external building blocks, which we specify next.   

\subsection{Specifications of gossip and quorum services}
\label{sec:basic}
\label{sec:coomSafe}
We consider a gossip functionality that has the following interface. Servers can send gossip messages $msg$ by calling $\text{gossip}(msg)$. When a gossip message arrives, the receiving server raises the gossip arrival event with a set $\{gossip[k]\}_{p_k \in \sP}$ that includes the most recently received message from every server. For the sake of simple presentation, we allow the server at node $p_i$ to use the item $gossip[i]$ for aggregating the gossip information that it later gossips to all other servers. The gossip functionality that we consider guarantees the following: (a) every gossip message that the receiver delivers to its upper layer was indeed sent by the sender, and (b) such deliveries occur according to the communication fairness guarantees (Section~\ref{sec:benignFailures}). That is, our gossip service is unreliable, as opposed to the one used by Cadambe et al.~\cite{DBLP:journals/dc/CadambeLMM17}. 

We consider a system in which the clients and servers behave according to the following terms of service. At any time, any node runs only at most one client (that is either a writer or a reader). That client calls the function $\text{qrmAccess}()$ sequentially. Moreover, the server algorithm acknowledges (by calling $\text{reply}()$) every request. For this client-server behavior, the quorum-based communication functionality guarantees the following. (a) At least a quorum of servers receive, deliver and acknowledge every request. (b) The (non-failing) requesting client receives at least a quorum of these acknowledgments. (c) Immediately before the call to $\text{qrmAccess}()$ returns, the client-side of this service clears its state from information related to the request. For the sake of simple presentation, we allow the client to call $\text{qrmAccess}(msg)$ with two kinds of parameters; $msg$ is either a single message to be sent to all servers, such as in the case of a query request, or a vector that includes an individual message for each server, such as in the case of a prewrite request.  

We detail the above requirements in Definition~\ref{def:coomSafe} and use Corollary~\ref{thm:basicCorollary} in the correctness proof of the proposed algorithm. (Section~\ref{sec:basic:App} of the Appendix presents a self-stabilizing implementation of such services, which we do not include here since it is not a major contribution.)
 
\begin{definition}[\textbf{Legal execution of the gossip and quorum services}]
	\label{def:coomSafe}
	Let $R$ be an execution of the algorithm that provides gossip and quorum services in which there is a client at $p_i \in \sP$, a server at $p_j \in \sP$ and another server at $p_k \in \sP$. 
	
	\begin{itemize} 
		\item \textbf{\emph{Correct behavior of the gossip functionality.~~}} Suppose that (1) every message that $p_k$ delivers to the upper layer as a gossip from $p_j$ was indeed sent by $p_j$ earlier in $R$. Moreover, (2) such deliveries occur infinitely often in $R$. In this case, we say that the behavior of the gossip functionality from the server at $p_j$ to the one at $p_k$ is correct.
		
		\item \textbf{\emph{Terms of service for the quorum-based communication functionality.~~}}
		Suppose that in $R$, at any time, any node runs only at most one client (that is either a writer or a reader). Moreover, that client calls the function $\text{qrmAccess}()$ sequentially, i.e., only after the return from $\text{qrmAccess}()$ may the client call $\text{qrmAccess}()$ again. Furthermore, suppose that the server algorithm acknowledges (by calling $\text{reply}()$) every request that was delivered to it. In this case, we say that $R$ satisfies the terms of service of the quorum-based communication functionality.
		
		\item \textbf{\emph{Correct behavior of the quorum-based communication functionality.~~}}
		Suppose that the client at $p_i$ sends a request, i.e., $p_i$ calls the function $\text{qrmAccess}()$ in step $a_{qrmAccess} \in R$.
		Moreover, after $a_{qrmAccess}$, execution $R$ includes steps (i) to (v), where (i) refers to the steps in $R$ in which at least a quorum of servers receive $a_{qrmAccess}$'s request, (ii) refers to steps in $R$ in which at least a quorum delivers $a_{qrmAccess}$'s request, (iii) refers to steps in $R$ in which at least a quorum acknowledges $a_{qrmAccess}$'s request and (iv) refers to steps in $R$ in which the client at $p_i$ receives at least a quorum of these acknowledgments to $a_{qrmAccess}$'s request, which results in (v) a step in $R$ in which $p_i$ lets the function (which $p_i$ had previously called in step $a_{qrmAccess}$) to return. Furthermore, any such return is only the result of the above sequence of steps (i) to (v). In this case, we say that the functionality of quorum-based communication is correct. 
		In addition, immediately before the call to $\text{qrmAccess}()$ returns, the client-side of this service clear its state from information related to the request.
		
		\item \textbf{\emph{A legal execution of gossip and quorum services}.~~} 
		Let $R'$ and $R''$ be a prefix, and respectively, a suffix of $R$, such that $R=R' \circ R''$ is an execution of gossip and quorum services that satisfies the terms of service of the quorum-based communication functionality. We say that $R''$ is legal when it presents: (1) a correct gossip functionality from the server at $p_j$ to the server at $p_k$, and (2) a correct functionality of quorum-based communication with respect to the client at $p_j$.
	\end{itemize}
\end{definition}
 
\begin{corollary}[Self-stabilizing gossip and quorum-based communications]
\label{thm:basicCorollary}
Let $R$ be an Algorithm~\ref{alg:comm}'s (unbounded) execution that satisfies the terms of service of the quorum-based communication functionality. Suppose that $R$ is fair and its starting system state is arbitrary. Within $\bigO(1)$ asynchronous cycles, $R$ reaches a suffix $R'$ in which 
\textbf{\emph{(1)}} the gossip, and 
\textbf{\emph{(2)}} the quorum-based communication functionalities are correct.
\textbf{\emph{(3)}} During $R'$, the gossip and quorum-based communication complete correctly their operations within $\bigO(1)$ asynchronous cycles.
\end{corollary}

\subsection{Self-stabilizing Global Reset}
\label{s:bld}
The proposed algorithm uses the reset mechanism for dealing with the case in which the system includes a tag of $(z_{\max},j):p_j \in \sP$ (Section~\ref{sec:resetServices}). We note that the reset mechanism requires the participation of all the nodes in the network, i.e., they require execution fairness (Section~\ref{sec:Dijkstra}). We specify the interface between the proposed algorithm and the self-stabilizing global reset mechanism.

\subsubsection{The \textit{localReset}() and \textit{globalReset}() functions}
During an execution that is legal (with respect to the reset mechanism), the self-stabilizing global reset process starts when any node, which we refer to as the \emph{(reset) initiator}, calls the $\mathit{globalReset}(t)$ function; concurrent calls are allowed. The reset mechanism lets every pair of nodes to exchange messages infinitely often so that it can make sure that all nodes complete the different phases of the reset process, which starts immediately after the first call to $\mathit{globalReset}()$. In the first phase, all client and server processes are disabled and each node calls the function $\mathit{localReset}()$. In the second phase, these processes are enabled, the reset process ends and the system resumes normal operation.

We assume that every machine, such as a server or a client, implements the function {\textit{localReset}}($t$). For the case of the servers, this local reset procedure removes any record from the server storage other than the ones with the tag $t$ and then replaces the tag $t=(z,k)$ in that record with the tag $(1,k)$. Note that when $t=t_0$ (Section~\ref{ss:emulatedObject}), no record is kept in the server storage. For the case clients, the call to ${\mathit{localReset}}(t)$ simply stops any client operation and the ignores the argument $t$.
The requirements below specify the set of legal executions. Note that the system has to reach a safe system state even when no global reset was (properly) initialized, e.g., no node has called $\mathit{globalReset}()$, but still, some nodes are performing reset due to transient faults. 

\subsubsection{Requirements}
\label{ss:BBBreq}
Within a bounded number of asynchronous cycles from the first step in $R$ that includes a call to ${\mathit{globalReset}}(t)$, the reset service disables all hosted processes, which are the servers and clients, and resets these processes by calling their ${\mathit{localReset}}(t)$ functions, which abort all read and write operations. Moreover, every node cleans its incoming and outgoing channels, e.g., it fills these channels with reset messages so that non-reset-related messages are absent from these channels, as in~\cite{Dolev:2000} Chapter 3.2. (By reset messages we mean messages that their type is only used by the reset mechanism.) Then, the reset mechanism enables every (local) machine. 

We further require the following.
We say that a given system state is \emph{reset-free} when all communication channels do not include reset-related messages, and all machines (clients and servers) are enabled. Given execution $R$ of the system, we say that $R$ \emph{does not include an explicit reset} when throughout $R$ no node $p_i \in \sP$ calls $\mathit{globalReset}()$. Suppose that execution $R$ does not include an explicit reset and that all of its system states are reset-free. In this case we say that $R$ \emph{does not include a spontaneous reset}. An execution $R$ that does not contain neither a spontaneous reset, nor an explicit reset, is \emph{reset-free}. When execution $R$ does include (a spontaneous or an explicit) reset, we require $R$ to \emph{be done with reset} within a bounded number of $\Psi$ asynchronous cycles.  Namely, (starting from an arbitrary system state) within $\Psi$ asynchronous cycles, the system reaches a system state after which the execution is reset-free.

\subsubsection{Possible implementations}
The proposed algorithm uses a self-stabilizing global reset mechanism that resembles the one in~\cite{DBLP:conf/wdag/AwerbuchPVD94}. Another way to go is to use a self-stabilizing consensus algorithm~\cite{DBLP:conf/netys/BlanchardDBD14}. Since similar mechanisms exist and they are not hard to extend so that the above specifications are met, we do not consider the algorithm for implementing the specified mechanism for self-stabilizing global reset to be within the scope of this work.

\section{An Unbounded Self-stabilizing CAS Algorithm}
\label{sec:unbAlg}
%
\begin{algorithm*}[t!]
\algoSize
\textbf{The client:~}\tcc{At any time, $p_i$'s client is a writer, a reader or none but not both}
\textbf{$writer(s)$:~}\tcc{Store the secret $s$ as a new version of the shared object}
\tcc{Query for finalized tags and after hearing \text{from} a quorum get the maximal tag}
\textbf{let} $(z,j) := \max (\{t' : ( t', \bullet ) \in \text{qrmAccess}((\bot,\bot,\text{`qry'}))\})$\nllabel{alg:w:SelMax}

$\text{qrmAccess}(((z+ 1,i),\{\Phi_{p_j}(s)\}_{p_j\in \sP},\text{`pre'}))$\tcc*{Prewrite and wait for a quorum of replies}\nllabel{ln:preWrite}
$\text{qrmAccess}(((z+ 1,i),\bot,\text{`fin'}))$\nllabel{ln:finalize}\tcc*{ Finalize and wait for a quorum of replies }
$\text{qrmAccess}(((z+ 1,i),\bot,\text{`FIN'}))$\nllabel{ln:FINALIZE}\tcc*{ FINALIZE and wait for a quorum of replies }
$\Return$\nllabel{alg:w:return}\;

\BlankLine

\textbf{$reader()$:~}\tcc*{The reader retrieves the current object version, or $\bot$ upon failure}
\textbf{let} $t := \max (\{t' : (t', \bullet ) \in \text{qrmAccess}((\bot,\bot,\text{`qry'}))\})$\nllabel{alg:r:SelMax}~\tcc*{Query as in line~\ref{alg:r:SelMax}}
\textbf{let} $Q:=\text{qrmAccess}((t,\bot,\text{`fin'}))$\nllabel{ln:finalizeR}\tcc*{ Ask and wait for finalized records from a quorum}
\lIf( \textbf{/*} \texttt{Test the number of responces} \textbf{*/}){$|\{(t,w,\text{`fin'})\in Q:w\neq\bot\}|\ngeq k_{threshold}$}{\Return $\bot$}\nllabel{alg:read:BWfail}
\lElse( \textbf{/*} \texttt{Use the retrived shares for decoding} \textbf{*/}){$\Return(\Phi^{-1}(w:\{(t,w,s)\in Q:w\neq\bot\}))$\nllabel{alg:read:BW}}

\BlankLine

\textbf{The server:}\\ 
$S\subset\sT\times(\sW\cup\{\bot\})\times \sD$ is a record set, where $\sT$ $=$ $\sZ \times P$ is the set of tags, $\sW$ the set of coded words and $\sD=\{\text{`pre'}, \text{`fin'}, \text{`FIN'} \}$ the set of phases. When $S=\emptyset$, we use the default triple $(t_0,w_{t_0,i}, \text{`fin'})$ when reporting on the triple with the highest locally known tag\nllabel{alg:srv:default}\;

\BlankLine

\noindent \textbf{Event handlers at the server:}\\
\Upon query arrival from $p_j$'s client to $p_i$'s server \Do{\label{alg:srv:queryInit}
      \lIf{$p_j$'s client is a reader}{$\text{reply}(j, (maxPhase(\sD \setminus \{\text{`pre'}\}), \bot, \text{`qry'}))$\label{alg:srv:queryInitIf}}
      \lElse{$\text{reply}(j, (maxPhase(\sD), \bot, \text{`qry'}))$\label{alg:srv:queryInitElse}}
}

   \Upon pre-write $(t, w, \text{`pre'})$ arrival from the $p_j$'s writer to $p_i$'s server 
   \Do{\label{alg:srv:pre-write}
      $updatePhase(t, w, \text{`pre'})$\label{alg:srv:pre-writeUp}\;
      $\text{reply}(j, (t, \bot, \text{`pre'}))$\label{alg:srv:pre-writeReply}\;
   }

\Upon finalize or FINALIZE $m:=( t, \bot, d):d\in (\sD \setminus \{\text{`pre'}\})$ arrival {from} $p_j$'s client to $p_i$'s server 
\label{alg:srv:fin}
\Do{

$updatePhase(t, \bot, d)$\nllabel{alg:srv:updatePhasetFin}\;
\lIf{$\exists s:=(t,w,d)\in S$ and $p_j$'s client is a reader}{${reply}(j,(t, w, d))$ \textbf{else} ${reply}(j,(t, \bot, d))$\label{alg:srv:finReply}}
}   

\Upon gossip $\{(pre[k],fin[k],FIN[k])=gossip[k]\}_{p_k \in \sP}$ arrival {from} $p_j$'s server to $p_i$'s server\label{alg:srv:uponGossip} \Do{
	$pre[i] \gets \max (\{pre[k],fin[k],FIN[k]\}_{p_k \in \sP} \cup \{maxPhase(\sD)\})$\nllabel{alg:srv:preMax}\; 
	$updatePhase(pre[i], \bot, \text{`pre'})$\nllabel{alg:srv:preMaxUpdt}\;
	$fin[i] \gets \max (\{fin[k],FIN[k]\}_{p_k \in \sP} \cup \{maxPhase(\sD \setminus \{\text{`pre'}\})\})$\nllabel{alg:srv:finMax}\; 
	$updatePhase(fin[i], \bot, \text{`fin'})$\nllabel{alg:srv:finMaxUpdt}\;
	$FIN[i] \gets \max (\{FIN[k]\}_{p_k \in \sP} \cup \{maxPhase(\{ \text{`FIN'}\})\}) \cup \{ t \in \sT: \{p_k \in \sP:fin[k]=t\}\in \sQ\}$\nllabel{alg:srv:FINMax}\;
	$updatePhase(FIN[i], \bot,  \text{`FIN'})$\nllabel{alg:srv:FINMaxUpdt}\;
	$\text{gossip}(tagTuple())$\nllabel{alg:srv:FINMaxGossipHelper}\;
}

\BlankLine

\noindent \textbf{Local functions at the server:}\\
\textbf{function} $maxPhase(phs)$ \lDo{~\textbf{return $\max(\{t : (t, \bullet, p) \in (S \cup \{(t_0,w_{0,i},\text{`fin'})\}) \land p \in phs \})$}\label{alg:srv:maxPhase}}

\textbf{function} $tagTuple()$ \lDo{~\textbf{return} $(maxPhase(\sD), maxPhase(\sD \setminus \{\text{`pre'}\}), maxPhase(\{ \text{`FIN'}\}))$\;\label{alg:srv:gossip}} 

\textbf{function} $updatePhase(t,w,u)$ \lDo{
\label{alg:srv:helper}
   \{\lIf {$\exists s := (t,w',c) \in S \land w' \neq \bot \land w = \bot$}{$S\gets (S \setminus \{s\}) \cup$ $\{(t,w',p)\})$, \textbf{where} $p:=upgradePhase(c,u)$ \textbf{else} $S\gets((S \setminus \{(t, \bullet )\})\cup\{(t, w, u)\})$\}\label{alg:srv:help_addEmpty}}}
   
\textbf{function} $upgradePhase(old, new)$~\lDo{~$\mathbf{switch}$ $(old, new):$\\ ~~~$\mathbf{case}(\text{`pre'},\text{`fin'})$: $\mathbf{return}~\text{`fin'}$; $\mathbf{case}(\text{`fin'}, \text{`FIN'})$: $\mathbf{return}~ \text{`FIN'}$; $\mathbf{default}$ $\mathbf{return}~old$;\label{alg:srv:upcu}}

\caption{\algoSize Private and Self-Stabilization CAS Algorithm, code for $p_i$'s client and server.}
\label{alg:cas}
\end{algorithm*}

This paper presents a self-stabilizing algorithm that uses a bounded amount of memory. For the sake of presentation simplicity, we start by presenting a self-stabilizing algorithm that has no such bounds as a `first attempt'. We then prove the correctness of the unbounded algorithm (Section~\ref{sec:proof}) before bounding the amount of storage needed (Section~\ref{sec:bound}) and as well as the number of possible tag values (Section~\ref{sec:bounded}).  

One of the key differences between self-stabilizing algorithms to non-self-stabilizing algorithms is that, due to a transient fault, the self-stabilizing system can start in a state $c$ that the non-self-stabilizing system can never reach. For example, in $c$ a single server may include a record with a finalized tag $t$ for which there is no quorum of servers that store records that include coded elements relevant to $t$. Due to the asynchronous nature of the system, we cannot bound the number of write operations that the system will take until at least one write operation install its records on all servers. Similar examples can be found when considering pre-write  records. We propose to overcome this challenge by letting the gossip server to exchanges message that includes that maximal tag values for each phase.

There is no self-stabilizing algorithm for end-to-end communication when there are no bound on the capacity of the communication channels~\citep[Chapter 3]{Dolev:2000}. Cadambe et al.~\cite{DBLP:journals/dc/CadambeLMM17} assume that all communication channels are reliable and cannot lose messages. We are not aware of a straightforward manner in which we can assume that these communication channels are both of bounded capacity and self-stabilizing, because the asynchronous nature of the system implies that there is no bound on the number of write operations that the system may finish before a given server receives a single gossip message. Therefore, we let the gossip service to repeatedly exchange among the servers their maximal tag values.  This way, the servers get to know eventually about the highest tag values.        

With these modifications in mind, we note that the client part of Algorithm~\ref{alg:cas} follows similar lines as the ones of Algorithm~\ref{BACKalg:cas} with the following notable differences. The prewrite phase (line~\ref{ln:preWrite}) associates the operation with the tag $(z+ 1,i)$, where $t=(z, \bullet)$ is the maximal prewrite tag returned from the query phase (line~\ref{alg:w:SelMax}). Moreover, Algorithm~\ref{alg:cas} uses an additional finalized phase (line~\ref{ln:FINALIZE}), which we refer to as FINALIZED. This phase helps the algorithm to assure that every complete write operation with tag $t$ has at least a quorum of servers with a finalized tag $t$.     

The server part of Algorithm~\ref{alg:cas} also implements the above modifications. The notable changes, with respect to Algorithm~\ref{BACKalg:cas}, include the following. Servers reply to queries from readers with the highest local finalized tag (line~\ref{alg:srv:queryInitElse}) whereas for the case of writers all local tags are considered (line~\ref{alg:srv:queryInitIf}). Also, upon gossip arrival (line~\ref{alg:srv:uponGossip}), the server processes all the gossip messages that have recently arrived for all servers. It first calculates the local maximal prewrite tag (lines~\ref{alg:srv:preMax} and~\ref{alg:srv:preMaxUpdt}), then the local maximal finalized tag (lines~\ref{alg:srv:finMax} and~\ref{alg:srv:finMaxUpdt}) before considering the FINALIZED one (lines~\ref{alg:srv:FINMax} and~\ref{alg:srv:FINMaxUpdt}) and then sending an updated gossip message (line~\ref{alg:srv:FINMaxGossipHelper}). Note that each server stores the highest tag that it has heard from each phase. Moreover, when a server discovers that it knows about a quorum of servers that store each a finalized record with tag $t$, it updates that record to have the FINALIZED phase (line~\ref{alg:srv:FINMax}). This way, an implicit FINALIZED record becomes explicitly FINALIZED.  

\section{Correctness Proof of our Self-stabilizing CAS Algorithm}
\label{sec:proof}
After the preliminaries (Section~\ref{sec:def}), we study basic properties of Algorithm~\ref{alg:cas} (Section~\ref{sec:basic}) before showing its ability to recover after the occurrence of transient-faults (Section~\ref{sec:convergence}). We then demonstrate the atomicity (Section~\ref{sec:atom}) and liveness (Section~\ref{sec:live}) of Algorithm~\ref{alg:cas}.
 
\subsection{Notation and definitions}
\label{sec:def}
We refer to the values of variable $X$ at node $p_i$ as $X_i$, i.e., the variable name with a subscript that indicates the node identifier. We denote to the storage variable, $S$, of $p_i$ as $S_{p_i}$ due to its centrality to the system state. Let $R$ be an execution, $c \in R$ a system state and $p_i \in \sP$ a node that executes the function $f()$ in a step $a \in R$ that appears in $R$ immediately after $c$. We denote by $f_a()$ the value that returns from $f()$'s execution during step $a$.

Each client procedure includes a finite sequence of requests that the client sends to the servers, where the responses received from one request to the servers are used for forming the next request to the servers. We associate each invocation of the client procedures with an operation $\pi$ that includes all of its steps in $R$ (either at the client or the servers) in which a node sends or receives messages due to $\pi$'s invocation of the client procedure. 
Definition~\ref{def:startEndSys} classifies operations by the way they start and end.

\begin{definition}[\textbf{Classifying operations by their start and end}]
\label{def:startEndSys}
Let $R$ be the algorithm execution with $\pi$ as a client operation. Denote by $c_{start}(\pi) \in R$ the system state that is followed immediately by step $a_{start}(\pi)$ that starts $\pi$. Moreover, $c_{end}(\pi)$ denotes the system state that follows immediately after a step $a_{end}(\pi)$ that ends $\pi$. We characterize $\pi$'s behavior in $R$ in the following manner.

\begin{itemize}
\item \textbf{\emph{Incomplete operations.}~~} Suppose that $\pi$'s first step, $a_{start}(\pi)$, does not include the execution of the first line of the $\pi$'s (write or read) procedure (lines~\ref{alg:w:SelMax}, and respectively,~\ref{alg:r:SelMax}). In this case, we say that $\pi$ is \emph{incomplete} in $R$. We say that a client request or a server reply is incomplete if it is part of an incomplete operation (due to stale information that appear in an arbitrary starting system state). Note that the expression ``operation $\pi$ is complete in suffix $R''$ of $R=R'\circ R''$ refers to the case in which $a_{start}(\pi) \in R'$ and $a_{end}(\pi) \in R''$.

\item \textbf{\emph{Failed operations.}~~} Suppose that $\pi$'s last step, $a_{end}(\pi)$, does not include the execution of the last line of the $\pi$'s (read or write) procedure. In this case, we say that $\pi$ \emph{fails} in $R$.

\item \textbf{\emph{Complete operations.}~~} Suppose that $\pi$ is eventually neither incomplete nor failed in $R$. In this case, we say that $\pi$ is \emph{complete} in $R$. Before the end of a given complete operation, we refer to it as an \emph{on-going} operation.
\end{itemize}
\end{definition}

\subsection{Basic Properties of Algorithm~\ref{alg:cas}}
\label{sec:basic}

\begin{lemma}[\textbf{Algorithm~\ref{alg:cas}'s progression}] 
\label{thm:progression}
Algorithm~\ref{alg:cas}'s operations, whether they are failed, incomplete or complete, end within $\bigO(1)$ asynchronous cycles in any fair execution of Algorithm~\ref{alg:cas}, which may start in any system state.
\end{lemma}

\begin{proof}
The server part of Algorithm~\ref{alg:cas} includes only non-blocking responses to Algorithm~\ref{alg:cas}'s requests. Thus, Algorithm~\ref{alg:cas}'s termination depends on the termination of each client phase. We, therefore, prove that every client phase ends eventually.

\noindent \textbf{The query, pre-write and finalize phases (of readers and writers) terminate.~~} 
We start by showing that at least a quorum of query responses arrive to every non-failing client (lines~\ref{alg:w:SelMax} and~\ref{alg:r:SelMax}). The proof uses parts (2) and (3) of Corollary~\ref{thm:basicCorollary} for showing the correct functionality of quorum-based communication in $R$ (Section~\ref{sec:coomSafe}) as long as the system satisfies the terms of service of the quorum-based communication functionality. To that end, we need to show that: (i) only the client calls the function $\text{qrmAccess}()$ and it does so sequentially, as well as (ii) the server algorithm acknowledges (by calling $\text{reply}()$, Section~\ref{sec:coomSafe}) requests that were delivered to it. From Section~\ref{ss:emulatedObject} and Algorithm~\ref{alg:cas}, we observe that: (i.a) any node run at most one client that is either a writer or a reader, (i.b) only the clients call the function $\text{qrmAccess}()$, (i.c) there is only one client (either reader or writer) per node, and (i.d) that client does not calls $\text{qrmAccess}()$ before the previous call returns. We also note that (ii) the server pseudo-code (Algorithm~\ref{alg:cas}) includes a response for every client request. In detail, any non-failing server, say, the one at node $p_j \in \sP$, replies to queries that $p_i$ delivers to it with the message $((\bot, \bot, \text{`qry'}), (t, \bot, \text{`qry'}))$ (line~\ref{alg:srv:queryInit}), where $(t,\bullet,\text{`fin'}) \in S_{p_i}$ and $t$ is $S_{p_i}$'s highest finalized local tag. Note that whenever $S_{p_i}=\emptyset$, the server at $p_i$ considers the tuple $(t_0, \bot, \text{`qry'})$ (line~\ref{alg:srv:default}), and thus the server at $p_j$ always replies to queries. From (i.a), (i.b), (i.c), (i.d)  and (ii) we get that the functionality of quorum-based communication is correct eventually even when starting from any system state, because the conditions of parts (2) and (3) in Corollary~\ref{thm:basicCorollary} hold (terms of service for the quorum-based communication functionality, Section~\ref{sec:coomSafe}). Therefore, the client at $p_i$ receives eventually at least a quorum of server responses (Section~\ref{sec:coomSafe}). Using the same arguments as above, this proof shows that at least a quorum of pre-write and finalize responses arrive to every non-failing client (lines~\ref{ln:preWrite},~\ref{ln:finalize}~\ref{ln:FINALIZE} and~\ref{ln:finalizeR}). 

Finally, we show that the above happens within $\bigO(1)$ asynchronous cycles. This is because each client operation considers a constant number of phases. Recall that each phase is associated with a client round, which completes within $\bigO(1)$ asynchronous cycles (Corollary~\ref{thm:basicCorollary}).
\end{proof}

\begin{definition}[\textbf{Classifying local maximal tags by their phase}]
\label{def:startEnd}
Let $R$ be an Algorithm~\ref{alg:cas}'s execution, and $a_{i,k} \in R$ a step (that the server at $p_i \in \sP$ takes) in which $p_i$ executes the function $maxPhase(phs)$ (line~\ref{alg:srv:maxPhase}) for the $k$-th time in $R$. We characterize $maxPhase_{a_{i,k}}(phs)$'s behavior in $R$ according to its argument $phs$ (Lemma~\ref{thm:noRemoval}) and consider the set $\emph{tags}(C,D)=\{ t:(t,\bullet, d) \in (S_{p_j} \cup \{(t_0,w_{0,i},\text{`fin'})\}) \land d \in D \land p_j \in C\}$, where $D=phs$ and $C=\{p_i\}$ in the system state that immediately precedes $a_{i,k}$. 

\begin{itemize}
\item A \emph{write maximal tag} is the returned value from ${maxPhase_{a_{i,k}}(\sD)}$ (lines~\ref{alg:srv:queryInitElse} and~\ref{alg:srv:preMax}), where $\sD=\{\text{`pre'}, \text{`fin'}, \text{`FIN'} \}$, i.e., the tag in the maximal tuple in $\emph{tags}(\{p_i\},\{\text{`pre'}, \text{`fin'}, \text{`FIN'} \})$. 
\item A \emph{read maximal tag} is the returned value from ${maxPhase_{a_{i,k}}(\sD \setminus \{\text{`pre'}\})}$ (lines~\ref{alg:srv:queryInitIf} and~\ref{alg:srv:finMax}), i.e., the tag in the maximal tuple in $\emph{tags}(\{p_i\},\{\text{`fin'}, \text{`FIN'} \})$.
\item An \emph{anchor maximal tag} is the returned value from ${maxPhase_{a_{i,k}}(\{ \text{`FIN'}\})}$ (line~\ref{alg:srv:FINMax}), i.e., the tag in the maximal tuple in $\emph{tags}(\{p_i\},\{ \text{`FIN'}\}))$.
\end{itemize}
\end{definition}
  
\begin{lemma}[\textbf{Servers do not remove their maximal records}]
\label{thm:noRemoval}
Servers (Algorithm~\ref{alg:cas}) keep in their storage the currently maximal (1.1) write, (1.2) read, and (1.3) anchor records (or any record with a tag that is higher than the ones in these records). 
\end{lemma}
\begin{proof}
\noindent \textbf{Part (1.1).~~}
We show that the server (Algorithm~\ref{alg:cas}) at $p_i$ does not remove from $S_{p_i}$ the maximal anchor record (Definition~\ref{alg:cas}), i.e., the tuple with the maximal tag in $\emph{tags}(S_{p_i},\sD)$, cf. Definition~\ref{def:startEnd}. Note that the server at $p_i$ updates and inserts records $(t,\bullet)$ to $S_{p_i}$ only via the $updatePhase()$ function (line~\ref{alg:srv:helper}). In case that $\exists (t,\bullet) \in S_{p_i}$, the function $updatePhase()$ calls the function $upgradePhase()$ (line~\ref{alg:srv:upcu}), which transfers $(t,\bullet)$'s phase from $\text{`pre'}$ to $\text{`fin'}$, and $\text{`fin'}$ to $\text{`FIN'}$, but otherwise it does not change $(t,\bullet)$'s phase, e.g., when $p=\text{`pre'}$. Moreover, in case $\nexists (t,\bullet) \in S_{p_i}$, the server at $p_i$ merely adds $(t,\bullet)$ to $S_{p_i}$. We study each call to $updatePhase()$ in Algorithm~\ref{alg:cas} and show that it does not remove the currently maximal write record. 

\begin{itemize}
\item When $(t, \bot, d):d\in \sD$ (lines~\ref{alg:srv:pre-write} and~\ref{alg:srv:fin}) arrives to the server at $p_i$, that server uses the $updatePhase()$ (lines~\ref{alg:srv:pre-writeUp} and~\ref{alg:srv:updatePhasetFin}) for making sure that $(t, \bullet, d)$ exists in $S_{p_i}$ (line~\ref{alg:srv:helper}) in a way that can only add a missing record $(t, \bot, d)$ to $S_{p_i}$ (when $(t, \bot, d) \notin S_{p_i}$) or transfer the phase of an existing record in $S_{p_i}$ according to $upgradePhase()$, which does not remove the currently maximal write record. Moreover, when $(t, \bot, d)$ arrives either from a reader or a writer, $p_i$ updates $S_{p_i}$ in a manner that differ only by the response that $p_i$ sends to the client (lines~\ref{alg:srv:pre-writeReply} and~\ref{alg:srv:finReply}), i.e., irrelevant to $p_i$'s server state after that send.

\item When gossip arrives to $p_i$ (line~\ref{alg:srv:uponGossip}), $p_i$ calculates its new maximal write record in a way that includes both the records in its own storage $S_{p_i}$ and the maximal records reported recreantly from all servers including itself (lines~\ref{alg:srv:preMax},~\ref{alg:srv:finMax} and~\ref{alg:srv:FINMax}). 
Note that $p_i$ might add a new maximal anchor record with tag $t$ whenever it discovers that there is a quorum of servers that have reported about a finalized record with tag $t$ (line~\ref{alg:srv:FINMax}). 
After calculating these new maximal values, $p_i$ updates its storage $S_{p_i}$ via $updatePhase()$ (lines~\ref{alg:srv:preMaxUpdt},~\ref{alg:srv:finMaxUpdt} and~\ref{alg:srv:FINMaxUpdt}), in a way that we showed above that it does not remove the currently maximal write record.
\end{itemize}

\noindent \textbf{Parts (1.2) and (1.3).~~} The proofs here follow similar arguments to the ones of Part (1.1); it is even simpler because $phs$'s values are different and thus $\text{`pre'}$ is irrelevant to Part (1.2) and only $\text{`FIN'}$ is relevant to Part (1.3). 

We note that the same arguments hold also for any record that has a tag that is higher than these maximal tags.
\end{proof}

\begin{lemma}[\textbf{Maximal tags arrive to every server eventually}]
\label{thm:eventuallyAscendingSent}
Suppose that the server at $p_i \in \sP$ calls $\text{gossip}(T_k)$ (Section~\ref{sec:coomSafe}) for an unbounded number of times in  Algorithm~\ref{alg:cas}'s execution, $R$, such that $T_k:=(t_{k, write},t_{k, read},t_{k, anchor})$ is the $k$-th gossip that $p_i$ sends. The server at $p_j \in \sP$ receives eventually at least one gossip $(t_{k, write},t_{k, read},t_{k, anchor})$, such that each respective tag is not less than its correspondent  in $(t_{1, write},t_{1, read},t_{1, anchor})$. Moreover, if $R$ is fair, each gossip message arrives within $\bigO(1)$ asynchronous cycles. 
\end{lemma} 

\begin{proof}
Lemma~\ref{thm:noRemoval}, and Claim~\ref{thm:neverRemove} facilitate the proof of Claim~\ref{thm:neverRemoveAgain}, which implies that first part of this lemma. Claim~\ref{thm:neverRemove} considers the sequence $t_{k, type}$, where $p\in \{read, write, anchor \}$ and $T_k:=(t_{k, write},t_{k, read}$, $t_{k, anchor})$.

\begin{claim}
\label{thm:neverRemove}
The sequence $t_{k, type}$ is non-decreasing.
\end{claim}
\begin{claimproof}
The three parts of Lemma~\ref{thm:noRemoval} show that the server never remove its currently maximal write, read and anchor records. The rest of the proof is implied directly from the fact that line~\ref{alg:srv:maxPhase} merely calculates the currently maximal write, read and anchor records.
\end{claimproof}

\begin{claim}
\label{thm:neverRemoveAgain}
The server at $p_j \in \sP$ receives at least one gossip $(t_{k, write},t_{k, read},t_{k, anchor})$ eventually, such that each respective tag is not less than its correspondent in $(t_{1, write},t_{1, read},t_{1, anchor})$.
\end{claim}
\begin{claimproof}
Let us consider the sequence $t_{k, type}$, where $type\in \{read, write, anchor \}$ and $T_k:=(t_{k, write}$, $t_{k, read}, t_{k, anchor})$. Let $a_{k} \in R$ be a step in which $p_i$ calls \text{gossip}($T_k$) for the $k$-th time in $R$. Let $a_{depart,k'} \in R$ be the first step that appears after $a_{k'}:k'\in \{1,\ldots\}$ and before $a_{k'+1}$ in $R$, if there is any such step, in which $p_i$ executes the event of gossip token departure. Let $a_{arrival,k'} \in R$ be the first step that appears after $a_{depart,k'}$ in $R$, if there is any such step, in which the server at $p_j$ delivers the token that $a_{depart,k'}$ transmits. By the correctness of the gossip functionality (Corollary~\ref{thm:basicCorollary}), step $a_{arrival,k'}$ exists eventually. The proof is done, because $T_{k'}$ includes only tags, $t_{k'}$, that are no less than their corresponding elements in $T_1$ (Claim~\ref{thm:neverRemove}).
\end{claimproof}

We complete this proof by considering the case in which $R$ is fair. Corollary~\ref{thm:basicCorollary}, Part (3) implies that step $a_{arrival,k'}$ exists within $\bigO(1)$ asynchronous cycles.
\end{proof}

Corollary~\ref{thm:nonDecr} considers the calls to $maxPhase(\sD)$ (lines~\ref{alg:srv:queryInitElse},~\ref{alg:srv:preMax} and~\ref{alg:srv:gossip}), to $maxPhase(\sD \setminus \{\text{`pre'}\})$ (lines~\ref{alg:srv:queryInitIf},~\ref{alg:srv:finMax} and~\ref{alg:srv:gossip}) as well as the to $maxPhase(\{ \text{`FIN'}\})$ (lines~\ref{alg:srv:FINMax}, and~\ref{alg:srv:gossip}). The same arguments as in the proof of Claim~\ref{thm:neverRemove} imply Corollary~\ref{thm:nonDecr}. 

\begin{corollary}
\label{thm:nonDecr}
Let $phs \in \{\sD, \sD \setminus \{\text{`pre'}\}, \{ \text{`FIN'}\} \}$ and $a_{i,k} \in R$ be a step in which the server at $p_i \in \sP$ executes $maxPhase(phs)$ for the $k$-th time in $R$. The sequence of $maxPhase_{a_{i,k}}(phs)$'s returned values is non-decreasing. 
\end{corollary}

\subsection{Recovery After the Occurrence of Transient-Faults}
\label{sec:convergence}

The correctness of Algorithm~\ref{alg:cas} assumes that the system execution is fair. That is, every node participates in the execution within a single asynchronous cycle. This way, the proof bounds the number of asynchronous cycles that it takes the system to remove stale information by receiving the largest tag values and then allowing the system to perform a valid write operation (Definition~\ref{def:invalid}).    

\begin{definition}[\textbf{Notation}] 
\label{def:Tpi}
Let $\pi$ be a (complete) operation in execution $R$. We denote by $\hat{Q}(\pi)$ the quorum of servers that $\pi$'s client receives their acknowledgments for $\pi$'s query. 
Suppose that $\pi$ is a write operation. 
Similarly, denote by $Q_{\text{pw}}(\pi)$ and $Q_{\text{fw}}(\pi)$ the quorums for $\pi$'s pre-write, and respectively, finalize phases. 
Let $\hat{T}(\pi)$ be the maximum arriving tag during $\pi$'s query (line~\ref{alg:r:SelMax}). Denote by $T(\pi)$ the tag of $\pi$, such that when $\pi$ is a write operation, $T(\pi)=\hat{T}(\pi)+1$ is the tag in use during $\pi$'s pre-write (line~\ref{ln:preWrite}) and when $\pi$ is a read operation, $T(\pi)=\hat{T}(\pi)$ is the maximum arriving tag during $\pi$'s query (line~\ref{alg:r:SelMax}). Denote by $c_{start}(R)$ the starting state of Algorithm~\ref{alg:cas}'s execution $R$. Let $T_{node}(R)$ be the set tags in the state of any node in $c_{start}(R)$. Let $T_{comm}(R)$ be the set tags in the payload of any message that is delivered during $R$ but it is never sent during $R$, because it was in transit in the communication channels in $R$'s starting system state, $c_{start}(R)$. We denote by $T(R)= T_{node}(R) \cup T_{comm}(R)$ the set that includes all the tags in $c_{start}(R)$.
\end{definition} 

For the sake of compatibility of our proposal with the one in~\cite{DBLP:journals/dc/CadambeLMM17}, we define the set of legal executions (Definition~\ref{def:safeSystemState}) in a way that considers a recovery period from arbitrary (transient) faults as well as the case in which the system starts from a well-initialized system state (Definition~\ref{def:safeStart}).

\begin{definition}[\textbf{A safe system start}]
\label{def:safeStart}
Let $c_{safe}$ be a system state in which: (1) no client nor server is executing any procedure, (2) the communication channels from the clients to the servers (servers to clients), $pingTx$ and every entry of $pingRx$ (respectively, $pongTx$ and every entry of $pongRx$) include the message $\langle \bot \rangle$ (respectively, $\langle \bot, \bot \rangle$), (3) the communication channels  between any two servers, $gossipTx$ and every entry of $gossipRx$ include the message $(t_0, t_0, t_0)$, and (4) the storage $S$ of every server is empty. In this case, we say that $c_{safe}$ is one of the safe system states. 
\end{definition} 

The definition of recovery Algorithm~\ref{alg:cas}'s period uses the term a valid client operation (Definition~\ref{def:invalid}).

\begin{definition}[\textbf{Valid client operations}]
\label{def:invalid}
Let $\pi$ be a complete operation in $R$. Suppose that there exists a system state $c \in R$, such that $c$ appears in $R$ before $\pi$'s start in $c_{start}(\pi)$ and $\pi$'s tag is greater than any tag that appears both in $c$ and $R$'s starting system state, i.e., $\max (T(c_{start}(R))\cap T(c)) < T(\pi)$ . In this case, we say that $\pi$ is \emph{valid}. 
\end{definition} 

Definition~\ref{def:safeSystemState} specifies legal executions as such that follow at least one complete and valid operation. 
 
\begin{definition}[\textbf{Recovery periods and legal executions}]
\label{def:safeSystemState}
Let $R = R_{recoveryPeriod} \circ R_{legalExecution}$ be an execution of Algorithm~\ref{alg:cas} that (is legal with respect to the external building blocks in Section~\ref{sec:coomSafe} and it) has an arbitrary starting system state, $c_{start}(R)$, with respect to Algorithm~\ref{alg:cas}. Suppose that within a finite number of steps the system reaches a state $c_{start}(\pi_{complete\&valid})$, such that (1) $c_{start}(\pi_{complete\&valid})$ is the starting system state of a complete and valid write operation $\pi_{complete\&valid}$, and (2) $R$'s suffix, $R_{legalExecution}$, starts at $c_{end}(\pi_{complete\&valid})$, where  atomicity and liveness hold with respect to any operation that is complete in suffix $R_{legalExecution}$ of $R$ (Definition~\ref{def:startEndSys}), which starts immediately after $c_{end}(\pi_{complete\&valid})$. 
In this case, we refer to $R_{recoveryPeriod}$ and $R_{legalExecution}$ as $R$'s recovery, and respectively, legal periods. We also consider any execution that starts from $c_{safe}$ (Definition~\ref{def:safeStart}) to be legal. Namely, $R_{legalExecution}$ is an asynchronous execution of Algorithm~\ref{alg:cas}.   
\end{definition}

Theorem~\ref{thm:convergence} shows that the system reaches a legal execution eventually. Theorems~\ref{thm:atomicity} and~\ref{thm:liveness} show that it takes merely a single complete and valid write operation $\pi_{complete\&valid}$  (Definition~\ref{def:safeSystemState}) to end the recovery period after which the system executes legally, because they demonstrate correct shared-memory emulation. Our proof shows that fair executions guarantee recovery. Once $\pi_{complete\&valid}$ had occurred, the correct system behavior no longer needs the above fairness assumption. Recall that, within $\bigO(1)$ asynchronous cycles, Algorithm~\ref{alg:cas}'s execution reaches a suffix in which the correctness of gossip and quorum-based communication is guaranteed (Corollary~\ref{thm:basicCorollary}). Therefore, Theorem~\ref{thm:convergence} considers an execution of algorithm~\ref{alg:cas} that (is legal with respect to the external building blocks in Section~\ref{sec:coomSafe}), because it demonstrates that the system reaches suffix $R_{no~incomplete}$ within $\bigO(1)$ asynchronous cycles. 

\begin{theorem}[\textbf{Recovery after the occurrence of transient-faults}]
\label{thm:convergence}
Let $R$ be a fair execution of algorithm~\ref{alg:cas} that (is legal with respect to the external building blocks in Section~\ref{sec:coomSafe} and it) has an arbitrary starting system state, $c_{start}(R)$, with respect to Algorithm~\ref{alg:cas}. Within $\bigO(1)$ asynchronous cycles, execution $R=R'\circ R_{no~incomplete}$ has a suffix $R_{no~incomplete}$ that does not include incomplete operations. Moreover, within $\bigO(1)$ asynchronous cycles, execution $R_{no~incomplete}$ reaches a suffix, $R_{completeNonStable}$, that does not include invalid operations.
\end{theorem}

\begin{proof}
The proof is implied by Claim~\ref{thm:noStaleWrite}, which uses claims~\ref{thm:completeOperations} and~\ref{thm:noOperation}. Leveraging Corollary~\ref{thm:basicCorollary}, Claim~\ref{thm:completeOperations} shows that Algorithm~\ref{alg:cas}'s executions stop having incomplete operations.
\begin{claim}
\label{thm:completeOperations}
Let $R$ be a fair execution of Algorithm~\ref{alg:cas} in which the gossip functionality behaves correctly. Within $\bigO(1)$ asynchronous cycles, $R$ includes a suffix, $R_{no~incomplete}$, that does not include: (1) operations that are incomplete in $R$, nor (2) incomplete client requests or server replies in $R$. Moreover, 
(3) $\exists c \in R_{no~incomplete}:\forall p_i \in \sP: \exists t \in \sT : t' \in (T(R_{no~incomplete})) \implies \exists (t, \bullet) \in S_{p_i}:t\geq t'$ in $c$.  
\end{claim}

\begin{claimproof}
\noindent \textbf{Part (1).~~} 
Lemma~\ref{thm:progression} implies that all incomplete operations end within $\bigO(1)$ asynchronous cycles.

\noindent \textbf{Part (2).~~} 
Suppose that all operations in $R_{no~incomplete}$ are complete, i.e., no incomplete request or replies enter the system throughout $R_{no~incomplete}$. (Due to Part (1) of this proof, we can make this assumption without losing generality.) Lemma~\ref{thm:progression} implies the correct behavior of the quorum-based communication functionality within $\bigO(1)$ asynchronous cycles, which implies Part (2). 

For the sake of simple presentation, the rest of this proof assumes that throughout $R_{no~incomplete}$, all of client requests and server replies were indeed (Section~\ref{sec:coomSafe}). 

\noindent \textbf{Part (3).~~} We start by showing that $T_{comm}(R_{no~incomplete})=\emptyset$ (Definition~\ref{def:Tpi}). Parts (1) and (2) of this proof says that $R_{no~incomplete}$ does not include the delivery of messages that were never sent in $R$. This implies $T_{comm}(R_{no~incomplete})=\emptyset$,  because $T_{comm}()$'s definition considers any message that is delivered but never sent during (due to the fact that they were in transit at the starting system state of $R_{no~incomplete}$).


Due to the above, we only show that within $\bigO(1)$ asynchronous cycles in $R_{no~incomplete}$, the system reaches a system state $c \in R_{no~incomplete}$, such that $\forall p_i \in \sP: \exists t \in \sT : t' \in (T_{node}(R_{no~incomplete}))$ $\implies \exists (t, \bullet) \in S_{p_i}:t\geq t'$. Let $t' \in (T_{node}(R_{no~incomplete}))$. Suppose that $t'$ appears in the client state at node $p_j \in \sP$. By the assumption that this theorem makes about $R$ fairness, we know that $p_j$'s client operation terminates within $\bigO(1)$ asynchronous cycles (Lemma~\ref{thm:progression}). Once that happen, the client state no longer includes any tag value (cf. part (c) of the quorum-based communication service and Corollary~\ref{thm:basicCorollary}). Suppose that $t'$ is part of the server state, i.e., $\exists p_j \in \sP : (t', \bullet) \in S_{p_j}$. Let us consider a choice of $p_j$ and $t'$, such that $t'$ is maximal. By Lemma~\ref{thm:eventuallyAscendingSent}, within $\bigO(1)$ asynchronous cycles, $p_i$'s server receives at least one gossip that includes a tag $t''\geq t'$ that is not smaller than $t'$. The proof is done by replacing $t$ with $t''$ in the invariant that we need to prove, i.e., $\forall p_i \in \sP: \exists t'' \in \sT : t' \in (T_{node}(R_{no~incomplete})) \implies \exists (t'', \bullet) \in S_{p_i}:t''\geq t'$.
\end{claimproof}

Claim~\ref{thm:noOperation} shows that an $R_{no~incomplete}$'s operation, $\pi$, uses a tag that is not smaller than any (maximal) tag $T_{maxQuery}(\pi)$ on the servers that participate in $\pi$'s query quorum, where $\emph{tags}(C,D)=\{ t:(t,\bullet, d) \in S_{p_j} \land d \in D \land p_j \in C\}$ (Definition~\ref{def:startEnd}) and $T_{maxQuery}(\pi)=\max \emph{tags}(\hat{Q}(\pi),\sD \setminus \{\text{`pre'}\})$ in $c_{start}(\pi) \in R_{no~incomplete}$.

\begin{claim}
\label{thm:noOperation}
Let $\pi$ be an $R_{no~incomplete}$'s operation. 
$T(\pi) \geq T_{maxQuery}(\pi)$ in $c_{start}(\pi) \in R_{no~incomplete}$. Moreover, $T(\pi) > T_{maxQuery}(\pi)$ when $\pi$ is a write operation.
\end{claim}

\begin{claimproof}
Due to the correctness of the quorum-based communication functionality during $R_{no~incomplete}$ (Claim~\ref{thm:completeOperations}), Corollary~\ref{thm:nonDecr} as well as lines~\ref{alg:w:SelMax},~\ref{alg:r:SelMax}, and~\ref{alg:srv:queryInit} to~\ref{alg:srv:queryInitElse}, it holds that $\hat{T}(\pi)$ is not smaller than any write or read tag in $S_{p_j}:p_j \in \hat{Q}(\pi)$ in $c_{start}(\pi)$. Moreover, $T(\pi)\geq\hat{T}(\pi)$ (Definition~\ref{def:Tpi}) and $T(\pi) > \hat{T}(\pi)$ when $\pi$ is a write operation. Thus, in $c_{start}(\pi)$, it holds that $T(\pi)$ is not smaller than any tag in $S_{p_j}:p_j \in \hat{Q}(\pi)$ (and it is actually greater when $\pi$ is a write operation). 
\end{claimproof}

Claim~\ref{thm:noStaleWrite} implies that any write operation $\pi_{write}$ in $R_{completeNonStable}$ is valid with respect to $R$ and by that we complete the proof.

\begin{claim}
\label{thm:noStaleWrite}
Within $\bigO(1)$ asynchronous cycles, execution $R_{no~incomplete}$ reaches a suffix, which we denote by $R_{completeNonStable}$, such that for any of $R_{completeNonStable}$'s write operations, $\pi_{write}$, it holds that $T(\pi_{write}) > \max (T(R_{no~incomplete}))$ in $c \in R_{completeNonStable}$.   
\end{claim}

\begin{claimproof}
The proof is implied by Part (3) of Claim~\ref{thm:completeOperations} and Claim~\ref{thm:noOperation}. 
\end{claimproof}\end{proof}

\subsection{Atomicity of Algorithm~\ref{alg:cas}}
\label{sec:atom}
We demonstrate that, after a recovery period (Definition~\ref{def:safeSystemState}), Algorithm~\ref{alg:cas} emulates shared atomic read/write memory. Some elements of the following proof are similar to arguments in~\citep[Theorem 1]{DBLP:journals/dc/CadambeLMM17}. Note that Theorem~\ref{thm:atomicity} considers $R_{legalExecution}$ but does not require fairness. By that it merely assumes that at least a single complete and valid write operation occurred during the recovery period (Definition~\ref{def:safeSystemState}) or that the system starts in a safe state (Definition~\ref{def:safeStart}).

\begin{theorem}[\textbf{Atomicity}]
\label{thm:atomicity}
Algorithm~\ref{alg:cas} is atomic in $R_{legalExecution}$.
\end{theorem}

The $\prec$ order satisfies the sufficient conditions for atomicity (Corollary~\ref{thm:prec}), which we borrow from~\cite{DBLP:journals/dc/CadambeLMM17}.

\begin{corollary}[Lemma 2 in~\cite{DBLP:journals/dc/CadambeLMM17}] 
\label{thm:prec}
Let $\Pi$ be the set of all operations in $R$. Suppose that $\prec$ is an irreflexive partial ordering of all the operations in $\Pi$ that satisfies: (1) when $\pi_1$'s return precedes $\pi_2$'s start in R, $\pi_2 \prec \pi_1$ is false. (2) When $\pi_1 \in \Pi$ is a write operation and $\pi_2 \in \Pi$ is any client operation, either $\pi_1 \prec \pi_2$ or $\pi_2 \prec \pi_1$ holds (but not both). (3) The value returned by each read operation is the value written by the last preceding write operation according to $\prec$ (or $v_0$, which is the default object value in the absence of such write).
\end{corollary}

\begin{definition}
\label{def:prec}
Define $\pi_1 \prec \pi_2$ if (i) $T(\pi_1) < T (\pi_2)$, or (ii) $T(\pi_1) = T(\pi_2)$, $\pi_1$ is a write and $\pi_2$ is a read. \end{definition} 

We show that $\prec$ satisfies the conditions of Corollary~\ref{thm:prec}. The proof of the closure property follows similar arguments to the ones made by Cadambe et al.~\cite{DBLP:journals/dc/CadambeLMM17}. 
It shows that by the time that operation $\pi$ ends, the tag $T(\pi)$ has finished propagating and installing the messages $\langle T(\pi), \bullet, \text{`fin'} \rangle$ in the storage of at least one quorum of servers (Lemma~\ref{thm:allServersStore}). It uses the visibility of $T(\pi)$ for claiming that the query phase of any operation that starts after $\pi$'s end, retrieve a tag that is at least as large as $T(\pi)$ (Lemma~\ref{thm:greaterT}). This is the basis of showing that each write operation has a unique tag (Lemma~\ref{thm:neqTag}). We complete the proof of Theorem~\ref{thm:atomicity} by demonstrating conditions (1) and (2) of Corollary~\ref{thm:prec} (using lemmas~\ref{thm:greaterT} and~\ref{thm:neqTag}) and well as condition (3) by considering read and write operations (Algorithm~\ref{alg:cas}) during $R_{legalExecution}$.

Lemma~\ref{thm:allServersStore} is a variation on Lemma 3 in~\cite{DBLP:journals/dc/CadambeLMM17}. We use Lemma~\ref{thm:noRemoval} for arguing that the servers (Algorithm~\ref{alg:cas}) store the currently maximal records (and any record with a higher tag). This variation is needed, because Lemma~\ref{thm:allServersStore} considers only operations that start after the (last) valid and complete write operation $\pi_{complete\&valid}$ (or a system that starts in a safe system state, cf. Definition~\ref{def:safeStart}).

\begin{lemma}[Storing the operation records] 
\label{thm:allServersStore}
Suppose that $\pi$ is a complete (read or write) operation in $R_{legalExecution}$. There is a quorum $Q_{\text{fw}}(\pi) \in \sQ$, such that all of its servers store the triple $(t, w, \text{`fin'})$, where $t = T(\pi)$ and $w \in \sW \cup \{ \bot \}$ in $c_{end}(\pi)$ and in every system state after $c_{end}(\pi)$.
\end{lemma}

\begin{proof}
Let $Q_{\text{fw}}(\pi)$ the quorum that $\pi$'s client (at node $p_i$) receives responses from $Q_{\text{fw}}(\pi)$'s servers during $\pi$'s finalize phase (lines~\ref{ln:finalize} and~\ref{ln:finalizeR}). Since $\pi$ is complete as well as the functionalities of gossip and quorum-based communication are correct in $R_{legalExecution}$ (Corollary~\ref{thm:basicCorollary}), it is true that the server at node $p_j \in Q_{\text{fw}}(\pi)$ responds to $\pi$'s finalize message (line~\ref{alg:srv:fin}) at some step $a_{\text{fw},j} \in R_{legalExecution}$.
Note that: (i) $p_j$'s response arrives eventually to $p_i$'s writer and that occurs before the system reaches $c_{end}(\pi)$, because $p_j \in Q_{\text{fw}}(\pi)$, as well as (ii) the servers (Algorithm~\ref{alg:cas}) keep in their storage the currently maximal (write, read, and anchor) records and any received record with a tag that is higher than the ones in these records (Lemma~\ref{thm:noRemoval}).
\end{proof}

\begin{lemma}[Similar to Lemma 4 in~\cite{DBLP:journals/dc/CadambeLMM17}] 
\label{thm:greaterT}
Let $\pi_i:i \in \{1,2\}$ be two complete operations in $R$, such that each $\pi_i$ starts immediately after the system states $c^s_i \in R:i \in \{1,2\}$ and returns immediately before $c^r_i \in R:i \in \{1,2\}$. Assume that $c^r_1$ appears before $c^s_2$ in $R$. (1) $T(\pi_2) \geq T(\pi_1)$ and (2) when $\pi_2$ is a write operation, $T(\pi_2) > T(\pi_1)$.
\end{lemma}

\begin{proof} 
Let $\hat{T}(\pi)$ be the maximum arriving tag during $\pi$'s query (lines~\ref{alg:w:SelMax} and~\ref{alg:r:SelMax}). It is suffices to show that $\hat{T}(\pi_2)\geq T(\pi_1)$ (Claim~\ref{thm:hatTgeqTpi1}), because when $\pi_2$ is a read, $T(\pi_2) = \hat{T}(\pi_2)$, and when $\pi_2$ is a write, $T(\pi_2) > \hat{T}(\pi_2)$ (see the pseudo-code of the reader and writer in Algorithm~\ref{alg:cas}). 

\begin{claim}
\label{thm:hatTgeqTpi1}
$\mathbf{\hat{T}(\mathbf{\pi}_2) \geq T(\mathbf{\pi}_1)}$.
\end{claim}
\begin{claimproof}
Let $\hat{Q}(\pi_i)$ be the set of nodes that their servers respond to $\pi_i$'s query (lines~\ref{alg:w:SelMax} and~\ref{alg:r:SelMax}). Note the existence of node $p_j \in \hat{Q}(\pi_2) \cap Q_{\text{fw}}(\pi_1)$ (Lemma~\ref{thm:quor}) that its server responds to $\pi_2$'s query with $(t, \bullet, \text{`qry'})$ (line~\ref{alg:srv:queryInit}) immediately after some system state $\hat{c}_{2,j} \in R$, where $t$ is the highest tag of a finalized (or FINALIZED) record that $p_j$ stores in $S_{p_j}$. We argue that $t\geq T(\pi_1)$, because $( T(\pi_1), \bullet,d) \in S_{p_j}: d \in \sD \setminus \{\text{`pre'}\}$ in $\hat{c}_{2,j}$ and $t$ is $S_{p_j}$'s highest finalized tag in $\hat{c}_{2,j}$. In detail, we argue the following. 

\begin{enumerate}
\item The fact that $p_j \in Q_{\text{fw}}(\pi_1)$ implies $( T(\pi_1), \bullet,d) \in S_{p_j}: d \in \sD \setminus \{\text{`pre'}\}$ as long as $\hat{c}_{2,j}$ appears after $c_{end}(\pi_1)$ in $R$ (Lemma~\ref{thm:allServersStore}). Moreover, $\hat{c}_{2,j}$ indeed appears after $c_{end}(\pi_1)$ in $R$, since $c^r_1$ appears before $c^s_2$ in $R$ (by this lemma assumption) and $c^s_2$ cannot appear after $\hat{c}_{2,j}$ (by the fact that $\hat{c}_{2,j}$ appears immediately before the response to a query that is sent immediately after $c^s_2$). 

\item The fact that $p_j \in \hat{Q}(\pi_2)$ implies that $p_j$ responds to $\pi_2$'s query (by $\hat{Q}(\pi_2)$'s definition), 

\item The server at $p_j$ replies with $(t,\bullet,\text{`qry'})$ to $\pi_2$'s query, such that $(t,\bullet,d) \in S_{p_j}:d \in \sD \setminus \{\text{`pre'}\}$, where $t$ is $S_{p_j}$'s highest finalized (or FINALIZED) tag (line~\ref{alg:srv:queryInit}) in $\hat{c}_{2,j}$. 
\end{enumerate}

Since $t\geq T(\pi_1)$, it holds that $\pi_2$'s query phase includes the reception of a response with a tag that is no smaller than $T(\pi_1)$. Thus, $\hat{T}(\pi_2)\geq T(\pi_1)$.
\end{claimproof}
\end{proof}

\begin{lemma}[cf. Lemma 5 in~\cite{DBLP:journals/dc/CadambeLMM17}]
\label{thm:neqTag}
Let $\pi_1$ and $\pi_2$ be two write operations in $R$. $T(\pi_1)\neq T(\pi_2)$.
\end{lemma}

\begin{proof} 
Denote by $id_i:\in\{1,2\}$ the identifier of the node that invokes operation $\pi_i$. 
Note that $id_1 \neq id_2$ implies $T(\pi_1)\neq T(\pi_2)$, because $T(\pi_i) = (z_i, id_i)$ (lines~\ref{ln:preWrite},~\ref{ln:finalize} and~\ref{ln:FINALIZE}, Algorithm~\ref{alg:cas}). 
Thus, until the end of this proof, we focus on the case in which $id_1 = id_2$. The client (at node $p_i$) performs sequentially the operations $\pi_1$ and $\pi_2$ (Section~\ref{sec:sys}), i.e., one of them ends before the other starts. Let us assume, without loss of generality, that $\pi_1$ ends before $\pi_2$ starts. $T(\pi_2) > T(\pi_1)$ (Lemma~\ref{thm:greaterT}) implies that $T(\pi_2) \neq T(\pi_1)$.
\end{proof}

\begin{proof}[\textbf{\emph{Proof of Theorem~\ref{thm:atomicity}}}] 
For any two operations $\pi_1$, $\pi_2$, the definition of $\prec$ (Corollary~\ref{thm:prec}) says $\pi_1 \prec \pi_2$ when: (i) $T(\pi_1) < T (\pi_2)$, or (ii) $T(\pi_1) = T(\pi_2)$ as long as $\pi_1$ is a write and $\pi_2$ is a read. Suppose that operations $\pi_1$ and $\pi_2$ occur in Algorithm~\ref{alg:cas}'s legal execution $R_{legalExecution}$. After verifying that ${\prec}$ is indeed a partial order, we show the three properties of Corollary~\ref{thm:prec}.

\noindent \textbf{The relation $\mathbf{\prec}$ is a partial order.~~} We demonstrate that ${\pi_1 \prec \pi_2} \implies \pi_2 \nprec \pi_1$ by assuming that this statement is false, i.e., $\pi_1 \prec \pi_2 \land \pi_2 \prec \pi_1$, and then show a contradiction. Note that $(T(\pi_1) \leq T(\pi_2)) \land (T(\pi_2) \leq T(\pi_1)) \implies T(\pi_1) = T(\pi_2)$ ($\leq$'s definition). 
Therefore, $\pi_1$ is a write and $\pi_2$ is a read (Part (ii), Definition~\ref{def:prec}). Using symmetrical arguments, $\pi_2$ is a write and $\pi_1$ is a read. A contradiction.

\noindent \textbf{Property (1) of Corollary~\ref{thm:prec}.~~} 
Assume that $\pi_1$ returns before $\pi_2$ starts in $R$. 
We show that whether $\pi_2$ is a read or a write, it holds that $\pi_2 \prec \pi_1$ is false. 

\begin{itemize}

\item When $\pi_2$ is a read, $T (\pi_2) \geq T(\pi_1)$ (Lemma~\ref{thm:greaterT} as well as the assumption that $\pi_1$ returns before $\pi_2$ starts). Thus, $\pi_2 \prec \pi_1$ is false, because otherwise, by Definition~\ref{def:prec} of the order $\prec$, it holds that: (i) $T (\pi_1) > T(\pi_2)$, which contradicts the above, or (ii) $\pi_2$ is a write (Definition~\ref{def:prec} of the order $\prec$). Moreover, with respect to case (ii), if $\pi_2$ is a write, $T(\pi_2) > T(\pi_1)$ (Lemma~\ref{thm:greaterT} as well as the assumption that $\pi_1$ returns before $\pi_2$ starts). Thus, $\pi_1 \prec \pi_2$ is true (case (i), Definition~\ref{def:prec} of the order $\prec$). Moreover, $\pi_2 \prec \pi_1$ is false ($\prec$ is a partial order). 

\item When $\pi_2$ is a write $T(\pi_2) > T(\pi_1)$ (Lemma~\ref{thm:greaterT} as well as the assumption that $\pi_1$ returns before $\pi_2$ starts). Thus, $\pi_1 \prec \pi_2$ is true (case (i), Definition~\ref{def:prec} of the order $\prec$). Moreover, $\pi_2 \prec \pi_1$ is false ($\prec$ is a partial order) 
\end{itemize}

\noindent \textbf{Property (2) of Corollary~\ref{thm:prec}.~~} 
Lemma~\ref{thm:neqTag} implies that only case (i) of Definition~\ref{def:prec} holds. This implies Property (2), i.e., either  $\pi_1 \prec \pi_2$ or  $\pi_2 \prec \pi_1$ (but not both) hold. 

\noindent \textbf{Property (3) of Corollary~\ref{thm:prec}.~~} 
We show that every read operation $\pi$ in a legal execution $R_{legalExecution}$ returns a value that a preceding, according to $\prec$, write operation writes. (In the absence of such write operations, the read operation $\pi$ returns $v_0$, which is the default object value, line~\ref{alg:srv:default}).) To that end, we argue that: (i) there is a unique coupling between object version values and tag values and (ii) the read operation $\pi$ returns the value associated with $T(\pi)$. 

\noindent \textbf{(i) Unique coupling between object version values and tag values.~~} 
Recall that the system reaches $R_{legalExecution}$ after the system has performed at least one complete and valid write operation $\pi_{greatFIN} \in R_{completeNonStable}$ (Definition~\ref{def:safeSystemState} and Theorem~\ref{thm:convergence}).
After $\pi_{greatFIN}$, any succeeding write operation $\pi_{furtherWrite}$ in $R_{legalExecution}$ couples uniquely between versions of the data object and write operations in $R$ (Section~\ref{sec:sys}). We know that all written versions are uniquely associated with tag values (Lemma~\ref{thm:neqTag}). We note that even when starting the system in a state that includes no written object values, the servers reply with $(t_0,w_{0,i}, \text{`fin'})$ (line~\ref{alg:srv:default}) and the reader returns the decoding of that value (line~\ref{alg:read:BW}).

\noindent \textbf{(ii) The read operation $\pi$ returns the value associated with $T(\pi)$.~~} 
The complete read operation $\pi_{legitimateRead} \in R_{legalExecution}$ returns a value that is the result of retrieving and inverting the MDS code $\Phi$ using $k$ coded elements (line~\ref{alg:read:BW} and Definition~\ref{def:startEnd}). These $k$ coded elements were obtained at some previous point by applying $\Phi$ to the value associated with $T(\pi)$, where $\pi \in \{\pi_{greatFIN}, \pi_{furtherWrite}\}$ (line~\ref{ln:preWrite}). Therefore, the read operation $\pi$ returns the value associated with $T(\pi)$ due to the correctness of $\Phi$ (Section~\ref{sec:sys}).
\end{proof}

\subsection{Liveness of Algorithm~\ref{alg:cas}}
\label{sec:live}
\begin{definition}[\textbf{Liveness criteria}]
\label{def:liveness}
Suppose that there are no more than $f$ server failures and that $1 \leq k \leq N-2f$. In any fair and legal execution of Algorithm~\ref{alg:cas}, it holds that: (1) every operation terminates, and (2) the server replies to a reader's finalize phase includes at least $k$ (different) coded elements (and thus read operations can decode the retrieved values). 
\end{definition}

\begin{theorem}[\textbf{Liveness}] 
\label{thm:liveness}
The liveness criteria (Definition~\ref{def:liveness}) hold in Algorithm~\ref{alg:cas}'s fair and legal executions.
\end{theorem}

\begin{proof}
Note that Lemma~\ref{thm:progression} implies Part (1) of the liveness criteria (Definition~\ref{def:liveness}). Therefore, we focus on proving that read operations can decode the retrieved values (Part (2) of Definition~\ref{def:liveness}). I.e., at least $k$ servers include coded elements in their replies to a reader's finalize phase. The proof is implied from claims~\ref{thm:yesh} and~\ref{thm:besh} and the fact that Algorithm~\ref{alg:cas}'s servers do not remove records from their storage. 

\begin{claim}
\label{thm:yesh}
The query of a read operation $\pi_r$ in $R_{legalExecution}$ always returns a tag $t$ that is either $t_0$ or refers to the tag of a write operation $\pi_w$ that had a complete pre-write phase in $R$.
\end{claim}

\begin{claimproof}
Definition~\ref{def:safeSystemState} implies that $\pi_w$ always occurred before the legal execution (or the servers only consider the default tuple with the tag $t_0$). Lemma~\ref{thm:noRemoval} says that the servers do not remove their maximal records. Upon the arrival of $\pi_r$'s query message, the server reply with $\pi_w$'s tag (line~\ref{alg:srv:queryInitIf}), which is $t$. 
\end{claimproof}

\begin{claim}
\label{thm:besh}
As long the no server removes the record $(t,\bullet)$ from its storage, if it had any such record in $c_{start}(\pi_r)$, at least $k$ servers include coded elements in their replies to $\pi_r$'s finalize phase.  
\end{claim}

\begin{claimproof}
Let $Q_{\text{pw}}(t)$ denote the set of nodes that their servers acknowledge the pre-write phase of the write operation $\pi_w$ for which $t=T(\pi_w)$. Let $c_i$ be the system state that occurs immediately before the server at $p_i$ acknowledges $\pi_r$'s finalize message (line~\ref{alg:srv:finReply}). We show that the storage $S_{p_i}$ of every node $p_i \in Q_{\text{pw}}(t)\cap Q_{\text{fw}}(t)$ includes a coded element in $c_i$. Since $p_i \in Q_{\text{pw}}(t)$, it holds that $(t, w_{t,i}, \bullet) \in S_{p_i}$ in any system state that follows the step in which $p_i$ received $\pi_w$'s pre-write message (line~\ref{alg:srv:pre-write} and by the assumption of this claim that no server removes the record $(t,\bullet)$ from its storage). Note that $p_i \in Q_{\text{fw}}(t)$ indeed acknowledges the reader's finalize message, because of Claim~\ref{thm:yesh} and the fact that $c_i$ appears in $R$ after $p_i$ acknowledges that pre-write message. Therefore, $p_i$ includes in its reply the coded element $w_{t,i}$. By the correctness of the quorum-based communication during legal executions (Theorem~\ref{thm:convergence}, Claim~\ref{thm:completeOperations}), $\pi_r$ receives at least $k$ coded elements in its finalize phase, because $|Q_{\text{pw}}(t) \cap Q_{\text{fw}}(t)| \geq k$ (Part(1) of Lemma~\ref{thm:quor}).
\end{claimproof}
\end{proof}

\section{A Bounded Set of Relevant Server Records}
\label{sec:bound}
Algorithm~\ref{alg:cas}'s servers store the entire set of records that have arrived from the clients and the gossip service. This is in addition to the records that originated from the system starting state. To the end of bounding the number records that each server needs to store, we consider the relevance of a record with respect to the way that the servers use it after any point of time, i.e., a record is irrelevant in system state $c \in R_{legalExecution}$ if the server at $p_i \in \sP$ never use it after $c$ for responding to a client request. Theorem~\ref{thm:implicitHereIsExplicitThereMaybe} and Corollary~\ref{thm:noRemoved} point out a set that includes all relevant records and bound it by $N + \delta + 3$ during executions $R_{legalExecution}$ in which there are no more than $\delta$ write operations that occur concurrently with any read operation.

\begin{definition}[Tag visibility] 
\label{def:visRes}
Let $R$ be an execution of Algorithm~\ref{alg:cas}, $\pi_r$ be a read operation and $\pi_w$ be a write operation in $R$. Denote by $c_{visibility}(\pi_r)=c_{end}(\pi_r)$, which refers to $\pi_r$'s ending system state. We say that $\pi_r$ has visibility in $R$ starting from $c_{visibility}(\pi_r)$. Moreover, denote by $c_{visibility}(\pi_w) \in R$ either:
(i) the first system state, if such a state exists, for which a quorum $Q \in \sQ$ of non-failing nodes that their servers store the finalized record $(T(\pi_w), \bullet, d) \in S_{p_j \in Q} : d \in \sD \setminus \{ \text{`pre'} \}$, or
(ii) when case (i) does not hold in $R$ (because operation $\pi_w$ fails in $R$), $c_{visibility}(\pi_w)=c_{end}(\pi_w)$, which refers to $\pi_w$'s ending system state. When case (i) holds for $\pi_w$, we say that $\pi_w$ has visibility in $R$ starting from $c_{visibility}(\pi_w)$. Otherwise, $\pi_w$'s visibility is not guaranteed in $R$.
\end{definition}

\begin{definition}[Explicit and implicit FINALIZED tags and records] 
\label{def:explicitImplicitPhases}
Suppose that the server at node $p_i$ stores a finalized (or FINALIZED) record $r=(t,\bullet,d) \in S_{p_i}: t \in \sT \land d \in \sD \setminus \{\text{`pre'} \}$ in system state $c \in R$. In this case, we say that tag $t$ and record $r$ are explicitly finalized (with respect to the server) at $p_i$. Moreover, we say that tag $t$ and record $r$ are explicitly FINALIZED at $p_i$ when $(t,\bullet,\text{`FIN'}) \in S_{p_i}: t \in \sT$ in system state $c \in R$. 

Suppose that the server at node $p_i$ stores two records $r_1,r_2 \in S_{p_i}: \exists_{p_j \in \sP} \forall_{ x \in \{1,2\}}  t_x=(z_x,j) \land r_x=(t_x,\bullet)$ in system state $c \in R$ that their tags, $t_x=(z_1,j)$, and respectively, $t_x=(z_2,j)$, are associated with the client at $p_j$. Moreover, suppose that $t_1<t_2$. In this case, we say that tag $t_1$ and record $r_1$ are implicitly FINALIZED (in with respect to the server) at $p_i$. We denote $S_{p_i}$'s explicit FINALIZED records in $c$ by $S_{i,\textnormal{expFIN}} := \{ (t,\bullet,\text{`FIN'}) \in S_{p_i} \}$ and $S_{p_i}$'s implicitly FINALIZED records in $c$ by $S_{i,\textnormal{impFIN}} := \{ ((z_1,j), \bullet) \in S_{p_i} : \exists ((z_2,j), \bullet) \in S_{p_i} \land z_1 < z_2 \}$.
\end{definition}

Claim~\ref{thm:implicitHereIsExplicitThere} shows that an implicitly FINALIZED record at a server implies explicitly FINALIZED records at a server quorum.

\begin{claim}
\label{thm:implicitHereIsExplicitThere}
Suppose that $R_{legalExecution}$ includes a write operation $\pi$, such that in system state $c \in R_{legalExecution}$ it holds that $T(\pi)$ is implicitly FINALIZED at $p_i$. (1) $\pi$ is visible in $c$. (2) There is a quorum $Q \in \sQ$ of nodes that their servers store the FINALIZED record $(T(\pi), \bullet, \text{`FIN'}) \in S_{p_j \in Q}$. Suppose that in $c$ it holds that $T(\pi)$ is explicitly FINALIZED at $p_i$, i.e., $(T(\pi), \bullet, \text{`FIN'}) \in S_{p_i}$. (3) $\pi$ is visible in $c$.
\end{claim}
\begin{claimproof}
We start the proof by showing that $\pi$ includes the entire execution of the FINALIZED phase before $R_{legalExecution}$ reaches the system state $c$. We do that by demonstrating that $\pi$ is not an incomplete operation nor a failed one. Recall that Claim~\ref{thm:completeOperations} implies that $R_{legalExecution}$ does not include (write) operations that are incomplete and thus $\pi$ is not an incomplete operation. This claim assumes that in system state $c$, it holds that $T(\pi)$ is implicitly FINALIZED at $p_i$. This means that, in $c$, the server at node $p_i$ stores two records $r_1,r_2 \in S_{p_i}: r_x=(t_x,\bullet), t_x \in \sT \land t_x=(z_x,j) \land p_j \in \sP$, such that $T(\pi)=t_1<t_2$ (Definition~\ref{def:explicitImplicitPhases}). By the assumption that each node $p_j \in \sP$ lets its client to run just one procedure at a time, by the assumption that failing clients do not resume (Section~\ref{sec:benignFailures}), and by the writer code (lines~\ref{alg:w:SelMax} and~\ref{alg:w:return}), we have that $\pi$ is not a failed operation. Therefore, $\pi$ is a complete write operation that ends before $c$. In particular, $\pi$'s finalized and FINALIZED phases are done before $R$ reaches $c$ and thus parts (1) and (2) are correct (by Definition~\ref{def:visRes} and the correct operation of the quorum-based communications Corollary~\ref{thm:basicCorollary}). To show that part (3) also holds, we note that during $R_{legalExecution}$, any write operation $\pi$, which is after $\pi_{complete\&valid}$, updates to the record $(T(\pi), \bullet, \text{`FIN'}) \in S_{p_i}$ occurs only after the completion of the finalized phase (line~\ref{ln:finalize} and~\ref{ln:FINALIZE}). Thus, visibility is implied (Definition~\ref{def:visRes}).    
\end{claimproof}

\begin{definition}[The done system state $c_{done}(\pi)$] 
Let $R$ be an execution of Algorithm~\ref{alg:cas} and $\pi$ be a client (read or write) operation in $R$. Let $a_{k}(\pi) \in R$ be the step in which a server (at node $p_i \in \sP$) adds or updates the record $(T(\pi), \bullet)$ to its server storage, $S_{p_i}$ for the $k$-th time. This update could be due to the $\pi$ operation itself, another read operation $\pi_r \neq \pi$ for which $T(\pi_r)=T(\pi)$, or the arrival of a gossip message $(\bullet, T(\pi), \bullet)$. Denote $c_{0}(\pi):=c_{start}(\pi)$, $c_{k}(\pi) $ is the system state that immediately follows $a_{i,k}(\pi)$ and $c_{last}(\pi)=c_{\ell}(\pi)$, where $\ell$ is the maximum value of $\ell$ for which $\exists c_{\ell}(\pi) \in R$. We denote by $c_{done}(\pi) \in \{c_{last}(\pi), c_{end}(\pi)\}$ the system state that appears latest in $R$ between $c_{last}(\pi)$ and $c_{end}(\pi)$.
\end{definition}

\begin{definition}[Concurrent operations] 
\label{def:concurrent}
Let $\pi_1$ and $\pi_2$ be two operations in $R$. Suppose that $\nexists x,y \in \{1,2\}: x\neq y$, such that $c_{done}(\pi_x)$ appears before $c_{start}(\pi_y)$ in $R$. In this case, we say that $\pi_1$ and $\pi_2$ appear to be concurrent in $R$.
\end{definition}

We note that one way to explain Definition~\ref{def:concurrent}, is to say the following. When $c_{done}(\pi_x)$ appears before $c_{start}(\pi_y)$ in $R$, we can say that $R$ orders $\pi_x$ before $\pi_y$ sequentially. Moreover, $\pi_1$ and $\pi_2$ appears to be concurrent in $R$ if, and only if, $R$ neither orders $\pi_x$ before $\pi_y$ nor $\pi_y$ before $\pi_x$.

\begin{definition}[$\delta$-bounded concurrent write operations during any read in $R$] 
%
%
\label{def:deltaEpsilon}
Suppose that for every read operation $\pi_r$ in $R$, it holds that there are at most $\delta$ write operations in $R$ that are concurrent with $\pi_r$. In this case, we say that the number of concurrent write operations that occur in $R$ during any read operation is bounded by $\delta$ in $R$. 
\end{definition}

\begin{definition}[Record relevance] 
\label{def:recRelevance}
Let $r=(t,\bullet) \in S_{p_i}: t \in \sT$ be a record that the server at node $p_i \in \sP$ stores in system state $c \in R$. Suppose that there is a step $a_i$ that appears in $R$ after $c$ and in which the server at node $p_i$ responses to a (1) writer query request (line~\ref{alg:srv:queryInitElse}), (2) reader query request (line~\ref{alg:srv:queryInitIf}) or (3) reader finalized request (lines~\ref{alg:srv:fin}) with a message that includes tag $t' \leq t$. In this case, we say that tag $t$ and record $r$ are of relevance to $c$ with respect to a (1) writer query request, (2) reader query request, and respectively, (3) reader finalized request.
\end{definition}

\begin{definition}[The $T_{i,writeQuery}$, $T_{i,readQuery}$ and $T_{i,readFinalized}$ sets] 
\label{def:recRelevanceT}
Let $p_i \in \sP$ be a node with a server. Let $t_{i,FINALIZED}=t_{i,1},t_{i,2},\ldots :(t_{i,k},\bullet) \in (S_{i,\textnormal{expFIN}} \cup S_{i,\textnormal{impFIN}})$ (Definition~\ref{def:recRelevance}) be a (possibly empty) sequence tags in a descending order that are explicitly or implicitly FINALIZED at $p_i$ in system state $c$. Let $maxT_{i,FINALIZED} = \max \{t_{i,x} \in t_{i,FINALIZED} \}$ and $T_{i,FINALIZED}= \{t_{i,x} \in t_{i,FINALIZED}: x \leq \delta+1\}$. Let $T_{i,notYetFIN}= \{ t : (t,\bullet) \in S_{p_i} \setminus (S_{i,\textnormal{expFIN}} \cup S_{i,\textnormal{impFIN}})\}$ be a (possibly empty) set of tags that are at $p_i$'s  record storage and are not in $
T_{i,FINALIZED}$ in system state $c$. Let $T_{i,writeQuery}= \{\max \{ t: (t, \bullet) \in S_{p_i} \}\}$, $T_{i,readQuery}= \{\max \{ t: (t, \bullet,d) \in S_{p_i} : d \in (\sD \setminus \{ \text{`pre'} \}) \}\}$ and $T_{i,readFinalized}=T_{i,notYetFIN}\cup T_{i,FINALIZED}$ in system state $c$.
\end{definition}

\begin{lemma}
\label{thm:cannotBeMoreThanDelta}
Suppose that during any read operation in $R_{legalExecution}$ there are at most $\delta$ concurrent write operations. Suppose that $R_{legalExecution}$ includes a read operation $\pi_r$ and a step $a_i \in R_{legalExecution}$ in which the server at $p_i$ responds with $(T(\pi_r),\bullet)$ to $\pi_r$'s finalize  request (line~\ref{alg:srv:fin}), such that $R_{legalExecution}$ includes a write operation $\pi_w$ for which $T(\pi_w)=T(\pi_r)$. (If there is more than just one such operation, we select the latest one that appears before $\pi_r$ and note that by Theorem~\ref{thm:atomicity} these operations cannot be concurrent.)
It holds that $T(\pi_r) \in T_{i,readFinalized}$ in any system state $c \in R_{legalExecution}$ that is between $c_{i,in} \in R_{legalExecution}$ and $c_{i,out} \in R_{legalExecution}$, where $c_{i,in}$ is $R_{legalExecution}$'s first system state for which $(T(\pi_r),\bullet) \in S_{p_i}$ holds and $c_{i,out}$ is the system state that immediately precedes $a_i$. 
\end{lemma}

\begin{proof}
\begin{claim}
\label{thm:noMoreThanN}
Let $c \in R_{legalExecution}$ be a system state. $|S_{i,notYetFIN}| \leq N$ holds in $c$ (Definition~\ref{def:recRelevanceT}).
\end{claim}
\begin{claimproof}
By the definition of $S_{i,\textnormal{impFIN}} := \{ ((z_1,j), \bullet) \in S_{p_i} : \exists ((z_2,j), \bullet) \in S_{p_i} \land z_1 < z_2 \}$ (Definition~\ref{def:explicitImplicitPhases}), it holds that $S_{i,notYetFIN}:=S \setminus (S_{i,\textnormal{expFIN}} \cup S_{i,\textnormal{impFIN}})$ does not include any record $((z_1,j), \bullet)$ for which $((z_2,j), \bullet) \in S_{i,notYetFIN}$ and $z_1<z_2$. Therefore, every client can have at most one tag that appear in a record that belongs to $S_{i,notYetFIN}$. The proof of this claim is implied by the upper bound on the number of clients, which is $N$ (Section~\ref{ss:emulatedObject}).
\end{claimproof}

\begin{claim}
\label{thm:theMaxVisable}
Let $t_{start}$ be the maximum visible tag in $c_{start}(\pi_r) \in R$. It holds that $t_{start} \leq T(\pi_r)$.
\end{claim}

\begin{claimproof}
By the assumption that $t_{start}$ is the maximum visible tag in $c_{start}(\pi_r)$, it holds that there is a quorum $Q \in \sQ$ of nodes that their servers store the finalized record $(T(\pi), \bullet, d) \in S_{p_j \in Q} : d \in \sD \setminus \{ \text{`pre'} \}$ (Definition~\ref{def:visRes}), such that $\pi$ is a write operation in $R$ and $T(\pi)=t_{start}$. Let $\hat{Q}(\pi_r)$ be the set of nodes that $\pi_r$'s client receives their query responses (Definition~\ref{def:Tpi}). Note the existence of node $p_j \in \hat{Q}(\pi) \cap Q$ (Lemma~\ref{thm:quor}) that its server responds to $\pi_r$'s query with a tag that is at least $t_{start}$ (line~\ref{alg:srv:finReply}). The rest of the proof is implied by line~\ref{alg:r:SelMax} and Part (2) of Corollary~\ref{thm:basicCorollary}.
\end{claimproof}

\begin{claim}
\label{thm:theMaxVisableI}
Let $t_{i,start}:=\max T_{i,FINALIZED}$ be $T_{i,FINALIZED}$'s the maximum tag in $c_{start}(\pi_r) \in R$. It holds that $t_{i,start} \leq T(\pi_r)$.
\end{claim}

\begin{claimproof}
Part (1) of Claim~\ref{thm:implicitHereIsExplicitThere} implies that tag $t_{i,start}$ is visible in $c_{start}(\pi_r)$. Let $t_{start}$ be the maximal tag that has visibility in $c_{start}(\pi_r) \in R$, i.e., $t_{i,start} \leq t_{start}$. By Claim~\ref{thm:theMaxVisable}, we have that $t_{i,start} \leq t_{start} \leq T(\pi_r)$, which implies this claim.
\end{claimproof}

\begin{claim}
\label{thm:inToNotYet}
Let $t_{visibility}(c) \in \sT$ be the maximum explicitly visible tag in system state $c \in R_{legalExecution}$.
Suppose that $T(\pi_r) \geq t_{visibility}(c)$ in system state $c \in R_{legalExecution}$ that is between $c_{j,in} \in R_{legalExecution}$ and $c_{j,out} \in R_{legalExecution}$, where $p_j \in \sP$. In $c$, it holds that: (1) $T(\pi_r) \geq maxT_{i,FINALIZED}$, and (2) $(T(\pi_r),\bullet) \in S_{p_j}$ implies $T(\pi_r) \in \{t\in T_{j,notYetFIN} : t\geq maxT_{j,FINALIZED} \} \cup \{maxT_{j,FINALIZED}\}$.
\end{claim}
\begin{claimproof}
\noindent \textbf{Part (1).~~} 
Recall that $maxT_{j,FINALIZED} = \max \{t_{j,x} \in t_{j,FINALIZED} \}$ (Definition~\ref{def:recRelevanceT}), where $t_{j,FINALIZED}=t_{j,1},t_{j,2},\ldots :(t_{j,k},\bullet) \in (S_{j,\textnormal{impFIN}} \cup S_{j,\textnormal{expFIN}})$ (Definition~\ref{def:recRelevance}).
Let us consider any tag that is either in $S_{j,\textnormal{impFIN}}$ or $S_{j,\textnormal{expFIN}}$, i.e., any tag that is FINALIZED either (i) implicitly or (ii) explicitly. That is, we look at the cases in which (i) $(T(\pi_r),\bullet) \in S_{p_j}: p_k \in \sP \land T(\pi_r) = (z_1,k) \land \exists ((z_2,k),\bullet) \in S_{p_j}: z_1<z_2$ in $c$, or (ii) $(T(\pi_r),\bullet,  \text{`FIN'}) \in S_{p_j}$ in $c$. Parts (1), and respectively, (3) of Claim~\ref{thm:implicitHereIsExplicitThere} imply that $T(\pi_r)$ has visibility in $c$. This claim assumption says that $T(\pi_r) \geq t_{visibility}(c)$. Therefore, $T(\pi_r) \geq maxT_{j,FINALIZED}$ (Claim~\ref{thm:theMaxVisableI}). 

\noindent \textbf{Part (2).~~} 
By definitions~\ref{def:explicitImplicitPhases} and~\ref{def:recRelevanceT}, $(T(\pi_r),\bullet) \in S_{p_j}$ implies that either $(T(\pi_r),\bullet) \in (S_{j,\textnormal{impFIN}} \cup S_{j,\textnormal{expFIN}})$ or $(T(\pi_r),\bullet) \in T_{j,notYetFIN}$. Part (1) of this proof consider the former case and implies that $T(\pi_r) \in \{maxT_{j,FINALIZED}\}$ in $c$. The latter refers to the cases that Part (1) of this proof do not consider. That is, $(T(\pi_r),\bullet,d) \in S_{p_j}: d \in \sD \setminus \{ \text{`FIN'}\} \land p_k \in \sP \land T(\pi_r) = (z_1,k) \land \nexists ((z_2,k),\bullet) \in S_{p_j}: z_1<z_2$, which implies $(T(\pi_r),\bullet,d) \in T_{j,notYetFIN}$. 
\end{claimproof}
\begin{claim}
\label{thm:inToNotYetBetter}
Let $t_{visibility}(c) \in \sT$ be the maximum explicitly visible tag in system state $c \in R_{legalExecution}$ and $c_1,c_2,\ldots$ be a sequence of all system states in $R_{legalExecution}$ (in the order that they appear in $R_{legalExecution}$). 
(1) $(t_{visibility}(c),\bullet) \in S_{p_j}$ implies $t_{visibility}(c) \in \{t\in T_{j,notYetFIN} : t\geq maxT_{j,FINALIZED} \} \cup \{maxT_{j,FINALIZED}\}$ in $c$.
(2) The sequence $t_{visibility}(c_1), t_{visibility}(c_2), \ldots$ is monotonically increasing, i.e., $t_{visibility}(c_k) \leq t_{visibility}(c_{k+1})$. 
\end{claim}
\begin{claimproof}
\noindent \textbf{Part (1).~~} 
This is implied by Part (1) of Claim~\ref{thm:inToNotYet} and the definition of $T_{j,notYetFIN}$.
\noindent \textbf{Part (2).~~} 
According to Algorithm~\ref{alg:cas}, the server at $p_j \in \sP$ does not remove the records in $\{t\in T_{j,notYetFIN} : t\geq maxT_{j,FINALIZED} \} \cup \{maxT_{j,FINALIZED}\}$. The $\max$ function properties imply this part.
\end{claimproof}

\begin{claim}
\label{thm:lemmaDelta}
Let $c_1,c_2 \in R_{legalExecution}$ be two system states that appear between $c_{j,in} \in R_{legalExecution}$ and $c_{j,out} \in R_{legalExecution}$, where $p_j \in \sP$. It holds that $|(S_1 \cup S_2) \setminus (S_1 \cap S_2)| \leq \delta$, where $S_x\in\{1,2\}=S_{p_i}$ in $c_x$.
\end{claim}

\begin{claimproof}
By this lemma assumption, any read operation $\pi$ in $R_{legalExecution}$ has at most $\delta$ concurrent write operations (Definition~\ref{def:deltaEpsilon}). Recall that $R_{legalExecution}$ does not include incomplete operations (Claim~\ref{thm:completeOperations}). Therefore, an update or an addition of the record $(t,\bullet)$ to $S_{p_i}$ (between $c_{i,in}$ and $c_{i,out}$) implies that there is write operation $\pi_{w'}$ that is concurrent (Definition~\ref{def:concurrent}) with the read operation $\pi$. Thus, this claim. (Note that the same holds for this lemma's read operation, $\pi_r$.)
\end{claimproof}

\begin{claim}
\label{thm:lemmaDeltaSequence}
Let $t_{i,FINALIZED,c'}=t_{i,FINALIZED}$ denote the value of the sequence $t_{i,FINALIZED}$ in $c' \in R$. 
Let $c \in R_{legalExecution}$ be a system state that is between $c_{i,in}$ and $c_{i,out}$. 
The sequence $t_{i,FINALIZED,c}$ includes at most $\delta$ tags that are greater than $T(\pi_r)$, which are not in $t_{i,FINALIZED,c_{start}(\pi_r)}$.
\end{claim}

\begin{claimproof}
Claim~\ref{thm:theMaxVisableI} implies that $T(\pi_r)$ is greater than any element in $t_{i,FINALIZED,c_{start}(\pi_r)}$. From Claim~\ref{thm:lemmaDelta}, we get that, between ${c_{start}(\pi_r)}$ and ${c_{i,out}}$, Algorithm~\ref{alg:cas} may add to the sequence $t_{i,FINALIZED}$ at most $\delta$ records. Hence, the claim.
\end{claimproof}

\begin{claim}
\label{thm:lemmaEnd}
Let $c \in R_{legalExecution}$ be a system state that is between $c_{i,in}$ and $c_{i,out}$. It holds that $T(\pi_r) \in T_{i,readFinalized}$ (Definition~\ref{def:recRelevanceT}) in $c$.
\end{claim}

\begin{claimproof}
\paragraph{Suppose that $\mathbf{c_{i,in}}$ appears before $\mathbf{c_{start}(\pi_r)}$ in $\mathbf{R_{legalExecution}}$.} 

We show that the conditions of Claim~\ref{thm:inToNotYet} hold in $c$ and thus $T(\pi_r) \in T_{i,readFinalized}$. Specifically, we show that $T(\pi_r) \geq t_{visibility}(c)$ in $c$ and that $(T(\pi_r),\bullet) \in S_{p_i}$ in $c$, because then we can complete the proof by using $T_{i,readFinalized}=T_{i,notYetFIN}\cup T_{i,FINALIZED}$ (Definition~\ref{def:recRelevanceT}).

\subparagraph{Let us look at the case in which $c$ appears between $\mathit{c_{i,in}}$ and $\mathbf{c_{start}(\pi_r)}$ in $\mathbf{R_{legalExecution}}$ (including both system states $\mathbf{c_{i,in}}$ and $\mathbf{c_{start}(\pi_r)}$ as possible values of $\mathbf{c}$).} 

Recall that $T(\pi_r) \geq t_{visibility}(c_{start}(\pi_r))$ (Claim~\ref{thm:theMaxVisable}) and that $t_{visibility}(c_{start}(\pi_r)) \geq t_{visibility}(c)$ (Part (2) of Claim~\ref{thm:inToNotYetBetter} and this case assumption that  $c$ appears no later than ${c_{start}(\pi_r)}$ in $R$). Thus, $T(\pi_r) \geq t_{visibility}(c)$. 

To the end of showing that $(T(\pi_r),\bullet) \in S_{p_i}$ in $c$, we start by assuming that $c=c_{i,in}$ and then consider every system state $c$ that appears between ${c_{i,in}}$ and ${c_{start}(\pi_r)}$ (including the latter state). Recall that $c_{i,in}$ is $R_{legalExecution}$'s first system state for which $(T(\pi_r),\bullet) \in S_{p_i}$ holds (cf. this lemma's statement). Therefore, $(T(\pi_r),\bullet) \in S_{p_j}$ implies $T(\pi_r) \in \{t\in T_{j,notYetFIN} : t\geq maxT_{j,FINALIZED} \} \cup \{maxT_{j,FINALIZED}\} \subseteq T_{i,readFinalized}$ in $c=c_{i,in}$ (Part (2) of Claim~\ref{thm:inToNotYet}). 

Now, let us continue by assuming that $c$ is the state in $R_{legalExecution}$ that immediately follows $c_{i,in}$ (and yet $c$ does not appear in $R_{legalExecution}$ after $c_{start}(\pi_r)$). By the same arguments as above, it holds that $T(\pi_r) \geq t_{visibility}(c)$. Algorithm~\ref{alg:cas} does not include a line in which a server removes a record from its storage. Thus, we only need to show that $T(\pi_r)$ does not leave the set $T_{i,readFinalized}$ in the transition from $c_{i,in}$ to $c$. We show more than that, i.e., $T(\pi_r)$ does not leave the set $\{t\in T_{j,notYetFIN} : t\geq maxT_{j,FINALIZED} \} \cup \{maxT_{j,FINALIZED}\} \subseteq T_{i,readFinalized}$ in the transition from $c_{i,in}$ to $c$. 

We note that it cannot be the case that in $c_{i,in}$ we have $T(\pi_r) \in \{maxT_{j,FINALIZED}\}$ and $T(\pi_r) \notin \{maxT_{j,FINALIZED}\}$ in $c$. The reason is that $T(\pi_r) \in \{maxT_{j,FINALIZED}\}$ in $c_{i,in}$ says that $T(\pi_r)$ is (either explicitly or implicitly) FINALIZED in $c_{i,in}$ and that status filed in the record cannot change to a status that is not (either explicitly or implicitly) FINALIZED (Algorithm~\ref{alg:cas} and the way that Definition~\ref{def:recRelevanceT} constructs $t_{i,FINALIZED}$). 

Suppose that in $c_{i,in}$ it holds that $T(\pi_r) \in \{t\in T_{j,notYetFIN} : t\geq maxT_{j,FINALIZED} \}$ and in $c$ it holds that $T(\pi_r) \notin \{t\in T_{j,notYetFIN} : t\geq maxT_{j,FINALIZED} \}$. This implies that tag $T(\pi_r)$ becomes (either explicitly or implicitly) FINALIZED during that transition (Definition~\ref{def:recRelevanceT}), That is, $T(\pi_r) \in \{maxT_{j,FINALIZED}\}$ in $c$ and the proof is done.    

The rest of the proof of this part is followed by repeating the same arguments for every two consecrative system states $c'$ and $c''$ that are between $c_{i,in}$ and $c_{start}(\pi_r)$.   

\subparagraph{Let us look at the case in which $\mathbf{c}$ appears between $\mathbf{c_{start}(\pi_r)}$ and $\mathbf{c_{i,out}}$ in $\mathbf{R_{legalExecution}}$ (including both system states $\mathbf{c_{start}(\pi_r)}$ and $\mathbf{c_{i,out}}$ as possible values of $\mathbf{c}$).} 

From the proof of the previous case, when $c = c_{start}(\pi_r)$, it holds that $T(\pi_r) \in \{t\in T_{j,notYetFIN} : t\geq maxT_{j,FINALIZED} \} \cup \{maxT_{j,FINALIZED}\} \subseteq T_{i,readFinalized}$. Recall also from the previous case that if Algorithm~\ref{alg:cas} causes $T(\pi_r)$ to leave the set $\{t\in T_{j,notYetFIN} : t\geq maxT_{j,FINALIZED} \}$, then $T(\pi_r)$ becomes a member of the sequence $t_{i,FINALIZED}$ (Algorithm~\ref{alg:cas} and the way that Definition~\ref{def:recRelevanceT} constructs $t_{i,FINALIZED}$). From Claim~\ref{thm:lemmaDeltaSequence}, we get that Algorithm~\ref{alg:cas} may move $T(\pi_r)$ down the sequence $t_{i,FINALIZED}$, by including other  (either explicitly or implicitly) FINALIZED records with higher tags, at most $\delta$ times between ${c_{start}(\pi_r)}$ and ${c_{i,out}}$ but still include $T(\pi_r)$ in $T_{i,FINALIZED}$. This implies $T(\pi_r) \in T_{i,readFinalized}$ (Definition~\ref{def:recRelevanceT}) for the case in which $c$ appears between $\mathit{c_{start}(\pi_r)}$ and $\mathit{c_{i,out}}$ in $\mathit{R_{legalExecution}}$ as well as the case in which ${c_{i,in}}$ appears before ${c_{start}(\pi_r)}$ in ${R_{legalExecution}}$.    

\paragraph{Suppose that $\mathbf{c_{i,in}}$ appears after $\mathbf{c_{start}(\pi_r)}$ in $\mathbf{R_{legalExecution}}$.} 

By this case assumption, it holds that the tag $T(\pi_r)$ does not appear in the sequence $t_{i,FINALIZED}$ in $c_{start}(\pi_r)$. 
From Claim~\ref{thm:lemmaDeltaSequence}, we get that Algorithm~\ref{alg:cas} may include in the sequence $t_{i,FINALIZED}$ at most $\delta$ (either explicitly or implicitly) FINALIZED records with higher tags than $T(\pi_r)$ during the period that is between ${c_{start}(\pi_r)}$ and ${c_{i,out}}$. 
During this period, the record $(T(\pi_r), \bullet) \in S_{p_i}$ does appear in the storage of the server at $p_i$. 
By the arguments above, it appears either in $\{t\in T_{j,notYetFIN} : t\geq maxT_{j,FINALIZED} \}$ or in the top $\delta+1$ tags of $t_{i,FINALIZED}$. Therefore, $T(\pi_r) \in T_{i,FINALIZED}$ (Definition~\ref{def:recRelevanceT}) and we can complete the proof by using $T_{i,readFinalized}=T_{i,notYetFIN}\cup T_{i,FINALIZED}$ (Definition~\ref{def:recRelevanceT}).
\end{claimproof}
\end{proof}

\begin{theorem}[Only $T_{i,writeQuery}$, $T_{i,readQuery}$ and $T_{i,readFinalized}$ are relevant and they are bounded]
\label{thm:implicitHereIsExplicitThereMaybe}
Let $r=(t,\bullet) \in S_{p_i}: t \in \sT$ be a record that the server at node $p_i \in \sP$ stores in system state $c \in  R_{legalExecution}$. Suppose that tag $t$ is of relevance to $c$ with respect to a (1) writer query request, (2) reader query request or (3) reader finalized request. The server at $p_i$ stores the record $r=(t, w_j,\bullet) \in S_{p_j \in Q}$ and $r\in relevant(S_i)$, such that (1) $r \in T_{i,writeQuery}$, (2) $r \in T_{i,readQuery}$, and respectively, (3) $r \in T_{i,readFinalized}$ in $c$. Moreover, $|relevant(S_i)| \leq N +\delta+3$, where $relevant(S_i) := T_{i,writeQuery} \cup T_{i,readQuery} \cup T_{i,readFinalized}$.
\end{theorem}

\begin{proof}~~

\textbf{Showing that $\mathbf{r \in T_{i,writeQuery}}$.~~}
By this lemma assumption and Definition~\ref{def:recRelevance} it implies that $T_{i,writeQuery}= \{\max \{ t: (t, \bullet) \in S_{p_i} \}\}$ in $c$ (Definition~\ref{def:recRelevanceT}).
The server at node $p_i$ replies to a reader by returning the maximal tag $t$ in any record stored in $S_{p_i}$ (line~\ref{alg:srv:queryInitIf}). Therefore, the server at $p_i$ and $T_{i,writeQuery}$ store in system state $c$ any record $(t',\bullet)$ that is relevant with respect to a writer query request.

\textbf{Showing that $\mathbf{r \in T_{i,readQuery}}$.~~}
By this lemma assumption and Definition~\ref{def:recRelevance} it implies that $T_{i,readQuery}= \{ \max \{ t: (t, \bullet,d) \in S_{p_i} : d \in (\sD \setminus \{ \text{`pre'} \}) \}\}$ in $c$ (Definition~\ref{def:recRelevanceT}).
The server at node $p_i$ replies to a reader by returning the tag $t$ in any finalized or FINALIZED record stored in $S_{p_i}$ (line~\ref{alg:srv:queryInitIf}). Therefore, the server at $p_i$  and $T_{i,readQuery}$ store in system state $c$ any record $(t',\bullet)$ that is relevant with respect to a reader query request.

\textbf{Showing that $\mathbf{r \in T_{i,readFinalized}}$.~~}
The proof of this case is implied by Lemma~\ref{thm:cannotBeMoreThanDelta}.

\textbf{The bound $\mathbf{|relevant(S_i)| \leq N +\delta+3}$.~~} This bound comes from the Definition~\ref{def:recRelevanceT}, which implies $|T_{i,writeQuery}|\leq 1$ and $|T_{i,readQuery}|\leq 1$ as well as the definition of $T_{i,readFinalized}$ and Claim~\ref{thm:lemmaDelta}, which implies that during executions that have at most $\delta$ concurrent write operations, it holds that $|T_{i,readFinalized}|\leq N +\delta+1$.
\end{proof}

Corollary~\ref{thm:noRemoved} is implied directly from the definition of the set $relevant(S_i)$ (Theorem~\ref{thm:implicitHereIsExplicitThereMaybe}), Theorem~\ref{thm:implicitHereIsExplicitThereMaybe} and line~\ref{alg:srv:uponGossip} to~\ref{alg:srv:FINMaxGossipHelper} of Algorithm~\ref{alg:cas}.

\begin{corollary}
\label{thm:noRemoved}
Let $p_i \in \sP$ be a node that hosts a server. The set $relevant(S_i)$ (Theorem~\ref{thm:implicitHereIsExplicitThereMaybe}) always includes the records $(t_{write},\bullet)$, $(t_{read},\bullet)$ and $(t_{anchor},\bullet)$, which $p_i$ gossips their tags in the triple $(t_{write},t_{read},t_{anchor})$ (line~\ref{alg:srv:FINMaxGossipHelper}).
\end{corollary}

\section{A Bounded Variation on Algorithm~\ref{alg:cas}}
\label{sec:bounded}

\begin{figure}[t!]

\fbox{
\begin{minipage}{0.975\linewidth}

\begin{enumerate}

\item\label{ln:stop} Once the server at $p_i \in \sP$ stores in $S_{p_i}$ a record with a tag that is at least $t_{top}=(MAXINT,minID)$, where $minID:=\max \{k : p_k \in \sP\}$, the server at $p_i$ suspends all responses to new write operations (at their query phase) while allowing the completion of the existing ones (until all servers agree on the highest finalized tag, cf. item~\ref{ln:restart}), where $MAXINT \in \mathbb{Z}^+$, say, $MAXINT = 2^{64}-1$. To that end, the server program tests whether $(maxPhase(\sD)\leq t_{top})$ and respond to the writer only when the tests passes. That is, we modify line~\ref{alg:srv:queryInitElse} to ``{else if} $(maxPhase(\sD)\leq t_{top})$ {then} $\text{reply}(j, (maxPhase(\sD), \bot, \text{`qry'}))$.''

\item\label{ln:restart} While the invocation of new write operations is suspended by the modified line~\ref{alg:srv:queryInitElse} (item~\ref{ln:stop}), the gossip procedure keeps on propagating the maximal tags, which is $tagTuple()$'s returned value in lines~\ref{alg:srv:uponGossip} and~\ref{alg:srv:FINMaxGossipHelper} (as we show in Claim~\ref{thm:convergenceMaxTagE}). Eventually, the servers at all nodes share the same triple of maximal tags. At that point in time, this Algorithm~\ref{alg:cas}'s variation uses the global reset procedure $\mathit{globalReset}(t)$ (Section~\ref{s:bld}) that (i) removes any record for all the server storages other than the ones with the tag $t=maxPhase(\sD \setminus \{\text{`pre'}\})$, (ii) replaces the tag $t=(z,k)$ in that record with the tag $(1,k)$ and (iii) stops forever all on-going client operations. To that end, between lines~\ref{alg:srv:FINMaxUpdt} and~\ref{alg:srv:FINMaxGossipHelper}, the server program also includes the following if-statement: ``if $(maxPhase(\sD)\geq t_{top}) \land (\forall p_k \in \sP : gossip[k] = tagTuple()) \land (tagTuple() = (t,t',t') \land t\geq t')$ then $\mathit{globalReset}(maxPhase(\sD \setminus \{\text{`pre'}\}))$ else $S \gets relevant(S)$,'' where $relevant(S)$ is taken from Theorem~\ref{thm:implicitHereIsExplicitThereMaybe}.

\end{enumerate}

\end{minipage}
}

\caption{A bounded extension of Algorithm~\ref{alg:cas}.}
\label{fig:boundCAS}
\end{figure}

We present a variation of Algorithm~\ref{alg:cas} that has bounded message and state size. Figure~\ref{fig:boundCAS} adds a couple of lines to the code of Algorithm~\ref{alg:cas} and uses the external building block $\mathit{globalReset}()$ (Section~\ref{s:bld}). Theorem~\ref{thm:bounded} demonstrates the correctness of the proposed variation. Note that the proof assumes the execution to be fair eventually in the manner of self-stabilizing systems in the presence of seldom fairness (Section~\ref{sec:sys}). Namely, once the storage of at least one server includes at least one record with a tag $t$ that is at least $t_{top}$ (Figure~\ref{fig:boundCAS}), we require the system execution to eventually be fair until all nodes return from the call to $\mathit{globalReset}()$. This requirement is indeed seldom, because such fair executions are needed only once in every $\bigO(z_{\max})$ write operations and during the recovery from rare transient faults  (Theorem~\ref{thm:atomicity}). After the recovery period and during the periods in which no server stores tag $t\geq t_{top}$, the execution is not required to be fair.

\begin{definition}[Legitimate overflows]
\label{def:wellStarted}
We say that system state $c$ is overflow-free when every tag $t < t_{top}$ in $c$ is smaller than the one that would trigger an overflow event. We say that execution $R$ has a legitimate overflow event if $R$'s starting system state $c$ is (i) both overflow-free and reset-free (Section~\ref{ss:BBBreq}) as well as (ii) the first step that immediately follows $c$ includes the start (the first sent request to the server) of a pre-write phase that has the tag $t \geq t_{top}$. Let $R''$ be a suffix of $R=R'\circ R''$ that (a) includes a starting system state in which any $p_i \in \sP$ (that hosts a server) stores a record $(t,\bullet) \in S_{p_i}$ with tag $t \geq t_{top}$ and (b) $R'$ is the shortest matching prefix of $R''$ in $R$. In this case, we say that $R''$ is an execution with a legitimate overflow record. (Note that $R''$ may have system states, including the starting one, with tags $t' \geq t_{top}$, such that $t\neq t'$.)    
\end{definition}

\begin{lemma}[Eventual recovery of Algorithm~\ref{alg:cas}'s variation in Figure~\ref{fig:boundCAS}]
\label{thm:boundedRecovery}
Let $R$ be a fair execution of the bounded variation of Algorithm~\ref{alg:cas} (Figure~\ref{fig:boundCAS}). Suppose that in $R$'s starting system state, $c$, it holds that there is a node $p_i \in \sP$ (that hosts a server) stores a record $(t,\bullet) \in S_{p_i}$ with tag $t \geq t_{top}$ (but $R$ is not necessarily an execution with a legitimate overflow record). Within $\bigO(\Psi)$ asynchronous cycles, $R$ reaches a system state that is reset- and overflow-free. 
\end{lemma} 

\begin{proof}
Recall that the reset procedure has a termination period within $\Psi$ asynchronous cycles (Section~\ref{ss:BBBreq}).
Thus, within $\Psi$ asynchronous cycles, the system reaches a state $c' \in R$ that is reset-free. Note that if $c'$ is also overflow-free, the proof is done. Therefore, we consider the complementary case and assume that $c$ is reset-free but not necessarily overflow-free, i.e., $(t,\bullet) \in S_{p_i}$ with tag $t \geq t_{top}$. Claim~\ref{thm:convergenceMaxTagE} shows that within $\bigO(1)$ asynchronous cycles, the overflow handling proceeds to the invocation of the reset procedure (item~\ref{ln:restart} of Figure~\ref{fig:boundCAS}), which in turn brings the system to a reset- and overflow-free state within $\Psi$ asynchronous cycles. Therefore, the proof is done, because we showed that within $\bigO(\Psi)$ asynchronous cycles, the system reaches a state in $R$  that is both reset- and overflow-free. 

\begin{claim}
\label{thm:convergenceMaxTagE}
Within $\bigO(1)$ asynchronous cycles, the system reaches a state, $c'$, in which Condition~$(\ref{eq:ready})$ holds, where $t, t', t'' \in \sT: t \leq t' \land t' \geq t''$.
\begin{equation}
\label{eq:ready}
\begin{split}
\exists p_i \in \sP : (maxPhase_i(\sD)=t'\geq t_{top}) ~\land ~\\ (\forall p_k \in \sP : gossip_i[k] = tagTuple_i()) ~\land ~\\ (tagTuple_i() = (t',t'',t''))
\end{split}
\end{equation}
\end{claim}

\begin{claimproof}
Suppose that this claim is false and $R$ includes a prefix $R'$ with more than $\bigO(1)$ asynchronous cycles in which Condition~$(\ref{eq:ready})$ does not hold 
in every $c'' \in R'$. 

\noindent \textbf{We show that $\mathbf{\exists c_{stop} \in R':\forall p_j \in \sP:(t,\bullet) \in S_{p_j}:t \geq t_{top}}$.~~}
Note that, within $\bigO(1)$ asynchronous cycles, the gossip protocol works correctly (Part (3) of Corollary~\ref{thm:basicCorollary}). Moreover, the function $tagTuple()$ returns $(maxPhase(\sD)$, $maxPhase(\sD \setminus \{\text{`pre'}\}), maxPhase(\{ \text{`FIN'}\}))$ (line~\ref{alg:srv:gossip}) and this triple is sent by the gossip service. This claim assumes that $(t,\bullet) \in S_{p_i}$, which implies that within $\bigO(1)$ asynchronous cycles of $R'$, the system reaches a system state $c_{stop} \in R'$ for which $\forall p_j \in \sP:(t,\bullet) \in S_{p_j}:t \geq t_{top}$ holds (Lemma~\ref{thm:eventuallyAscendingSent}).

\noindent \textbf{We show that there is no step that follows immediately after $\mathbf{c_{stop}}$ in which any server responds to a query request of a write operation.~~}
No server responds to a query request due to the fact that $\forall p_j \in \sP:(t,\bullet) \in S_{p_j}$ in ${c_{stop}}$ and item~\ref{ln:stop} of Figure~\ref{fig:boundCAS}.

\noindent \textbf{We show that Condition~$(\ref{eq:ready})$ holds in $\mathbf{c'' \in R'}$.~~}
Every write operation that has started before $c_{stop}$ terminates eventually (Theorem~\ref{thm:liveness} with respect to Part (1) of Definition~\ref{def:liveness}) or they cannot proceed beyond the pre-write phase. (This is because $R$ is a fair execution and each write operation occurs within a constant number of phases and gossip rounds, we note that termination occurs within $\bigO(1)$ asynchronous cycles, because each phase occurs within $\bigO(1)$ asynchronous cycles, as we show in Part (3) of Corollary~\ref{thm:basicCorollary}.) Let $c'''$ be the first system state in which all of these write operations have terminated (or have stopped forever to proceed beyond the pre-write phase). Let $t'=\max_{p_j \in \sP} maxPhase_j(\sD):t'\geq t \geq t_{top}$ and $t''=\max_{p_k \in \sP} maxPhase_k(\sD \setminus \{\text{`pre'}\})$ in $c'''$. Recall that the server at $p_j$ gossips $(t', \bullet)$ and the server at $p_k$ gossips $(\bullet, t'', \bullet)$. Lemma~\ref{thm:eventuallyAscendingSent} implies that within $\bigO(1)$ asynchronous cycles in $R'$, the system reaches a state $c''''\in R'$ in which $\forall p_\ell \in \sP: tagTuple_\ell()=(t',t'',\bullet)$. By line~\ref{alg:srv:FINMax}, we have that Condition~$(\ref{eq:ready})$ holds in $c''''$ and so does this claim, because we have reached a contradiction with the assumption at the beginning of this proof.
\end{claimproof}
\end{proof}

\begin{lemma}
\label{thm:bounded}
Let $R$ be a fair execution of the bounded variation of Algorithm~\ref{alg:cas} (Figure~\ref{fig:boundCAS}) with a legitimate overflow record. (1) Within $\bigO(1)$ asynchronous cycles, the system reaches the first system state $c \in R$ in which it holds that there is a node $p_i \in \sP$ (that hosts a server that) stores a record $(t,\bullet) \in S_{p_i}$ with tag $t \geq t_{top}$. Also, we can write $R=R' \circ R_{sameTagTuple} \circ R''$, such that (2) within a prefix $R'$ of $\bigO(1)$ asynchronous cycles, the system reaches an unbounded suffix, $R_{sameTagTuple} \circ R''$, that has a prefix $R_{sameTagTuple}$ of $\bigO(1)$ asynchronous cycles, such that Condition~$(\ref{eq:ready})$ holds in its starting system state, $c' \in R_{sameTagTuple}$. Moreover, (3) only then at least one node calls $\mathit{globalReset}(t'')$, all nodes participate in that procedure and within $\bigO(\Psi)$ asynchronous cycles they resume, which leads to the end of $R_{sameTagTuple}$. Furthermore, (4) suppose that in $c' \in R_{sameTagTuple}$ it holds that $\exists p_j \in \sP : \{(t''=(z,k),\bullet, d) : d \in \sD \setminus \{\text{`pre'}\} \} \subseteq S_{p_j} $, where $t''$ is the tag value taken from Condition~$(\ref{eq:ready})$. Then, there is system state $c'' \in R_{sameTagTuple}$ that follows $c''$ and in which $\forall p_\ell \in \sP : S_{p_\ell} = \{((1,k),\bullet, \text{`FIN'})\}$. Otherwise, (5) $t''=t_0$ (line~\ref{alg:srv:default}) and $\forall p_\ell \in \sP : S_{p_\ell} = \emptyset$ in $c'$ and $c''$. 
\end{lemma}

\begin{proof}
\noindent \textbf{Part (1).~~} Lemma~\ref{thm:progression} implies this part of the proof.

\noindent \textbf{Part (2).~~} Claim~\ref{thm:convergenceMaxTagE} implies this part.

\noindent \textbf{Part (3).~~} Note that, by this lemma's assumption that $R$ has a legitimate overflow record, the starting system state of $R$ includes no tag that is greater or equal to $t_{top}$. Moreover, there is a node $p_i \in \sP$ for which it holds that $\forall p_k \in \sP : gossip_i[k] = (t',t'',t'')$ in $c'$, where $t'\geq t_{top} \land t' \geq t''$, because of Part (2) of this proof which implies that Condition~$(\ref{eq:ready})$ holds in $c'$. Therefore, the only way in which $gossip_i[k] = (t',t'',t'')$ can hold in $c'$, is if $tagTuple_k()=(t',t'',t'')$ holds in some system state that appears in $R$ before (and perhaps also after) $c'$ (and then these tags are gossiped from $p_k$ to $p_i$), because these tag values do not appear in $R$'s starting state. Moreover, at least one node calls $\mathit{globalReset}(t'')$ (due to item~\ref{ln:restart} of Figure~\ref{fig:boundCAS} and the fact that Condition~$(\ref{eq:ready})$ holds in $c'$). Therefore, all nodes resume within $\bigO(\Psi)$ asynchronous cycles (Section~\ref{ss:BBBreq}), which leads to the end of $R_{sameTagTuple}$.

\noindent \textbf{Part (4).~~} 
Claim~\ref{thm:noGlobalIntersect} considers the case in which more than one node calls $\mathit{globalReset}(t'')$ (item~\ref{ln:restart} of Figure~\ref{fig:boundCAS}) and implies that all such calls in $R_{sameTagTuple}$ refers to the same FINALIZED tag  $t''=(\bullet,k)$. This part of the proof is implied by the fact that a call to $\mathit{globalReset}(t'')$ indeed replaces $t''$ by $(1,k)$, cf. item~\ref{ln:restart} of Figure~\ref{fig:boundCAS}.

\begin{claim}
\label{thm:noGlobalIntersect}
Suppose that $R_{sameTagTuple}$ includes two steps, $a_i$ and $a_j$, in which $p_i$ and $p_j$ call $\mathit{globalReset}_i(t_i)$, and respectively, $\mathit{globalReset}_j(t_j)$. It is true that $t_i=t_j$. 
\end{claim}
 
\begin{proof}
We prove this claim by assuming that $t_i\neq t_j$ and then demonstrating a contradiction. Suppose, without the loss of generality, that $a_i$ appears in $R_{sameTagTuple}$ before $a_j$. Let us start the proof by assuming that $t_i < t_j$ before considering the complementary case of $t_i > t_j$. We show that neither case is possible and thus this claim is correct.

\noindent \textbf{The case of $\mathbf{t_i < t_j}$.~~} 
By Part (2) of the proof of this lemma, we know that $p_i$ calls $\mathit{globalReset}_i(t_i)$ in step $a_i$ only after $p_j$ has seen that the tag $t_i$ is a FINALIZED record, because Condition~$(\ref{eq:ready})$ must hold with respect to $t_i$ (item~\ref{ln:restart} of Figure~\ref{fig:boundCAS}) and the definition of $tagTuple()$ (line~\ref{alg:srv:gossip}). This is true starting from some system state that appears in $R$ before the steps $a_i$ and $a_j$. 
When $p_i$ takes step $a_i$ and calls $\mathit{globalReset}_i(t_i)$, it is true that $p_i$ has not seen $t_j$ in a finalized (or FINALIZED) record, due to this case assumption that $t_i < t_j$ and the fact that Condition~$(\ref{eq:ready})$ holds with respect to tag $t_i$ in the system state that immediately precedes $a_i$ (item~\ref{ln:restart} of Figure~\ref{fig:boundCAS}). 
Once $p_i$ takes step $a_i$ and calls $\mathit{globalReset}_i(t_i)$, the function $\mathit{globalReset}()$ disables $p_i$'s server and therefore $p_i$'s server does not receive or send in $R_{sameTagTuple}$ gossip messages after $a_i$. Therefore, $p_i$ does not receive the tag $t_j$ (in a finalized or FINALIZED record) in any step that follows $a_i$ in $R_{sameTagTuple}$. Moreover, the fact that $p_i$ does not gossip after $a_i$ implies that $p_j$ cannot receive from $p_i$ a gossip message with $(\bullet,t_j,t_j)$ (lines~\ref{alg:srv:uponGossip} and~\ref{alg:srv:FINMaxGossipHelper}). Thus, $gossip_j[i]=(\bullet,t_j,t_j)$ does not hold in any system state in $R_{sameTagTuple}$ that follows $a_i$. This is in contradiction to the assumption that $p_j$ takes step $a_j \in R_{sameTagTuple}$, because this step requires Condition~$(\ref{eq:ready})$ to hold with respect to $t_j$ (item~\ref{ln:restart} of Figure~\ref{fig:boundCAS}).

\noindent \textbf{The case of $\mathbf{t_i > t_j}$.~~} 
By Part (2) of the proof of this lemma, we know that $p_i$ calls $\mathit{globalReset}_i(t_i)$ in step $a_i$ only after $p_j$ has seen the tag $t_i$ is a FINALIZED record, because Condition~$(\ref{eq:ready})$ must hold with respect to $t_i$ (item~\ref{ln:restart} of Figure~\ref{fig:boundCAS}) and the definition of $tagTuple()$ (line~\ref{alg:srv:gossip}). This is true starting from some system state that appears in $R$ before the steps $a_i$ and $a_j$. The fact that, in the system state that immediately precedes step $a_j$ in which $p_j$ calls $\mathit{globalReset}_j(t_j)$, node $p_j$ has indeed seen $t_i$ in a finalized (or FINALIZED) record, demonstrates a contradiction due to this case assumption that $t_i > t_j$ and our assumption that $a_j$ appears in $R$ after $a_i$, because $p_j$ should select $t_j$ according to Condition~$(\ref{eq:ready})$ (item~\ref{ln:restart} of Figure~\ref{fig:boundCAS}).
\end{proof}

\noindent \textbf{Part (5).~~} 
This part refers to a case in which no server stores any finalized record. Thus, by the arguments above, there is at least one nodes that calls $\mathit{globalReset}_j(t_0)$. Since no server stores any record, the tag $t_0$ is used (line~\ref{alg:srv:default}), and this part of the proof flows simply from item~\ref{ln:restart} of Figure~\ref{fig:boundCAS}.
\end{proof}

\begin{theorem}[Bounded self-stabilizing $CAS(k)$ in the presence of seldom fairness]
\label{boundeVari}
Algorithm~\ref{alg:cas}'s variation in Figure~\ref{fig:boundCAS} is a bounded (message and state) size self-stabilizing algorithm (in the presence of seldom fairness) for implementing $T_{\text{CAS}(k)}$'s task. Both the recovery and the overflow periods end within $\bigO(\Psi)$ asynchronous cycles.
\end{theorem}

\begin{proof}
We demonstrate that the proposed algorithm is self-stabilizing (in the presence of seldom fairness). To that end, we show that (1) the proposed algorithm can always recover within $\bigO(\Psi)$ asynchronous cycles from an arbitrary starting system state of a fair execution, (2) during arbitrary executions that start from a legitimate system state, the system execution is legal, but (3) once in every $\bigO(z_{\max})$ write operations, the system stops providing liveness until the system execution becomes fair and then within $\bigO(\Psi)$ asynchronous cycles (during which safety is not violated) liveness is regained.

\noindent \textbf{Part (1).~~} 
Theorem~\ref{thm:liveness} (with respect to Part (1) of Definition~\ref{def:liveness}) demonstrates that the operations of Algorithms~\ref{alg:cas} always terminate. Lemma~\ref{thm:boundedRecovery} demonstrates that the added mechanisms for dealing with overflow events (Figure~\ref{fig:boundCAS}) always finish to deal with overflows and then the system simply follows Algorithms~\ref{alg:cas} for a period of at least $z_{\max}$ write operations. 

The proof of Theorem~\ref{thm:convergence} considers a complete write operation, $\pi_{write}$, that its tag is greater than any tag that is present throughout any earlier stage of the recovery process, i.e., including the set of tags that appeared in the starting system state. 
%
%
As long as no overflow event occurs, within $\bigO(1)$ asynchronous cycles, the system can complete the write operation, $\pi_{write}$, such that $T(\pi_{write})$'s record stays at the set $relevant(S_i)$ at least until a later write operation is completed. If an overflow handling is needed, recovery occurs within $\bigO(\Psi)$ asynchronous cycles (Lemma~\ref{thm:bounded}).


\noindent \textbf{Part (2).~~} 
Algorithm~\ref{alg:cas}'s correctness (theorems~\ref{thm:atomicity} and~\ref{thm:liveness}) implies this part.

\noindent \textbf{Part (3).~~} 
Lemma~\ref{thm:bounded} implies this part.
\end{proof}

\section{Cost Analysis}
%
%
\label{sec:cost}
The main complexity measures of self-stabilizing systems in the presence of seldom fairness are (Section~\ref{sec:timeComplexity}): (ii) the maximum length overflow period, which is of $\bigO(\Psi)$ asynchronous cycles (Theorem~\ref{boundeVari}) for the proposed solution, and (ii) the maximum length of the period during which the system recovers after the occurrence of transient failures, which is $\bigO(1)$ asynchronous cycles for the case of the unbounded solution (Theorem~\ref{thm:convergence}) but can also take $\bigO(\Psi)$ asynchronous cycles if the recovery period includes an overflow (or a recovery of the overflow mechanism).

Cadambe et al.~\cite{DBLP:journals/dc/CadambeLMM17} present a version of the $CAS(k)$ algorithm that includes elements of garbage collection that recycles merely the stored objects and never the meta-data, i.e., it never removes the records themselves, because the garbage collector removes only the coded elements and always keeps the tags and the phase indices. In the context of self-stabilizing systems, this implies that the storage is unbounded, because a single transient fault can clog the storage. Cadambe et al. also explicitly say that when the execution is unfair, an infinite storage is required~\cite[Table 1]{DBLP:journals/dc/CadambeLMM17}. One of the advantages of Algorithm~\ref{alg:cas}'s variation in Figure~\ref{fig:boundCAS} is that it offers bounded local storage of $\bigO(N+\delta)$ records also during periods in which the execution is unfair (Theorem~\ref{thm:implicitHereIsExplicitThereMaybe}), i.e., $\bigO( (\log_2 |\sV|) (N+\delta) )$ bits in total.
  
The proposed solution has write operations that include four phases rather than three, as in~\cite{DBLP:journals/dc/CadambeLMM17}. Cadambe et al.~\cite[Theorem 4]{DBLP:journals/dc/CadambeLMM17} analyze the communication costs of $CAS(k)$ and show that they can be made as small as $\frac{N}{N-f} \log_2 |\sV|$ bits by choosing $k=N-2f$. Moreover, Algorithm~\ref{alg:cas} and its variation in Figure~\ref{fig:boundCAS} consider gossip messages that include three tags, rather than just one as in Cadambe et al.~\cite{DBLP:journals/dc/CadambeLMM17}. When comparing the communication costs of a self-stabilizing algorithm to another that does not consider recovery from transient faults, we have to take into consideration that fact that self-stabilizing algorithms can never stop communicating (because then the system can be first brought to a state in which communication stops and then a transient-fault merely change some part of that state, which the algorithm cannot correct because it has stopped communicating, see~\cite[Chapter 3.2]{Dolev:2000} for details). Therefore, the proposed algorithm never stops sending gossip messages whereas the one by Cadambe et al.~\cite[Theorem 4]{DBLP:journals/dc/CadambeLMM17} sends $\bigO(N^2)$ gossip messages per client operation.

\section{Discussion}
\label{sec:disc}
We studied the implementation of private coded atomic storage protocol, which
is resilient to malicious servers. For the case of asynchronous
message-passing networks that provide fair communication, we proposed a self-stabilizing algorithm that preserves privacy and recovers after the occurrence of transient faults. Our solution requires the system to first reach a  fair execution before the algorithm guarantees recovery. Moreover, once in a practically infinite number of write operations, the proposed solution again requires fair execution to the end of dealing with counter overflows. Since overflow events of $64$-bit integers, in any practical settings, can only be the result of transient faults, and since transient faults are very rare, we believe that our novel stabilization criteria are applicable to a range of similar problems that require self-stabilizing tag schemes. Thus, as future work, we propose to the study of self-stabilizing (in the presence of seldom fairness) of consensus~\cite{DBLP:conf/netys/BlanchardDBD14,DBLP:journals/jcss/DolevKS10}, virtual synchrony~\cite{DBLP:conf/sss/DolevGMS15} and other  shared register emulation schemes~\cite{DBLP:journals/corr/abs-1805-03691}, to name a few.   

\subsection{Extension: recyclable client identifiers}
\label{sec:incarnationNumber}
In Section~\ref{sec:benignFailures}, we assume that clients that fail-stop and never return to take steps. We present here an elegant extension that is based on well-known techniques. This extension allows the nodes to recycle their client identifiers whenever they resume operation after failing. That is, we tolerate detectable restarts of the client nodes. For dealing with detectable restarts of the server nodes, we point out the existence of self-stabilizing quorum reconfiguration~\cite{DBLP:journals/corr/DolevGMS16,DBLP:conf/netys/DolevGMS17,DBLP:conf/middleware/DolevGMS16}.  
  
The client identifier could be a pair that includes the identifier of node $p_i$ and an incarnation number that is incremented whenever a failing node resumes and then wishes to invoke client operations. A well-known technique for maintaining a persistent incarnation number (without assuming access to a stable storage or that the storage is not prone to corruption by transient-faults) is to let a quorum service to emulate a shared counter, similar to the one in~\cite[Section 3.3]{DBLP:journals/corr/DolevGMS15}. Namely, $p_i$ queries all servers for its current maximal incarnation number and waits for a quorum of replies. Then, the client sends to all servers its updated incarnation number, which is the maximum in all query replies plus one and waits for a quorum of replies before invoking the next client operation. Note that each node $p_i$ that hosts a client has to maintain its own incarnation counter (by using the client identifier for partitioning the space of incarnation numbers) and that the above procedure is executed only when $p_i$ resumes after failing. 

We end the description of this extension by saying that the bounded variation of Algorithm~\ref{alg:cas} (Section~\ref{fig:boundCAS}) includes a global reset procedure that resets all clients and server. This reset procedure can also reset the above mechanisms for recyclable client identifiers. In addition, whenever the incarnation number reaches its maximum value, the global reset procedure is triggered. Note that the latter happens only after the occurrence of a transient fault. Thus, the client never runs out of incarnation numbers (in any participial settings). The description of the above procedure would not be complete without saying that the servers gossip periodically the set of all known pairs of client identifiers and their incarnation numbers. When such a set is received, the server updates its local set by adding all the pairs that come from clients that it does not know and updating all the existing pairs with the highest incarnation number. Note that the size of these sets is bounded by the number of possible clients, $N$.

\subsection{Conclusions}
We view the criteria of self-stabilizing algorithms (in the presence of seldom fairness) as an attractive alternative to both Dijkstra's self-stabilization criterion~\cite{DBLP:journals/cacm/Dijkstra74}, which for the sake of bounded recovery period usually models fail-stop failure as transient faults~\cite{Dolev:2000}, and the less restrictive criteria of pseudo-self-stabilizing~\cite{DBLP:conf/opodis/DolevDPT12,DBLP:journals/dc/BurnsGM93} and practically-self-stabilizing systems~\cite{DBLP:journals/jcss/DolevKS10,DBLP:conf/sss/DolevGMS15,DBLP:conf/netys/BlanchardDBD14}, which do not model fail-stop failure as transient faults, but do not offer a bounded recovery period. We consider self-stabilizing systems (in the presence of seldom fairness) to be (i) wait-free (since they do not assume execution fairness) in the absence of transient faults, and (ii) offering a bounded recovery period from transient faults, as in Dijkstra's self-stabilization criterion~\cite{DBLP:journals/cacm/Dijkstra74}. This is offered at the expense of compromising liveness (without jeopardizing safety) for a bounded period that occurs once in every practically infinite number of operations, of say, $2^{64}$.    
 
\paragraph*{Acknowledgments:}
We thank Robert Gustafsson and Andreas Lindh\'{e} for useful discussions and for helping to improve the presentation significantly.

\part*{Appendix}

\section{Self-stabilizing Gossip and Quorum Services}
\label{sec:basic:App}
Algorithm~\ref{alg:comm} provides an implementation that satisfies the requirements of Definition~\ref{def:coomSafe}. We start the algorithm description by refining the model with respect to the variables that Algorithm~\ref{alg:comm} uses as well as its interfaces. We then detail the way in which Algorithm~\ref{alg:comm} provides the requirements that appear in Definition~\ref{def:coomSafe}. 

\subsection{Refined model}
We assume that the system has access to a self-stabilizing end-to-end (reliable FIFO) message delivery protocol~\cite{DBLP:journals/ipl/DolevDPT11,DBLP:conf/sss/DolevHSS12} (over unreliable non-FIFO channels that are subject to packet omissions, reordering, and duplication). The self-stabilizing algorithms in~\cite{DBLP:journals/ipl/DolevDPT11,DBLP:conf/sss/DolevHSS12} circulate a token between any pair of senders and receivers. Our pseudo-code interfaces that protocol via events for token departure and arrival. Moreover, the self-stabilizing algorithms in~\cite{DBLP:journals/ipl/DolevDPT11,DBLP:conf/sss/DolevHSS12} guarantees that the receiver raises eventually a token arrival event with a token that the sender had transmitted upon its previous token departure event; exactly one token exists at any time that succeeds the recovery period. 

\subsubsection{Variables}
The gossip servers require two kinds of buffers (line~\ref{alg:comm:gossipDef}): one for the messages that go out ($gossipTx$) and another for the ones that come in ($gossipRx[]$). Node $p_i$ stores in $gossipTx$ the needed to transmit messages and in $gossipRx[j]:j\neq i$ the ones that are arriving to $p_i$ from $p_j$'s gossip. We use $gossipRx[i]$ for aggregating these received values.   

Four kinds of buffers facilitate the quorum-based communications, because this communication pattern has a round-trip nature that a client initiates and includes a number of nodes (line~\ref{alg:comm:pingpong}). The initiating client write the message to $pingTx$. The end-to-end communication protocol~\cite{DBLP:journals/ipl/DolevDPT11,DBLP:conf/sss/DolevHSS12} transfers that message to the server-side and stores it in $pongTx[]$. The server processes the arriving messages and stores its reply in $pongRx[]$ so that the end-to-end protocol could transfer this reply back to the client-side, which stores it in $pingRx[]$. When any of this buffers does not store a message, the $\bot$ symbol is used. 
Note that the client requests take the form of $(tag, word, phase) \in F=(\sT \cup \{\bot\})  \times (\sW \cup \{\bot \})\times (\sD \cup \{\text{`qry'}\})$, where $phase =\text{`qry'}$ whenever the clients sends a query (rather than another phase that does appear in $\sD$). The reply has the form of $( ping, pong ) \in F\cup \{\bot \} \times F  \cup\{\bot \}$.
We use the variable $aggregated$ for letting the client to aggregate the server responses for the latest request (line~\ref{alg:comm:aggregated}).

\subsubsection{Interface that Algorithm~\ref{alg:comm} assumes to be available}
Every pair of nodes maintains a pair of self-stabilizing communication channels~\cite{DBLP:journals/ipl/DolevDPT11,DBLP:conf/sss/DolevHSS12}, i.e., one channel in each direction, which circulates a token between the sender and the receiver. The code interfaces that protocol via events for token departure (lines~\ref{alg:comm:gossipServerDeparture},~\ref{alg:comm:pingpongClientDeparture} and~\ref{alg:comm:pingpongServerDeparture}) and arrival (lines~\ref{alg:comm:gossipClientArrival},~\ref{alg:comm:pingpongServerArrival} and~\ref{alg:comm:pingpongClientArrival}). Moreover, this protocol guarantees that the receiver raises eventually a token arrival event with a token that the sender had transmitted upon its previous token departure event. Exactly one token exists at any time. In the code, token arrival and departure events raise the respective events according to the message type, which are gossip and quorum. The event handlers to these events perform local operations, release the token by calling ${send}(receiverID, Payload)$ as well as raising other local events on the calling node.
Algorithm~\ref{alg:comm} also assumes access to a primitive, which we call $suspend()$. When an event calls $suspend(var, const)$, the system suspends the calling event until $var=const$ (while allowing other non-suspended events on that node to run).

\subsection{The details of Algorithm~\ref{alg:comm}}
The gossip functionality simply sends the message in $gossipTx$ whenever the token arrives to the client (line~\ref{alg:comm:gossipServerDeparture}), stores it in $gossipRx[j]$ whenever the token arrives to the server (line~\ref{alg:comm:gossipClientArrival}) and raises the gossip arrival event at the server-side with all the most recently arrived gossip messages from each server (line~\ref{alg:comm:gossipClientArrivalDeliver}). Note that when a gossip message arrives, the receiving server raises an event with $\{gossip[k]\}_{p_k \in \sP}$, which includes the most recently received messages. We allow a simpler presentation of Algorithm~\ref{alg:cas} by letting the server at node $p_i$ to use the item $gossip[i]$ for aggregating the gossip information that it later gossips to all other servers.

The ping-pong protocol is inspired by the Communicate protocol proposed by Attiya at el.~\citep[Section 3]{Attiya:1995}. The client sends its message to the server upon token departure (line~\ref{alg:comm:pingpongClientDeparture}), update the server buffer upon token arrival (line~\ref{alg:comm:pingpongServerDepartureNotBot}), and then the server sends its reply upon its token departure (line~\ref{alg:comm:pingpongServerDepartureNotBot}) before the token arrival event allows the client to accumulate the server replies (line~\ref{alg:comm:pingpongServerArrival}). The client then tests whether it has accumulated replies from a quorum of servers (line~\ref{alg:comm:sQ}). If this is the case, it lets the calling client procedure to receive the accumulated replies (line~\ref{alg:comm:pingpongClientArrivalReturn}).

Note that, for the sake of simple presentation of Algorithm~\ref{alg:cas}, we allow the client to call $\text{qrmAccess}(msg)$ either when $msg$ is a single message to be sent to all servers or when $msg$ is a vector that includes an individual message for each server (line~\ref{alg:comm:pingpongClientDeparture} and line~\ref{alg:comm:load}).

\begin{algorithm*}[t!]
\algoSize

\noindent \textbf{Variables:}
$gossipTx$ and $gossipRx[]$ are buffers, where $p_i$'s server stores in $gossipTx$ the needed to transmit message. In $gossipRx[j]:j\neq i$, $p_i$'s server stores the most recently received $p_j$'s gossip and allowing the use of $gossipRx[i]$ for aggregating these received values.\label{alg:comm:gossipDef}   

$pingTx$, $pingRx[]$, $pongTx[]$ and $pongRx[]$ are buffers, where $p_i$'s client stores in $pingTx$ its request, or $\bot$. In $pingRx[j]$, $p_i$'s server stores the most recently received $p_j$'s request, or $\bot$. In $pongTx[j]$, $p_i$'s server stores the reply to $p_j$, or $\bot$. In $pongRx[j]$, $p_i$'s client stores the most recently received $p_j$'s acknowledgment, or $\bot$. The client request and $pingTx$ has the form $(tag, word$, $phase) \in F=(\sT \cup \{\bot\})  \times (\sW \cup \{\bot \})\times (\sD \cup \{\text{`qry'}\})$. The reply form is $( ping, pong ) \in F\cup \{\bot \} \times F  \cup\{\bot \} $.\label{alg:comm:pingpong}

$aggregated$ is a variable in which $p_i$'s client stores the collected server responses for the latest request\label{alg:comm:aggregated}\;

\BlankLine

\noindent \textbf{Interface in use:}
\label{alg:comm:interfacesNeeded}
Every pair of nodes maintain self-stabilizing communication channel~\cite{DBLP:journals/ipl/DolevDPT11,DBLP:conf/sss/DolevHSS12}, which circulates a token between the sender and the receiver. Token arrival and departure events raises the respective events according to channel type, which are gossip and quorum.\;

$suspend(var, const)$ suspends the calling event until $var=const$ (while allowing other events to run)\;

\BlankLine\BlankLine

\noindent \textbf{Interface provided:}\\
\label{alg:comm:interfacesProvided}
\textbf{function} $\text{gossip}(msg)$ \lDo{~$gossipTx \gets msg$;\label{alg:comm:gossip}}

\textbf{function} $\text{qrmAccess}(msg)$ \lDo{ $phaseInit(msg)$; $\textbf{return}({wait}())$;\label{alg:comm:quorumAccess}}

\textbf{function} $\text{reply}(j,m)$ \lDo{ \{\textbf{if} $pingRx[j].phase = \text{`qry'}$ \textbf{then} $pongTx[j] \gets (m.tag, \bot, \text{`qry'})$ \textbf{else} $pongTx[j] \gets$ $(pingRx[j].tag, m.word, pingRx[j].phase)$\}\label{alg:comm:acknowledge}}

\BlankLine

\noindent \textbf{Event handlers and local functions:}\\
\label{alg:comm:function}
\textbf{upon} gossip token departure from $p_i$'s server to $p_j$'s server\label{alg:comm:gossipServerDeparture} \lDo{~\textbf{send}$(j, gossipTx)$;}

\textbf{upon} gossip token $m$ arrival from $p_j$'s server to $p_i$'s server\label{alg:comm:gossipClientArrival} \Do{
     $gossipRx[j] \gets m$; \textbf{raise} the gossip arrival event with the message $\{gossipRx[k]\}_{p_k \in \sP}$\label{alg:comm:gossipClientArrivalDeliver};   
}

\textbf{upon} pingpong token departure from $p_i$'s client to $p_j$'s server\label{alg:comm:pingpongClientDeparture} \lDo{~\textbf{send}$(j,load(j,pingTx))$;}

\textbf{upon} pingpong token departure from $p_i$'s server to $p_j$'s client\label{alg:comm:pingpongServerDeparture} \Do{
\lIf(~\textbf{/*} \texttt{ignore the server during channel resets} \textbf{*/}){$pingRx[j]=\bot$}{\textbf{send}$(j, (\bot, \bot))$}
\lIf{$pingRx[j]\neq\bot$}{\textbf{send}$(j,(pingRx[j], pongTx[j]))$\label{alg:comm:pingpongServerDepartureNotBot}}
}

\textbf{upon} $pingpong$ token arrival from $p_j$'s client to $p_i$'s server\label{alg:comm:pingpongServerArrival} \Do{
     $pingRx[j] \gets pingpong$\tcc*{$pingpong$ is either $\bot$ or $msg$}
     \If(~\textbf{/*} \texttt{don't interrupt the server during channel resets} \textbf{*/} ){$pingRx[j] \neq \bot$}{\textbf{raise} the event of  $pingRx[j].phase$ arrival from $p_j$'s server with $pingRx[j].msg$ (unless it is $\bot$)\label{alg:comm:pingpongServerArrivalDeliver}}
}

\textbf{upon} $pingpong=(ping, pong)$ token arrival from $p_j$'s server to $p_i$'s client \label{alg:comm:pingpongClientArrival} \Do{
     \lIf{$load(j,pingTx) = ping \land (pong = \bot \lor pong.tag = \bot \lor ((ping.phase \neq \text{`qry'} ) \implies (ping.tag =$ $pong.tag)))$\label{alg:comm:pingTxPing}}{$pongRx[j] \gets pong$}      

\If(~ \textbf{/*} \texttt{test for a quorum of acknowledgments} \textbf{*/}){$\{p_j:pongRx[j] \neq \bot\} \in \sQ$\label{alg:comm:sQ}}{

$aggregated \gets \{pongRx[j].word : pongRx[j] \neq \bot \}$\;

$clear()$\tcc*{returns from $ping.phase$ with $aggregated$ as the acknowledgment set}\label{alg:comm:pingpongClientArrivalReturn}}      
}

\textbf{function} $clear()$ \lDo{ \textbf{begin} \lForEach{$p_k \in \sP$}{$pongRx[k] \gets \bot$ \textbf{end}; $pingTx \gets \bot$}\label{alg:comm:clear}}

\textbf{function} $phaseInit(m)$ \lDo{ \textbf{begin} \lForEach{$p_k \in \sP$}{$pongRx[k] \gets \bot$ \textbf{end}; $pingTx \gets m$}\label{alg:comm:phaseInit}}

\textbf{function} ${wait}()$ \lDo{ $suspend(pingTx, \bot)$; \textbf{let} $x=aggregated$; $aggregated\gets \emptyset$;  \textbf{return}$(x)$;\label{alg:comm:wait}}

\textbf{function} ${load}(j,m)$ \lDo{~\textbf{if $m=((\bullet, w_1, \bullet), \ldots, (\bullet, w_N, \bullet))$ then return $(\bullet, w_j, \bullet)$ else return $m$;}\label{alg:comm:load}}
 
\caption{\algoSize Self-stabilizing gossip and quorum-based communication, code for node $p_i$.}\label{alg:comm}
\end{algorithm*}

\subsection{Correctness of Algorithm~\ref{alg:comm}}
Theorem~\ref{thm:basic} shows that the system always reaches a legal execution (Definition~\ref{def:coomSafe}) and uses Remark~\ref{thm:tokenCyrc}. 

\begin{remark}
\label{thm:tokenCyrc}
Recall that we assume that use of a self-stabilizing communication channel between any pair of nodes, such as in~\cite{DBLP:journals/ipl/DolevDPT11,DBLP:conf/sss/DolevHSS12}, for guaranteeing reliable end-to-end message delivery. In~\cite{DBLP:journals/ipl/DolevDPT11,DBLP:conf/sss/DolevHSS12}, the receiver raises eventually a token arrival event with a message that the sender had transmitted upon its previous token departure event, such that, at any time, there is exactly one token carrying one message. And, the token traversal direction alternates, i.e., go from one peer to another, then back, then go again and so on.
\end{remark}

Theorem~\ref{thm:basic} provides a proof for Corollary~\ref{thm:basicCorollary}.
 
\begin{theorem}[Self-stabilizing gossip and quorum-based communications]
\label{thm:basic}
Let $R$ be an Algorithm~\ref{alg:comm}'s (unbounded) execution that satisfies the terms of service of the quorum-based communication functionality. Suppose that $R$ is fair and its starting system state is arbitrary. Within $\bigO(1)$ asynchronous cycles, $R$ reaches a suffix $R'$ in which 
\textbf{\emph{(1)}} the gossip, and 
\textbf{\emph{(2)}} the quorum-based communication functionalities are correct.
\textbf{\emph{(3)}} During $R'$, the gossip and quorum-based communication complete correctly their operations within $\bigO(1)$ asynchronous cycles.
\end{theorem}

\begin{proof}
\noindent \textbf{Part (1).~~} 
Let $a_{\text{gossip},\ell} \in R$ be a step in which the server at $p_j \in \sP$ calls $\text{gossip}(m_\ell)$ for the $\ell$-th time in $R$ (line~\ref{alg:comm:gossip}). Note that $a_{\text{gossip},\ell}$ copies $m$ to $gossipTx_i$ (line~\ref{alg:comm:gossip}). Let $a_{depart,\ell'} \in R$ be the first step in $R$ that appears after $a_{\text{gossip},\ell'}:\ell'\in \{1,\ldots\}$ and before $a_{\text{gossip},{\ell'+1}}$, if there is any such step, in which $p_j$ executes the event of gossip token departure (line~\ref{alg:comm:gossipServerDeparture}). Note that $a_{\text{depart},{\ell'}}$ transmits to the server at $p_k \in \sP$ the token $m_{\ell'}$, where $m_{\ell'}=gossipTx_i[i]$ in any system state that is between $a_{\text{gossip},\ell'}$ and $a_{depart,\ell'}$. That token arrives eventually to the server at $p_k$, which raises the respective event (line~\ref{alg:comm:gossipClientArrival}) and then the event of gossip arrival with the message $m_{\ell'}$ (line~\ref{alg:comm:gossipClientArrivalDeliver}). Thus, the gossip functionality is correct (Definition~\ref{def:coomSafe}), because (1) $m_{\ell'}$ was indeed sent by $p_j$, and (2) $p_k$ delivers such message infinitely often.\\

\noindent \textbf{\emph{Correct behavior of the gossip functionality.~~}} Suppose that (1) every message that $p_k$ delivers to the upper layer as a gossip from $p_j$ (line~\ref{alg:comm:gossipClientArrival}) was indeed sent by $p_j$ earlier in $R$ (line~\ref{alg:comm:gossip}). Moreover, (2) such deliveries occur infinitely often in $R$. Furthermore, (3) at any time, the communication channel from $p_j$ to $p_k$ does not include a message that $p_j$ has never sent. In this case, we say that the behavior of the gossip functionality from the server at $p_j$ to the one at $p_k$ is correct.

\noindent \textbf{Part (2).~~} 
Claims~\ref{thm:quorumReceive} and~\ref{thm:quorumReceiveMore} imply the proof of this part.

\begin{claim}
\label{thm:quorumReceive}
Suppose that the client at $p_i$ sends a request, i.e., $p_i$ calls the function $\text{qrmAccess}(m)$ (line~\ref{alg:comm:quorumAccess}) in step $a_{qrmAccess} \in R$, where $m \neq \bot$. After $a_{qrmAccess}$, execution $R$ includes steps (i) to (v) (Definition~\ref{def:coomSafe}).
\end{claim}
\begin{claimproof}
In $a_{qrmAccess}$, node $p_i$ assigns $m_i \neq \bot$ to $pingTx_i$ (line~\ref{alg:comm:quorumAccess}) and then suspends the running client (while allowing other events to run concurrently) until $pingTx_i=\bot$, where $m_i \in F$. Recall that only the client calls the function $\text{qrmAccess}()$ (Algorithm~\ref{alg:cas}) and it does so sequentially (terms of service for the quorum-based communication functionality, Definition~\ref{def:coomSafe}). Thus, after $a_{qrmAccess}$, the only way in which $p_i$ assigns $\bot$ to $pingTx_i$ is by executing line~\ref{alg:comm:pingpongClientArrivalReturn}. Therefore, the invariant of $pingTx_i =m_i\neq \bot$ holds until the client at $p_i$ takes the step (v), which returns from the function that it had called in step $a_{depart}$ with $aggregated$ as the acknowledgment set. We show that after $a_{qrmAccess}$, the system takes the steps (i) to (v).

\noindent \textbf{Steps (i) and (ii).~~}
Let $a_{depart}$ be the first step in $R$ that appears after $a_{qrmAccess}$ and in which $p_i$ executes the event of ping-pong token departure from the client at $p_i$ to the server at $p_j$ (line~\ref{alg:comm:pingpongClientDeparture}). Note that $p_i$ transmits the message $m_i=pingTx_i$ in $a_{depart}$. That token arrives eventually to the server at $p_j$ (Remark~\ref{thm:tokenCyrc}), which raises the respective event at some step $a_{j} \in R$ (line~\ref{alg:comm:pingpongServerArrival}) and then the event of $pingRx[j].phase$ arrival (also in step $a_{j}$), which delivers the arriving token, $pingRx[j].msg$ (line~\ref{alg:comm:pingpongServerArrivalDeliver}), unless the latter is $\bot$. Let $\hat{Q}_{(i),(ii)}\subseteq P$ include the set of nodes, such as $p_j$, that their servers raise the latter two events (in step $a_{j}$). Note that, as long as the invariant of $pingTx_i =m_i\neq \bot$ holds, $\hat{Q}_{(i),(ii)}$ includes more and more nodes. Thus, $\hat{Q}_{(i),(ii)} \in \sQ$ eventually (Property (2) of Lemma~\ref{thm:quor}).

\noindent \textbf{Steps (iii).~~}
Recall that this lemma assumes that in $R$ the server at $p_j \in \sP$ acknowledges (by calling $\text{reply}(msg)$, line~\ref{alg:comm:acknowledge}) requests that $p_j$ delivers to it (terms of service for the quorum-based communication functionality, Definition~\ref{def:coomSafe}). Let $a_{j'} \in R$ refer to these steps and $\hat{Q}_{(iii)}\subseteq P$ include the set of nodes, such as $p_j$, that take these steps, $a_{j'}$. Note that, as long as the invariant of $pingTx_i =m_i\neq \bot$ holds, $\hat{Q}_{(iii)}$ includes more and more nodes. Thus, $\hat{Q}_{(i),(ii),(iii)} \in \sQ$ eventually (Property (2) of Lemma~\ref{thm:quor}). 

\noindent \textbf{Steps (iv) and (v).~~}
The server at $p_j$ eventually sends the token $(pingRx_j[i]$, $pongTx_j[i])$ (line~\ref{alg:comm:pingpongServerDepartureNotBot}), cf. Remark~\ref{thm:tokenCyrc}. Note that the sent token includes (in the ping field) $a_{qrmAccess}$'s request and (in the pong field) the same phase and tag of the arriving request (line~\ref{alg:comm:acknowledge}), if the sent request includes any such values. Moreover, that token arrives eventually from the server at $p_j$ to the client at $p_i$ (Remark~\ref{thm:tokenCyrc}) at some step $a_{j''} \in R$, which raises the respective event (line~\ref{alg:comm:pingpongClientArrival}). 
Recall that node $p_i$ cannot change the value of $pingTx_i \neq \bot$ before the if-statement condition of line~\ref{alg:comm:pingTxPing} holds.
By the fact that the arriving token includes the same phase and tag, if there is any, of the sent request $pingTx_i$, node $p_i$ collects the arriving acknowledgments until the if-statement condition of line~\ref{alg:comm:pingTxPing} holds eventually. 
Let $\hat{Q}_{(iv)}\subseteq P$ be the set of nodes, such as $p_j$, for which the $p_i$ takes the step $a_{j''}$. Note that, as long as the invariant of $pingTx_i =m_i\neq \bot$ holds, $\hat{Q}_{(iv)}$ includes more and more nodes. Thus, $\hat{Q}_{(iv)} \in \sQ$ eventually (Property (2) of Lemma~\ref{thm:quor}). This implies that $p_i$ also takes the step $a_{i} \in R$ in which $p_i$ lets the function (which $p_i$ had previously called in step $a_{qrmAccess}$) to return (line~\ref{alg:comm:pingpongClientArrivalReturn}) by calling $clear()$ (line~\ref{alg:comm:pingpongClientArrivalReturn}) and by that allowing the resume of the client and the return from the function that has sent the request. Only then does the invariant $pingTx_i \neq \bot$ stops from holding. We remind that $aggregated$ holds the set of server replies that were sent between $a_{qrmAccess}$ and $a_{i}$ as well as matched $p_i$'s request, used the same phase and tag, if there were any such values in $p_i$'s request.\end{claimproof}

\begin{claim}
\label{thm:quorumReceiveMore}
Suppose that the state of the client at $p_i$ includes a non-$\bot$ value in $pingTx_i$. Eventually, $pingTx_i=\bot$ (and the client at $p_i$ resumes, if it had been suspended).
\end{claim}
\begin{claimproof}
Suppose that the client at $p_i$ calls eventually the function $\text{qrmAccess}()$ (line~\ref{alg:comm:quorumAccess}). Then, by Claim~\ref{thm:quorumReceive} the proof of this claim is done. 
Suppose that, throughout $R$, the client at $p_i$ does not call the $\text{qrmAccess}()$ function and yet $pingTx_i\neq \bot$ (in $R$'s starting system state). We show that $pingTx_i = \bot$ eventually. 
Recall that $pingTx \neq \bot$ has the form of $(tag, \bullet, phase) \in F$, where $F=(\sT \cup \{\bot\})  \times \{\bullet\} \times (\sD \cup \{\text{`qry'}\})$ and the replies have the form of $(ping, pong) \in F\cup\{\bot \} \times F\cup\{\bot \} $ (line~\ref{alg:comm:pingpong}). Note that eventually when $p_j$ sends a token back to $p_i$ (line~\ref{alg:comm:pingpongServerDeparture}). That token is either $(\bot,\bot)$ or $(pin_j,pon_j)$, where $pin_j$ is the request $pingRx_j[i]$ that $p_j$ had received from $p_i$ and $pon_j$ is $(pingRx[j].tag, \bullet, pingRx[j].phase)$ (line~\ref{alg:comm:acknowledge}).
By the same arguments that appears in the proof of Claim~\ref{thm:quorumReceive}, this proof is done. Namely, as long as $pingTx \neq \bot$ we have $pin_j \neq \bot$ and we can apply `steps (iv) and (v)' in the proof of Claim~\ref{thm:quorumReceive}.\end{claimproof}

\noindent \textbf{Part (3).~~} 
By the proof of Part (1) of this lemma, we see that the correctness invariant of the gossip service holds within $\bigO(1)$ asynchronous cycles because it considers the propagation of a single message from $p_j$ to every $p_k \in \sP$, i.e., it requires a single complete server iteration (with round-trips). By the proof of Part (2) of this lemma, we see that the correctness invariant of the quorum-based communication service holds within $\bigO(1)$ asynchronous cycles because steps (i) and (v) consider the propagation of a single message round-trip from a client to a quorum of servers, i.e., it requires a single complete client round.
\end{proof}  

\bibliographystyle{plain}
\bibliography{bib/trBib}

\end{document}